\def\arXiv{}

\documentclass[11pt]{article}

\usepackage{acronym}
\usepackage{amsmath}
\usepackage{amsthm}
\usepackage{amssymb}
\usepackage{authblk}
\usepackage[american]{babel}
\usepackage{dsfont}
\usepackage[inline]{enumitem}
\usepackage[top=1in,left=1in,bottom=1in,right=1in]{geometry}
\usepackage[utf8x]{inputenc}
\usepackage{lineno}
\ifdefined \arXiv
	\usepackage[disable]{todonotes}
\else
	\usepackage{todonotes}
\fi
\usepackage[pdfborder={0 0 0}]{hyperref} 

\acrodef{APSP}{All-Pairs Shortest Paths}
	\acrodefindefinite{APSP}{an}{an}
\acrodef{APWP}{All-Pairs Widest Paths}
	\acrodefindefinite{APWP}{an}{an}
\acrodef{BFS}{Breadth-First Search}
	\acused{BFS}
\acrodef{DAG}{Directed Acyclic Graph}
\acrodef{FRT}{Fakcharoenphol, Rao, and Talwar~\cite{frt-tbaamtm-04}}
	\acrodefindefinite{FRT}{an}{a}
\acrodef{kDSDP}[\texorpdfstring{$k$}{k}-DSDP]{\texorpdfstring{$k$}{k}-Distinct-Shortest Distance Problem}
\acrodef{kSDP}[\texorpdfstring{$k$}{k}-SDP]{\texorpdfstring{$k$}{k}-Shortest Distance Problem}
\acrodef{kSSP}[\texorpdfstring{$k$}{k}-SSP]{\texorpdfstring{$k$}{k}-Source Shortest Paths}
\acrodef{kSWP}[\texorpdfstring{$k$}{k}-SWP]{\texorpdfstring{$k$}{k}-Source Widest Paths}
\acrodef{LE}{Least Element}
	\acrodefindefinite{LE}{an}{a}
\acrodef{MBF}{Moore-Bellman-Ford}
	\acrodefindefinite{MBF}{an}{a}
\acrodef{MBF-like}{Moore-Bellman-Ford-like}
	\acrodefindefinite{MBF-like}{an}{a}
\acrodef{MSSP}{Multi-Source Shortest Paths}
	\acrodefindefinite{MSSP}{an}{a}
\acrodef{MSWP}{Multi-Source Widest Paths}
	\acrodefindefinite{MSSP}{an}{a}
\acrodef{SLF}{Simple Linear Function}
	\acrodefindefinite{SLF}{an}{a}
\acrodef{SPD}{Shortest Path Diameter}
	\acrodefindefinite{SPD}{an}{a}
\acrodef{SSSP}{Single-Source Shortest Paths}
	\acrodefindefinite{SSSP}{an}{a}
\acrodef{SSWP}{Single-Source Widest Paths}
	\acrodefindefinite{SSWP}{an}{a}
\acrodef{WPP}{Widest Path Problem}

\numberwithin{equation}{section}

\newtheorem{theorem}{Theorem}[section]
\newtheorem{corollary}[theorem]{Corollary}
\newtheorem{definition}[theorem]{Definition}
\newtheorem{example}[theorem]{Example}
\newtheorem{lemma}[theorem]{Lemma}
\newtheorem{observation}[theorem]{Observation}
\newtheorem{remark}[theorem]{Remark}

\setlist[enumerate,1]{label=(\arabic*)}
\setlist[enumerate,2]{label=(\alph*),ref=(\arabic{enumi}\alph*)}

\ifdefined \arXiv
	\nolinenumbers
\else
	\linenumbers
\fi

\renewcommand{\epsilon}{\varepsilon}
\def\dash---{\kern.16667em---\penalty\exhyphenpenalty\hskip.16667em\relax}
\newcommand{\bigO}{\operatorname{O}}
\newcommand{\bigOT}{\operatorname{\tilde{O}}}
\newcommand{\bigOmega}{\operatorname{\Omega}}
\newcommand{\bigOmegaT}{\operatorname{\tilde{\Omega}}}
\newcommand{\bigTheta}{\operatorname{\Theta}}
\newcommand{\bigThetaT}{\operatorname{\tilde{\Theta}}}
\newcommand{\bigo}{\operatorname{o}}
\newcommand{\bigoT}{\operatorname{\tilde{o}}}

\newcommand{\dotcup}{\mathbin{\dot{\cup}}}
\newcommand{\E}{\operatorname{\mathds{E}}}
\newcommand{\N}{\mathds{N}}
\newcommand{\Prob}{\operatorname{\mathds{P}}}

\newcommand{\quotient}[2]{{#1}/_{#2}}
\newcommand{\R}{\mathds{R}}
\newcommand{\T}{\top}
\newcommand{\width}{\operatorname{width}}
\newcommand{\Z}{\mathds{Z}}

\newcommand{\diam}{\operatorname{D}}
\newcommand{\dist}{\operatorname{dist}}
\newcommand{\hop}{\operatorname{hop}}
\newcommand{\id}{\operatorname{id}}
\newcommand{\level}{\operatorname{\lambda}}
\newcommand{\mhspaths}{\operatorname{MHSP}}
\newcommand{\paths}{\operatorname{P}}
\newcommand{\poly}{\operatorname{poly}}
\newcommand{\polylog}{\operatorname{polylog}}
\newcommand{\proofinappendix}{\noindent Proof in Appendix~\ref{app:proofs}.\medskip}
\newcommand{\Rdist}{\R_{\geq 0} \cup \{ \infty \}}
\newcommand{\spd}{\operatorname{SPD}}
\newcommand{\weight}{\operatorname{\omega}}

\title{\bf\Large Parallel Metric Tree Embedding\protect\\ based on an Algebraic View on {M}oore-{B}ellman-{F}ord\footnote{%
	This work extends and subsumes the extended abstract that appeared in the
	\emph{Proceedings of the 28th ACM Symposium on Parallelism in Algorithms and Architectures (SPAA~2016),} pages 455--466, 2016~\cite{fl-pmteboaavombf-16}.}}

\author[1,2]{Stephan~Friedrichs}
\author[1]{Christoph~Lenzen}
\affil[1]{Max Planck Institute for Informatics, Saarbr\"ucken, Germany\\
	Email:~\texttt{\{sfriedri,clenzen\}@mpi-inf.mpg.de}}
\affil[2]{Saarbr\"ucken Graduate School of Computer Science}

\date{}

\begin{document}

\pagestyle{plain}
\setcounter{page}{0}

\maketitle

\thispagestyle{empty}

\begin{abstract}
	A \emph{metric tree embedding} of expected \emph{stretch~$\alpha \geq 1$} maps a weighted $n$-node graph $G = (V, E, \weight)$ to a weighted tree $T = (V_T, E_T , \weight_T)$ with $V \subseteq V_T$ such that, for all $v,w \in V$, $\dist(v, w, G) \leq \dist(v, w, T)$ and $\E[\dist(v, w, T)] \leq \alpha \dist(v, w, G)$.
	Such embeddings are highly useful for designing fast approximation algorithms, as many hard problems are easy to solve on tree instances.
	However, to date the best parallel $(\polylog n)$-depth algorithm that achieves an asymptotically optimal expected stretch of $\alpha \in \bigO(\log n)$ requires $\bigOmega(n^2)$ work and a metric as input.

	In this paper, we show how to achieve the same guarantees using $\polylog n$ depth and $\bigOT(m^{1+\epsilon})$ work, where $m = |E|$ and $\epsilon > 0$ is an arbitrarily small constant.
	Moreover, one may further reduce the work to $\bigOT(m + n^{1+\epsilon})$ at the expense of increasing the expected stretch to $\bigO(\epsilon^{-1} \log n)$.

	Our main tool in deriving these parallel algorithms is an algebraic characterization of a generalization of the classic \acl{MBF} algorithm.
	We consider this framework, which subsumes a variety of previous ``\acl{MBF-like}'' algorithms, to be of independent interest and discuss it in depth.
	In our tree embedding algorithm, we leverage it for providing efficient query access to an approximate metric that allows sampling the tree using $\polylog n$ depth and $\bigOT(m)$ work.

	We illustrate the generality and versatility of our techniques by various examples and a number of additional results.
	Specifically, we
	\begin{enumerate*}
	\item
		improve the state of the art for determining metric tree embeddings in the Congest model,

	\item
		determine a $(1 + \hat\epsilon)$-approximate metric regarding the distances in a graph $G$ in polylogarithmic depth and $\bigOT(nm^{1 + \epsilon})$ work, and

	\item
		improve upon the state of the art regarding the $k$-median and the the buy-at-bulk network design problems.
	\end{enumerate*}
\end{abstract}

\todo[inline]{Update ``to appear'' and arXiv bibtex entries like~\cite{en-hwchaatasp-16}}

\newpage

\section{Introduction}
\label{sec:introduction}

In many graph problems the objective is closely related to distances in the graph.
Prominent examples are shortest path problems, minimum weight spanning trees, a plethora of Steiner-type problems~\cite{Steiner-site}, the traveling salesman, finding a longest simple path, and many more.

If approximation is viable or mandatory, a successful strategy is to approximate the distance structure of the weighted graph $G$ by a simpler graph~$G'$, where ``simpler'' can mean fewer edges, smaller degrees, being from a specific family of graphs, or any other constraint making the considered problem easier to solve.
One then proceeds to solve a related instance of the problem on $G'$ and maps the solution back to~$G$, yielding an approximate solution to the original instance.
Naturally, this requires a mapping of bounded impact on the objective value.

A standard tool are \emph{metric embeddings,} mapping $G = (V, E, \weight)$ to $G' = (V', E', \weight')$, such that $V \subseteq V'$ and $\dist(v, w, G) \leq \dist(v, w, G') \leq \alpha \dist(v, w, G)$ for some $\alpha \geq 1$ referred to as \emph{stretch}.\footnote{%
	$\dist(\cdot, \cdot, G)$ denotes the distance in $G$ and defines a metric space.
	See definitions in Section~\ref{sec:notation}.%
} An especially convenient class of metric embeddings are \emph{metric tree embeddings,} plainly because very few problems are hard to solve on tree instances.
The utility of tree embeddings originates in the fact that, despite their extremely simple topology, it is possible to randomly construct an embedding of any graph $G$ into a tree $T$ so that the \emph{expected stretch $\alpha = \max\{ \E_T[\dist(v, w, T)] / \dist(v,w,G) \mid v,w \in V \}$} satisfies $\alpha \in \bigO(\log n)$~\cite{frt-tbaamtm-04}.
By linearity of expectation, this ensures an expected approximation ratio of $\bigO(\log n)$ for most problems;
repeating the process $\log(\epsilon^{-1})$ times and taking the best result, one obtains an $\bigO(\log n)$-approximation with probability at least $1 - \epsilon$.

A substantial advantage of tree embeddings lies in the simplicity of applying the machinery once they are computed:
Translating the instance on $G$ to one on~$T$, solving the instance on~$T$, and translating the solution back tends to be extremely efficient and highly parallelizable;
we demonstrate this in Sections~\ref{sec:kmedian} and~\ref{sec:buyatbulk}.
Note also that the embedding can be computed as a preprocessing step, which is highly useful for online approximation algorithms~\cite{frt-tbaamtm-04}.
Hence, a low-depth small-work parallel algorithm for embedding weighted graphs into trees in the vein of \ac{FRT} would give rise to fast and efficient parallel approximations for a large class of graph problems.
Unfortunately, the trade-off between depth and work achieved by state-of-the-art parallel algorithms for this purpose is suboptimal.
Concretely, all algorithms of $\polylog n$ depth use $\bigOmega(n^2)$ work, whereas we are not aware of any stronger lower bound than the trivial $\bigOmega(m)$ work bound.\footnote{%
	Partition $V = A \dotcup B$ evenly, and add spanning trees of $A$ and $B$ consisting of edges of weight~$1$.
	Connect $A$ and $B$ with $m - n + 2$ edges, all of weight $W \gg n\log n$, but w.p.~$1/2$, pick one of the connecting edges uniformly at random and set its weight to~$1$.
	To approximate the distance between $a \in A$ and $b \in B$ better than factor $W / n \gg \log n$ w.p.\ substantially larger than~$1/2$, any algorithm must examine $\bigOmega(m)$ edges in expectation.}

\paragraph*{Our Contribution}

Our main contribution is to reduce the amount of work for sampling from the \ac{FRT} distribution\dash---a random distribution of tree embeddings\dash---to $\bigOT(m^{1+\epsilon})$ while maintaining $\polylog n$ depth.
This paper is organized in two parts.
The first establishes the required techniques:
\begin{itemize}
\item
	Our key tool is an algebraic interpretation of \ac{MBF-like} algorithms described in Section~\ref{sec:mbf}.
	As our framework subsumes a large class of known algorithms and explains them from a different perspective\dash---we demonstrate this using numerous examples in Section~\ref{sec:examples}\dash---and we consider it to be of independent interest.

\item
	Section~\ref{sec:h} proposes a sampling technique for embedding a graph $G$ in which $d$-hop distances $(1 + \hat\epsilon)$-approximate exact distances into a complete graph~$H$, where $H$ has polylogarithmic \ac{SPD} and preserves $G$-distances $(1 + \hat\epsilon)^{\bigO(\log n)}$-approximately.

\item
	We devise an oracle that answers \ac{MBF-like} queries by efficiently simulating an iteration of \iac{MBF-like} algorithm on $H$ in Section~\ref{sec:oracle}.
	It uses only the edges of $G$ and polylogarithmic overhead, resulting in $\bigOT(dm)$ work w.r.t.~$G$, i.e., subquadratic work, per iteration;
	we use $d \in \polylog n$.
\end{itemize}
The second part applies our techniques and establishes our results:
\begin{itemize}
\item
	A first consequence of our techniques is that we can query the oracle with \ac{APSP} to determine w.h.p.\ a $(1 + \bigo(1))$-approximate metric on $G$ using $\bigOT(nm^{1+\epsilon})$ work and $\polylog n$ depth.
	We discuss this in Section~\ref{sec:metric}.

\item
	In Section~\ref{sec:frt}, we show that for any constant $\epsilon > 0$, there is a randomized parallel algorithm of depth $\polylog n$ and work $\bigOT(m^{1 + \epsilon})$ that computes a metric tree embedding of expected stretch $\bigO(\log n)$ w.h.p.
	This follows from the above techniques and from the fact that sampling from the \ac{FRT} distribution is \ac{MBF-like}.
	Applying the spanner construction of Baswana and Sen~\cite{bs-sltracsswg-07} as a preprocessing step, the work can be reduced to $\bigOT(m + n^{1+\epsilon})$ at the expense of stretch $\bigO(\epsilon^{-1} \log n)$.

\item
	Our techniques allow to improve over previous distributed algorithms computing tree embeddings in the Congest~\cite{p-dclsa-00} model.
	We reduce the best known round complexity for sampling from a tree embedding of expected stretch $\bigO(\log n)$ from $\bigOT(n^{1/2 + \epsilon} + \diam(G))$, where $\epsilon > 0$ is an arbitrary constant and $\diam(G)$ is the unweighted hop diameter of~$G$, to $(n^{1/2} + \diam(G)) n^{\bigo(1)}$.
	This is detailed in Section~\ref{sec:distributed}.

\item
	We illustrate the utility of our main results by providing efficient approximation algorithms for the $k$-median and buy-at-bulk network design problems.
	Blelloch et~al.~\cite{bgt-pptekmbabnd-12} devise polylogarithmic depth parallel algorithms based on \ac{FRT} embeddings for these problems assuming a metric as input.
	We provide polylogarithmic depth parallel algorithms for the more general case where the metric is given implicitly by~$G$, obtaining more work-efficient solutions for a wide range of parameters.
	The details are given in Sections~\ref{sec:kmedian} and~\ref{sec:buyatbulk}, respectively.
\end{itemize}
Section~\ref{sec:conclusion} concludes the paper.

\paragraph*{Our Approach}

The algorithm of Khan et~al.~\cite{kkmpt-edaapte-12}, formulated for the Congest model~\cite{p-dclsa-00}, gives rise to an $\bigOT(\spd(G))$-depth parallel algorithm sampling from the \ac{FRT} distribution.
The \ac{SPD} is the maximum, over all $v,w \in V$, of the minimum hop-length of a shortest $v$-$w$-path.
Intuitively, $\spd(G)$ captures the number of iterations of \ac{MBF-like} algorithms in~$G$:
Each iteration updates distances until the $(\spd(G) + 1)$-th iteration does not yield new information.
Unfortunately, $\spd(G) = n - 1$ is possible, so a naive application of this algorithm results in poor performance.

A natural idea is to reduce the number of iterations by adding ``shortcuts'' to the graph.
Cohen~\cite{c-ptnlwasusp-00} provides an algorithm of depth $\polylog n$ and work $\bigOT(m^{1+\epsilon})$ that computes a \emph{$(d, \hat\epsilon)$-hop set} with $d \in \polylog n$:
This is a set $E'$ of additional edges such that $\dist(v, w, G) \leq \dist^d(v, w, G') \leq (1 + \hat\epsilon) \dist(v, w, G)$ for all $v, w \in V$, where $\hat\epsilon \in 1 / \polylog n$ and $\dist^d(v, w, G')$ is the minimum weight of a $v$-$w$-path with at most $d$ edges in $G$ augmented with~$E'$.
Note carefully that $\epsilon$ is different from~$\hat\epsilon$.
In other words, Cohen computes a metric embedding with the additional property that polylogarithmically many \ac{MBF-like} iterations suffice to determine $(1 + 1 / \polylog n)$-approximate distances.

The course of action might now seem obvious:
Run Cohen's algorithm, then run the algorithm by Khan et~al.\ on the resulting graph for $d \in \polylog n$ rounds, and conclude that the resulting output corresponds to a tree embedding of the original graph $G$ of stretch $\bigO((1 + 1 / \polylog n) \log n) = \bigO(\log n)$.
Alas, this reasoning is flawed:
Constructing \ac{FRT} trees crucially relies on the fact that the distances form a metric, i.e., satisfy the triangle inequality.
An approximate triangle inequality for approximate distances is insufficient since the \ac{FRT} construction relies on the \emph{subtractive} form of the triangle inequality, i.e., $\dist(v,w,G') - \dist(v,u,G') \leq \dist(w,u,G')$ for arbitrary $u,v,w \in V$.

Choosing a different hop set does not solve the problem:
Hop sets guarantee that $d$-hop distances \emph{approximate} distances, but any hop set that fulfills the triangle inequality on $d$-hop distances has to reduce the \ac{SPD} to at most~$d$, i.e., yield \emph{exact} distances:

\begin{observation}\label{obs:hop-set-spd}
	Let $G$ be a graph augmented with a $(d, \hat\epsilon)$-hop set.%
	\footnote{By the definitions in Section~\ref{sec:notation}.}
	If $\dist^d(\cdot, \cdot, G)$ is a metric, then $\dist^d(\cdot, \cdot, G)=\dist(\cdot, \cdot, G)$, i.e., $\spd(G) \leq d$.
\end{observation}

\begin{proof}
	Let $\pi$ be a shortest $u$-$v$-path in~$G$.
	Since $\dist^d(\cdot, \cdot, G)$ fulfills the triangle inequality,
	\begin{equation}
		\dist(u,v,G)\leq \dist^d(u, v, G)
			\leq \sum_{\{u_1, u_2\} \in \pi} \dist^d(u_1, u_2, G)
			\leq \sum_{\{u_1, u_2\} \in \pi} \weight(u_1, u_2)
			= \dist(u, v, G).\qedhere
	\end{equation}
\end{proof}

We overcome this obstacle by embedding $G'$ into a complete graph $H$ on the same node set that $(1 + \bigo(1))$-approximates distances in~$G$ but fulfills $\spd(H) \in \polylog n$.
In other words, where Cohen preserves distances \emph{exactly} and ensures existence of \emph{approximately} shortest paths with few hops, we preserve distances \emph{approximately} but guarantee that we obtain \emph{exact} shortest paths with few hops.
This yields a sequence of embeddings:
\begin{enumerate}
\item
	Start with the original graph~$G$,

\item
	augment $G$ with a $(d, 1 / \polylog n)$-hop set~\cite{c-ptnlwasusp-00}, yielding~$G'$, and

\item
	modify $G'$ to ensure a small \ac{SPD}, resulting in $H$ (Section~\ref{sec:h}).
\end{enumerate}
Unfortunately, this introduces a new obstacle:
As $H$ is complete, we cannot explicitly compute $H$ without incurring $\bigOmega(n^2)$ work.

\paragraph*{\acs{MBF-like} Algorithms}

This is where our novel perspective on \ac{MBF-like} algorithms comes into play.
We can simulate an iteration of any \ac{MBF-like} algorithm on~$H$, using only the edges of $G'$ and polylogarithmic overhead, resulting in an oracle for \acs{MBF-like} queries on~$H$.
Since $\spd(H) \in \polylog n$, the entire algorithm runs in polylogarithmic time and with a polylogarithmic work overhead w.r.t.~$G'$.

In an iteration of \iac{MBF-like} algorithm,
\begin{enumerate*}
\item
	the information stored at each node is \emph{propagated} to its neighbors,

\item
	each node \emph{aggregates} the received information, and

\item
	optionally \emph{filters} out irrelevant parts.
\end{enumerate*}
For example, in order for each node to determine the $k$ nodes closest to it, each node stores node--distance pairs (initially only themselves at distance~$0$) and then iterates the following steps:
\begin{enumerate*}
\item
	communicate the node--distance pairs to the neighbors (distances uniformly increased by the corresponding edge weight),

\item
	aggregate the received values by picking the node-wise minimum, and

\item
	discard all but the pairs corresponding to the $k$ closest sources.
\end{enumerate*}

It is well-known~\cite{agm-oteotapspp-97,m-sfasdp-02,z-apspubsarmm-02} that distance computations can be performed by multiplication with the (weighted) adjacency matrix $A$ over the min-plus semiring $\mathcal{S}_{\min,+} = (\Rdist, \min, +)$ (see Definition~\ref{def:semiring} in Appendix~\ref{app:algebra}).
For instance, if $B = A^h$ with $h \geq \spd(G)$, then $b_{vw} = \dist(v,w,G)$.
In terms of~$\mathcal{S}_{\min,+}$, propagation is the ``multiplication'' with an edge weight and aggregation is ``summation.''
The $(i+1)$-th iteration results in $x^{(i+1)} = r^V A x^{(i)}$, where $r^V$ is the (node-wise) filter and $x \in M^V$ the node values.
Both $M$ and $M^V$ form semimodules\dash---a semimodule supports scalar multiplication (propagation) and provides a semigroup (representing aggregation), compare Definition~\ref{def:semimodule} in Appendix~\ref{app:algebra}\dash---over~$\mathcal{S}_{\min,+}$.

In other words, in an $h$-iteration \ac{MBF-like} algorithm each node determines its part of the output based on its $h$-hop distances to all other nodes.
However, for efficiency reasons, various algorithms~\cite{akpw-gtgaksp-95,b-pamsaa-96,b-aamtm-98,hkn-atdacsssp-15,lp-frtcusm-13,lp-idsfc-14,lp-fpdea-15} compute only a subset of these distances.
The role of the filter is to remove the remaining values to allow for better efficiency.
The core feature of \iac{MBF-like} algorithm is that filtering is \emph{compatible} with propagation and aggregation:
If a node discards information and then propagates it, the discarded parts must be ``uninteresting'' at the receiving node as well.
We model this using a \emph{congruence relation} on the node states;
filters pick a suitable (efficiently encodable) representative of the node state's equivalence class.

\paragraph*{Constructing \acs{FRT} Trees}

This helps us to sample from the \ac{FRT} distribution as follows.
First, we observe that \iac{MBF-like} algorithm can acquire the information needed to represent \iac{FRT} tree.
Second, we can \emph{simulate} any \ac{MBF-like} algorithm on~$H$\dash---without explicitly storing~$H$\dash---using polylogarithmic overhead and \ac{MBF-like} iterations on~$G'$.
The previously mentioned sampling technique decomposes the vertices and edges of $H$ into $\Lambda \in \bigO(\log n)$ \emph{levels.}
We may rewrite its adjacency matrix as $A_H = \bigoplus_{\lambda = 0}^\Lambda P_{\lambda} A_{\lambda}^d P_{\lambda}$, where $\oplus$ is the ``addition'' of functions induced by the semimodule, $P_{\lambda}$~is a projection on nodes of at least level~$\lambda$, and $A_{\lambda}$ is a (slightly stretched) adjacency matrix of $G'$.
We are interested in $r^V A_H^h x^{(0)}$\dash---$h$~iterations on the graph $H$ followed by applying the node-wise filter~$r^V$.
The key insight is that the congruence relation allows us to apply intermediate filtering steps without changing the outcome, as filtering does not change the equivalence class of a state.
Hence, we may compute $(r^V \bigoplus_{\lambda = 0}^\Lambda P_{\lambda} (r^V A_{\lambda})^d P_{\lambda})^h x^{(0)}$ instead.
This repeated application of $r^V$ keeps the intermediate results small, ensuring that we can perform multiplication with $A_{\lambda}$ with $\bigOT(|E| + |E'|) \subseteq \bigOT(m^{1+\epsilon})$ work.
Since $d \in \polylog n$, $\Lambda \in \bigO(\log n)$, and each $A_\lambda$ accounts for $|E|+|E'|$ edges, this induces only polylogarithmic overhead w.r.t.\ iterations in~$G'$, yielding a highly efficient parallel algorithm of depth $\polylog n$ and work $\bigOT(m^{1 + \epsilon})$.

\subsection{Related Work}
\label{sec:rw}

We confine the discussion to undirected graphs.

\paragraph*{Classical Distance Computations}

The earliest\dash---and possibly also most basic\dash---algorithms for \ac{SSSP} computations are Dijkstra's algorithm~\cite{d-ntpcg-59} and the \acf{MBF} algorithm~\cite{b-rp-58,f-nft-56,m-sp-59}.
From the perspective of parallel algorithms, Dijkstra's algorithm performs excellent in terms of work, requiring $\bigOT(m)$ computational steps, but suffers from being inherently sequential, processing one vertex at a time.

\paragraph*{Algebraic Distance Computations}

The \ac{MBF} algorithm can be interpreted as a fixpoint iteration $Ax^{(i+1)} = Ax^{(i)}$, where $A$ is the adjacency matrix of the graph $G$ and ``addition'' and ``multiplication'' are replaced by $\min$ and~$+$, respectively.
This structure is known as the the min-plus semiring\dash---a.k.a.\ tropical semiring\dash---$\mathcal{S}_{\min,+} = (\Rdist,\min,+)$ (compare Section~\ref{sec:notation}), which is a well-established tool for distance computations~\cite{agm-oteotapspp-97,m-sfasdp-02,z-apspubsarmm-02}.
From this point of view, $\spd(G)$ is the number of iterations until a fixpoint is reached.
\ac{MBF} thus has depth $\bigOT(\spd(G))$ and work $\bigOT(m\spd(G))$, where small $\spd(G)$ are possible.

One may overcome the issue of large depth entirely by performing the fixpoint iteration on the matrix, by setting $A^{(0)} := A$ and iterating $A^{(i+1)} := A^{(i)} A^{(i)}$; after $\lceil \log \spd(G) \rceil \leq \lceil \log n \rceil$ iterations a fixpoint is reached~\cite{clrs-ia-09}.
The final matrix then has as entries exactly the pairwise node distances, and the computation has polylogarithmic depth.
This comes at the cost of $\bigOmega(n^3)$ work (even if $m \ll n^2$) but is as work-efficient as $n$ instances of Dijkstra's algorithm for solving \ac{APSP} in dense graphs, without incurring depth $\bigOmega(n)$.

Mohri~\cite{m-sfasdp-02} solved various shortest-distance problems using the $\mathcal{S}_{\min,+}$ semiring and variants thereof.
While Mohri's framework is quite general, our approach is different in crucial aspects:
\begin{enumerate}
\item
	Mohri uses an individual semiring for each problem and then solves it by a general algorithm.
	Our approach, on the other hand, is more generic as well as easier to use:
	We use off-the-shelf semirings\dash---usually just $\mathcal{S}_{\min,+}$\dash---and combine them with appropriate semimodules carrying problem-specific information.
	Further problem-specific customization happens in the definition of a congruence relation on the semiring;
	it specifies which parts of a node's state can be discarded because they are irrelevant for the problem.
	We demonstrate the modularity and flexibility of the approach by various examples in Section~\ref{sec:examples}, which cover a large variety of distance problems.

\item
	In our framework, node states are semimodule elements and edge weights are semiring elements;
	hence, there is no multiplication of node states.
	Mohri's approach, however, does not make that distinction and hence requires the introduction of an artificial ``multiplication'' between node states.

\item
	Mohri's algorithm can be interpreted as a generalization of Dijkstra's algorithm~\cite{d-ntpcg-59}, because it maintains a queue and, in each iteration, applies a relaxation technique to the dequeued element and its neighbors.
	This strategy is inherently sequential;
	to the best of our knowledge, we are the first to present a general algebraic framework for distance computations that exploits the implicit parallelism of the \ac{MBF} algorithm.

\item
	In Mohri's approach, choosing the global queueing strategy is not only an integral part of an algorithm, but also simplifies the construction of the underlying semirings, as one may rule that elements are processed in a ``convenient'' order.
	Our framework is flexible enough to achieve counterparts even of Mohri's more involved results without such assumptions;
	concretely, we propose a suitable semiring for solving the \ac{kSDP} and the \ac{kDSDP} in Section~\ref{sec:examples-allpaths}.
\end{enumerate}

\paragraph*{Approximate Distance Computations}

As metric embeddings reproduce distances only approximately, we may base them on approximate distance computation in the original graph.
Using rounding techniques and embedding $\mathcal{S}_{\min,+}$ into a polynomial ring, this enables to use fast matrix multiplication to speed up the aforementioned fixpoint iteration $A^{(i+1)} := A^{(i)} A^{(i)}$~\cite{z-apspubsarmm-02}.
This reduces the work to $\bigOT(n^{\omega})$ at the expense of only $(1 + \bigo(1))$-approximating distances, where $\omega < 2.3729$~\cite{lg-ptfmm-14} denotes the fast matrix-multiplication exponent.
However, even if the conjecture that $\omega = 2$ holds true, this technique must result in $\bigOmega(n^2)$ work, simply because $\bigOmega(n^2)$ pairwise distances are computed.

Regarding \ac{SSSP}, there was no work-efficient low-depth parallel algorithm for a long time, even when allowing approximation.
This was referred to as the ``sequential bottleneck:''
Matrix-matrix multiplication was inefficient in terms of work, while sequentially exploring (shortest) paths resulted in depth $\bigOmega(\spd(G))$.
Klein and Subramanian~\cite{ks-arpafsssp-97} showed that depth $\bigOT(\sqrt{n})$ can be achieved with $\bigOT(m\sqrt{n})$ work, beating the $n^2$ work barrier with sublinear depth in sparse graphs.\
As an aside, similar bounds were later achieved for exact \ac{SSSP} computations by Shi and Spencer~\cite{ss-twtssspp-99}.

In a seminal paper, Cohen~\cite{c-ptnlwasusp-00} proved that \ac{SSSP} can be $(1+\bigo(1))$-approximated at depth $\polylog n$ and near-optimal $\bigOT(m^{1+\epsilon})$ work, for any constant choice of $\epsilon > 0$;
her approach is based on the aforementioned hop-set construction.
Similar guarantees can be achieved deterministically.
Henziger et~al.~\cite{hkn-atdacsssp-15} focus on Congest algorithms, which can be interpreted in our framework to yield hop sets $(1 + 1 / \polylog n)$-approximating distances for $d \in 2^{\bigO(\sqrt{\log n})} \subset n^{\bigo(1)}$, and can be computed using depth $2^{\bigO(\sqrt{\log n}) }\subset n^{\bigo(1)}$ and work $m 2^{\bigO(\sqrt{\log n})}\subset m^{1+\bigo(1)}$.
In a recent breakthrough, Elkin and Neiman obtained hop sets with substantially improved trade-offs~\cite{en-hwchaatasp-16}, both for the parallel setting and the Congest model.

Our embedding technique is formulated independently from the underlying hop-set construction, whose performance is reflected in the depth and work bounds of our algorithms.
While the improvements by Elkin and Neiman do not enable us to achieve a work bound of $m^{1+\bigo(1)}$ when sticking to our goals of depth $\polylog n$ and expected stretch $\bigO(\log n)$, they can be used to obtain better trade-offs between the parameters.

\paragraph{Metric Tree Embeddings}

When metrically embedding into a tree, it is, in general, impossible to guarantee a small stretch.
For instance, when the graph is a cycle with unit edge weights, it is impossible to embed it into a tree without having at least one edge with stretch~$\bigOmega(n)$.
However, \emph{on average} the edges in this example are stretched by a constant factor only, justifying the hope that one may be able to randomly embed into a tree such that, for each pair of nodes, the \emph{expected} stretch is small.
A number of elegant algorithms~\cite{akpw-gtgaksp-95,b-pamsaa-96,b-aamtm-98,frt-tbaamtm-04} compute tree embeddings, culminating in the one by \acf{FRT} that achieves stretch $\bigO(\log n)$ in expectation.
This stretch bound is optimal in the worst case, as illustrated by expander graphs~\cite{b-aamtm-98}.
Mendel and Schwob show how to sample from the \ac{FRT} distribution in $\bigOT(m)$ steps~\cite{ms-fckrpsg-09}, matching the trivial $\Omega(m)$ lower bound up to polylogarithmic factors.
However, their approach relies on a pruned version of Dijkstra's algorithm for distance computations and hence does not lead to a low-depth parallel algorithm.

Several parallel and distributed algorithms compute \ac{FRT} trees~\cite{bgt-pptekmbabnd-12,gl-nodte-14,kkmpt-edaapte-12}.
These algorithms and ours have in common that they represent the embedding by \ac{LE} lists, which were first introduced in~\cite{c-sefatcr-97,ck-sdaon-07}.
In the parallel case, the state-of-the-art solution due to Blelloch et~al.~\cite{bgt-pptekmbabnd-12} achieves $\bigO(\log^2 n)$ depth and $\bigO(n^2 \log n)$ work.
However, Blelloch et~al.\ assume the input to be given as an $n$-point \emph{metric,} where the distance between two points can be queried at constant cost.
Note that our approach is more general as a metric can be interpreted as a complete weighted graph of \ac{SPD}~$1$;
a single \ac{MBF-like} iteration reproduces the result by Blelloch et~al.
Moreover, this point of view shows that the input required to achieve subquadratic work must be a sparse graph.
For graph inputs, we are not aware of any algorithms achieving $\polylog n$ depth and a non-trivial work bound, i.e., not incurring the $\bigOmega(n^3)$ work caused by relying on matrix-matrix multiplication.

In the distributed setting, Khan et~al.~\cite{kkmpt-edaapte-12} show how to compute \ac{LE} lists in $\bigO(\spd(G) \log n)$ rounds in the Congest model~\cite{p-dclsa-00}.
On the lower bound side, trivially $\bigOmega(\diam(G))$ rounds are required, where $\diam(G)$ is the maximum hop distance (i.e., ignoring weights) between nodes.
However, even if $\diam(G) \in \bigO(\log n)$, $\bigOmegaT(\sqrt{n})$ rounds are necessary~\cite{dhkknppw-dvhda-12,gl-nodte-14}.
Extending the algorithm by Khan et~al., in~\cite{gl-nodte-14} it is shown how to obtain a round complexity of $\bigOT(\min\{n^{1/2+\epsilon},\spd(G)\} + \diam(G))$ for any $\epsilon > 0$, at the expense of increasing the stretch to $\bigO(\epsilon^{-1}\log n)$.
We partly build on these ideas;
specifically, the construction in Section~\ref{sec:h} can be seen as a generalization of the key technique from~\cite{gl-nodte-14}.
As detailed in Section~\ref{sec:distributed}, our framework subsumes these algorithms and can be used to improve on the result from~\cite{gl-nodte-14}:
Leveraging further results~\cite{hkn-atdacsssp-15,lp-fpdea-15}, we obtain a metric tree embedding with expected stretch $\bigO(\log n)$ that is computed in $\min\{n^{1/2+\bigo(1)} + \diam(G)^{1+\bigo(1)}, \bigOT(\spd(G))\}$ rounds.

\subsection{Notation and Preliminaries}
\label{sec:notation}

We consider weighted, undirected graphs $G = (V, E, \weight)$ without loops or parallel edges with \emph{nodes~$V$,} \emph{edges~$E$,} and edge \emph{weights $\weight\colon E \to \R_{>0}$}.
Unless specified otherwise, we set $n := |V|$ and $m := |E|$.
For an edge $e = \{v,w\} \in E$, we write $\weight(v, w) := \weight(e)$, $\weight(v, v) := 0$ for $v \in V$, and $\weight(v,w) := \infty$ for $\{v,w\} \notin E$.
We assume that the ratio between maximum and minimum edge weight is polynomially bounded in $n$ and that each edge weight and constant can be stored with sufficient precision in a single register.\footnote{%
	As we are interested in approximation algorithms, $\bigO(\log n)$ bits suffice to encode values with sufficient precision.%
} We assume that $G$ is connected and given in the form of an adjacency list.

Let $p \subseteq E$ be a path.
$p$~has $|p|$ \emph{hops,} and \emph{weight} $\weight(p) := \sum_{e \in p} \weight(e)$.
For the nodes $v, w \in V$ let $\paths(v, w, G)$ denote the set of paths from $v$ to $w$ and $\paths^h(v, w, G)$ the set of such paths using at most $h$ hops.
We denote by $\dist^h(v, w, G) := \min\{ \weight(p) \mid p \in \paths^h(v, w, G) \}$ the minimum weight of an $h$-hop path from $v$ to~$w$, where $\min \emptyset := \infty$;
the \emph{distance} between $v$ and $w$ is $\dist(v,w,G):=\dist^n(v,w,G)$.
The \emph{shortest path hop distance} between $v$ and $w$ is $\hop(v, w, G) := \min\{ |p| \mid p\in \paths(v, w, G) \land \weight(p) = \dist(v, w, G) \}$;
$\mhspaths(v, w, G) := \{ p \in \paths^{\hop(v,w,G)}(v, w, G) \mid \weight(p) = \dist(v, w, G) \}$ denotes all \emph{min-hop shortest paths} from $v$ to~$w$.
Finally, the \emph{\acf{SPD}} of $G$ is $\spd(G) := \max\{ \hop(v, w, G) \mid v,w \in V \}$, and $\diam(G) := \min\{ h \in \N \mid \forall v,w \in V\colon \dist^h(v,w,G) < \infty \}$ is the unweighted \emph{hop diameter} of~$G$.

We sometimes use $\min$ and $\max$ as binary operators, assume $0 \in \N$, and define, for a set $N$ and $k \in \N$, $\binom{N}{k} := \{ M \subseteq N \mid |M| = k \}$ and denote by $\id\colon N \to N$ the identity function.
Furthermore, we use weak asymptotic notation hiding polylogarithmic factors in~$n$:
$\bigO(f(n) \polylog(n)) = \bigOT(f(n))$, etc.

\paragraph*{Model of Computation}

We use an abstract model of parallel computation similar to those used in circuit complexity;
the goal here is to avoid distraction by details such as read or write collisions or load balancing issues typical to PRAM models, noting that these can be resolved with (at most) logarithmic overheads.
The computation is represented by \iac{DAG} with constantly-bounded maximum indegree, where nodes represent words of memory that are given as input (indegree~$0$) or computed out of previously determined memory contents (non-zero indegree).
Words are computed with a constant number of basic instructions, e.g., addition, multiplication, comparison, etc.;
here, we also allow for the use of independent randomness.
For simplicity, a memory word may hold any number computed throughout the algorithm.
As pointed out above, $\bigO(\log n)$-bit words suffice for our purpose.

An algorithm defines, given the input, the \ac{DAG} and how the nodes' content is computed, as well as which nodes represent the output.
Given an instance of the problem, the \emph{work} is the number of nodes of the corresponding \ac{DAG} and the \emph{depth} is its longest path.
Assuming that there are no read or write conflicts, the work is thus (proportional to) the time required by a single processor (of uniform speed) to complete the computation, whereas the depth lower-bounds the time required by an infinite number of processors.
Note that the \ac{DAG} may be a random graph, as the algorithm may use randomness, implying that work and depth may be random variables.
When making probabilistic statements, we require that they hold for all instances, i.e., the respective probability bounds are satisfied after fixing an arbitrary instance.

\paragraph*{Probability}

A claim holds \emph{with high probability~(w.h.p.)} if it occurs with a probability of at least $1 - n^{-c}$ for any fixed choice of $c \in \R_{\geq 1}$;
$c$~is a constant in terms of the $\bigO$-notation.
We use the following basic statement frequently and implicitly throughout this paper.

\begin{lemma}\label{lem:whp}
	Let $\mathcal{E}_1, \dots, \mathcal{E}_k$ be events occurring w.h.p., and $k \in \poly n$.
	$\mathcal{E}_1 \cap \dots \cap \mathcal{E}_k$ occurs w.h.p.
\end{lemma}

\begin{proof}
	We have $k \leq an^b$ for fixed $a, b \in \R_{>0}$ and choose that all $\mathcal{E}_i$ occur with a probability of at least $1 - n^{-c'}$ with $c' = c + b + \log_n a$ for some fixed $c \geq 1$.
	The union bound yields
	\begin{equation}
		\Prob[\overline{\mathcal{E}_1 \cap \dots \cap \mathcal{E}_k}]
			\leq \sum_{i=1}^k \Prob[\bar{\mathcal{E}_i}]
			\leq k n^{-c'}
			= an^b n^{-c-b-\log_n a}
			= n^{-c},
	\end{equation}
	hence $\mathcal{E}_1 \cap \dots \cap \mathcal{E}_k$ occurs w.h.p.\ as claimed.
\end{proof}

\paragraph*{Hop Sets}

A graph $G = (V, E, \weight)$, \emph{contains a $(d, \hat\epsilon)$-hop set} if
\begin{equation}\label{eq:hop-set}
	\forall v,w \in V\colon\quad
		\dist^d(v, w, G) \leq (1 + \hat\epsilon) \dist(v, w, G),
\end{equation}
i.e., if its $d$-hop distances are a $(1 + \hat\epsilon)$-approximation of its distances.
This definition is based on Cohen~\cite{c-ptnlwasusp-00}, who describes how to efficiently add edges to $G$ to establish this property.

\paragraph*{Distance Metrics}

The min-plus semiring ${\mathcal S}_{\min,+} = (\Rdist, \min, +)$, also referred to as the tropical semiring, forms a semiring, i.e., a ring without additive inverses (see Definition~\ref{def:semiring} in Appendix~\ref{app:algebra}).
Unless explicitly stated otherwise, we associate $\oplus$ and $\odot$ with the addition and multiplication of the underlying ring throughout the paper;
in this case we use $a \oplus b := \min\{a,b\}$ and $a \odot b := a+b$.
Observe that $\infty$ and $0$ are the neutral elements w.r.t.\ $\oplus$ and~$\odot$, respectively.
We sometimes write $x \in \mathcal{S}_{\min,+}$ instead of $x \in \Rdist$ to refer to the elements of a semiring.
Furthermore, we follow the standard convention to occasionally leave out $\odot$ and give it priority over~$\oplus$, e.g., interpret $ab \oplus c$ as $(a \odot b) \oplus c$ for all $a,b,c \in \mathcal{S}_{\min,+}$.

The min-plus semiring is a well-established tool to determine pairwise distances in a graph via the distance product, see e.g.~\cite{agm-oteotapspp-97,m-sfasdp-02,z-apspubsarmm-02}.
Let $G = (V, E, \weight)$ be a weighted graph and let $A \in \mathcal{S}_{\min,+}^{V \times V}$ be its \emph{adjacency matrix~$A$,} given by
\begin{equation}\label{eq:minplus-adjacencymatrix}
	(a_{vw}) := \begin{cases}
	0 & \text{if }v=w\\
	\weight(v,w) & \text{if }\{v,w\}\in E\\
	\infty & \text{otherwise.}
	\end{cases}
\end{equation}
Throughout this paper, the operations involved in matrix addition and multiplication are the operations of the underlying semiring, i.e., for square matrices $A, B$ with row and column index set $V$ we have
\begin{align}
	(A \oplus B)_{vw} &= \min\{ a_{vw}, b_{vw} \}\text{ and} \\
	(AB)_{vw}         &= \min_{u \in V} \{ a_{vu} + b_{uw} \}.
\end{align}
The \emph{distance product $A^h$} corresponds to $h$-hop distances, i.e., $(A^h)_{vw} = \dist^h(v, w, G)$~\cite{agm-oteotapspp-97}.
In particular, this corresponds to the exact distances between all pairs of nodes for $h \geq \spd(G)$.

\section{\acs{MBF-like} Algorithms}
\label{sec:mbf}

The \acf{MBF} algorithm~\cite{b-rp-58,f-nft-56,m-sp-59} is both fundamental and elegant.
In its classical form, it solves the \ac{SSSP} problem:
In each iteration, each node communicates its current upper bound on its distance to the source node~$s$ (initially $\infty$ at all nodes but~$s$) plus the corresponding edge weight to its neighbors, which then keep the minimum of the received values and their previously stored one.
Iterating $h$ times determines all nodes' $h$-hop distances to~$s$.

Over the years, numerous algorithms emerged that use similar iterative schemes for distributing information~\cite{akpw-gtgaksp-95,b-pamsaa-96,b-aamtm-98,frt-tbaamtm-04,hkn-atdacsssp-15,lp-frtcusm-13,lp-idsfc-14,lp-fpdea-15}.
It is natural to ask for a characterization that captures all these algorithms.
In this section, we propose such a characterization: the class of \emph{\acs{MBF-like} algorithms.}
The common denominator of these algorithms is the following:
\begin{itemize}
\item
	An initial state vector $x^{(0)} \in M^V$ contains information initially known to each node.

\item
	In each \emph{iteration,} each node first \emph{propagates} information along all incident edges.

\item
	All nodes then \emph{aggregate} the received information.
	This and the previous step are precisely the same as updating the state vector $x^{(i)}$ by the matrix-vector product $x^{(i+1)} = A x^{(i)}$ over the min-plus semiring.

\item
	Finally, irrelevant information is \emph{filtered} out before moving on to the next iteration.
\end{itemize}
As a concrete example consider \ac{kSSP}, the task of determining for each node the list of its $k$ closest nodes.
To this end, one needs to consider all nodes as sources, i.e., run the multi-source variant of the classic \ac{MBF} algorithm with all nodes as sources.
Nodes store values in $(\Rdist)^V$, so that in iteration $i$ each node $v \in V$ can store $\dist^i(v,w,G) \in \Rdist$ for all $w \in V$.
Initially, $x^{(0)}_{vw}$ is $0$ if $v = w$ and $\infty$ everywhere else (the $0$-hop distances).
\emph{Propagating} these distances over an edge of weight $\weight(e)$ means uniformly increasing them by~$\weight(e)$.
During \emph{aggregation,} each node picks, for each target node, the smallest distance reported so far.
This is costly, since each node might learn non-$\infty$ distances values for \emph{all} other nodes.
To increase efficiency, we \emph{filter} out, in each iteration and at each node, all source--distance pairs but the $k$ pairs with smallest distance.
This reduces the amount of work per iteration from $\bigThetaT(mn)$ to $\bigThetaT(mk)$.

The filtering step generalizes from classic \ac{MBF} to \iac{MBF-like} algorithm, with the goal of reducing work.
The crucial characteristics exploited by this idea are the following.
\begin{itemize}
\item
	Propagation and aggregation are interchangeable.
	It makes no difference whether two pieces of information are propagated separately or as a single aggregated piece of information.

\item
	Filtering or not filtering after aggregation has no impact on the correctness (i.e., the output) of an algorithm, only on its efficiency.
\end{itemize}

In this section, we formalize this approach for later use in more advanced algorithms.
To this end, we develop a characterization of \ac{MBF-like} algorithms in Sections~\ref{sec:mbf-propagation-aggregation}--\ref{sec:mbf-definition} and establish basic properties in Section~\ref{sec:mbf-functions}.
We demonstrate that our approach applies to a wide variety of known algorithms in Section~\ref{sec:examples}.
In order to maintain self-containment without obstructing presentation, basic algebraic definitions are given in Appendix~\ref{app:algebra}.

\subsection{Propagation and Aggregation}
\label{sec:mbf-propagation-aggregation}

Let $M$ be the set of node states, i.e., the possible values that \iac{MBF-like} algorithm can store at a vertex.
We represent propagation of $x \in M$ over an edge of weight $s \in \Rdist$ by $s \odot x$, where $\odot\colon \Rdist \times M \to M$, aggregation of $x,y \in M$ at some node by $x \oplus y$, where $\oplus\colon M \times M \to M$, and filtering is deferred to Section~\ref{sec:mbf-filters}.
Concerning the aggregation of information, we demand that $\oplus$ is associative and has a neutral element $\bot \in M$ encoding ``no available information,'' hence $(M, \oplus)$ is a semigroup with neutral element~$\bot$.
Furthermore, we require for all $s,t \in \Rdist$ and $x,y \in M$ (note that we ``overload'' $\oplus$ and~$\odot$):
\begin{align}
	0 \odot x            &= x \label{eq:module-one} \\
	\infty \odot x       &= \bot \label{eq:module-zero} \\
	s \odot (x \oplus y) &= (s \odot x) \oplus (s \odot y) \label{eq:module-dist1} \\
	(s \oplus t) \odot x &= (s \odot x) \oplus (t \odot x) \label{eq:module-dist2} \\
	(s \odot t) \odot x  &= s \odot (t \odot x) \label{eq:module-dist3}.
\end{align}
Our requirements are quite natural:
Equations~\eqref{eq:module-one} and~\eqref{eq:module-zero} state that propagating information over zero distance (e.g.\ keeping it at a vertex) does not alter it and that propagating it infinitely far away (i.e., ``propagating'' it over a non-existing edge) means losing it, respectively.
Note that $0$ and $\infty$ are the neutral elements w.r.t.\ $\odot$ and $\oplus$ in $\mathcal{S}_{\min,+}$.
Equation~\eqref{eq:module-dist1} says that propagating aggregated information is equivalent to aggregating propagated information (along identical distances),
Equation~\eqref{eq:module-dist2} means that propagating information over a shorter of two edges is equivalent to moving it along both edges and then aggregating it (information ``becomes obsolete'' with increasing distance),
and Equation~\eqref{eq:module-dist3} states that propagating propagated information can be done in a single step.

Altogether, this is equivalent to demanding that $\mathcal M = (M, \oplus, \odot)$ is a zero-preserving semimodule (see Definition~\ref{def:semimodule} in Appendix~\ref{app:algebra}) over~$\mathcal{S}_{\min,+}$.
A straightforward choice of $\mathcal{M}$ is the direct product of $|V|$ copies of $\Rdist$, which is suitable for most of the applications we consider.
\begin{definition}[Distance Map]\label{def:distance-map}
	The \emph{distance map semimodule} $\mathcal{D} := ((\Rdist)^V, \oplus, \odot)$ is given by, for all $s \in \mathcal{S}_{\min,+}$ and $x,y \in \mathcal{D}$,
	\begin{align}
		(x \oplus y)_v &:= x_v \oplus y_v = \min\{x_v, y_v\} \\
		(s \odot x)_v  &:= s \odot x_v    = s + x_v
	\end{align}
	where $\bot := (\infty,\ldots,\infty)^\T \in \mathcal{D}$ is the neutral element w.r.t.~$\oplus$.
\end{definition}

\begin{corollary}\label{cor:distance-map}
	$\mathcal{D}$ is a zero-preserving semimodule over $\mathcal{S}_{\min,+}$ with zero $\bot = (\infty, \dots, \infty)^\T$ by Lemma~\ref{lem:semiring-to-the-k}.
\end{corollary}

Distance maps can be represented by only storing the non-$\infty$ distances (and their indices from~$V$).
This is of interest when there are few non-$\infty$ entries, which can be ensured by filtering (see below).
In the following, we denote by $|x|$ the number of non-$\infty$ entries of $x \in \mathcal{D}$.
The following lemma shows that this representation allows efficient aggregation.
\begin{lemma}\label{lem:aggregation}
	Suppose $x_1, \dots, x_n \in \mathcal D$ are stored in lists of index--distance pairs as above.
	Then $\bigoplus_{i=1}^n x_i$ can be computed with $\bigO(\log n)$ depth and $\bigO(\sum_{i=1}^n |x_i| \log n)$ work.
\end{lemma}

\begin{proof}
	We sort $\bigcup_{i=1}^n x_i$ in ascending lexicographical order.
	This can be done in parallel with $\bigO(\log (\sum_{i=1}^n|x_i|)) \subseteq \bigO(\log n)$ depth and $\bigO(\sum_{i=1}^n |x_i| \log n)$ work~\cite{a-nlnsn-83}.
	Then we delete each pair for which the next smaller pair has the same index;
	the resulting list hence contains, for each $v \in V$ for which there is a non-$\infty$ value in some list~$x_i$, the minimum such value.
	As this operation is easy to implement with $\bigO(\log n)$ depth and $\bigO(\sum_{i=1}^n |x_i| \log n)$ work, the claim follows.
\end{proof}

While $\mathcal{S}_{\min,+}$ and $\mathcal D$ suffice for most applications and are suitable to convey our ideas, it is sometimes necessary to use a different semiring.
We elaborate on this in Section~\ref{sec:examples}.
Hence, rather than confining the discussion to semimodules over~$\mathcal{S}_{\min,+}$, in the following we make general statements about an arbitrary semimodule $\mathcal{M} = (M, \oplus, \odot)$ over an arbitrary semiring $\mathcal{S} = (S, \oplus, \odot)$ wherever it does not obstruct the presentation.
It is, however, helpful to keep $\mathcal{S} = \mathcal{S}_{\min,+}$ and $\mathcal{M} = \mathcal{D}$ in mind.

\subsection{Filtering}
\label{sec:mbf-filters}

\ac{MBF-like} algorithms achieve efficiency by maintaining and propagating\dash---instead of the full amount of information nodes are exposed to\dash---only a filtered (small) representative of the information they obtained.
Our goal in this section is to capture the properties a filter must satisfy to not affect output correctness.
We start with a congruence relation, i.e., an equivalence relation compatible with propagation and aggregation, on~$\mathcal M$.
A filter $r\colon \mathcal{M} \to \mathcal{M}$ is a projection mapping all members of an equivalence class to the same representative within that class, compare Definition~\ref{def:representative}.

\begin{definition}[Congruence Relation]\label{def:congruence}
	Let $\mathcal M = (M, \oplus, \odot)$ be a semimodule over the semiring $\mathcal S$ and $\sim$ an equivalence relation on~$M$.
	We call $\sim$ a \emph{congruence relation on~$\mathcal M$} if and only if
	\begin{align}
		\forall s \in \mathcal{S}, \forall x,x' \in \mathcal{M}\colon \quad
			& x \sim x' \Rightarrow sx \sim sx' \label{eq:mult-eq} \\
		\forall x,x',y,y' \in \mathcal{M}\colon \quad
			& x \sim x' \land y \sim y' \Rightarrow x \oplus y \sim x' \oplus y'. \label{eq:sum-eq}
	\end{align}
\end{definition}
A congruence relation induces a quotient semimodule.
\begin{observation}
	Denote by $[x]$ the equivalence class of $x \in \mathcal{M}$ under the congruence relation $\sim$ on the semimodule~$\mathcal M$.
	Set $\quotient{M}{\sim} := \{ [x] \mid x \in \mathcal{M} \}$.
	Then $\quotient{\mathcal M}{\sim} := (\quotient{M}{\sim}, \oplus, \odot)$ is a semimodule with the operations $[x] \oplus [y] := [x\oplus y]$ and $s \odot [x] := [sx]$.
\end{observation}

\Iac{MBF-like} algorithm performs efficient computations by implicitly operating on this quotient semimodule, i.e., on suitable, typically small, representatives of the equivalence classes.
Such representatives are obtained in the filtering step using the a \emph{representative projection,} also referred to as \emph{filter.}
We refer to this step as filtering since, in all our applications and examples, it discards a subset of the available information that is irrelevant to the problem at hand.

\begin{definition}[Representative Projection]\label{def:representative}
	Let $\mathcal{M} = (M, \oplus, \odot)$ be a semimodule over the semiring $\mathcal S$ and $\sim$ a congruence relation on~$\mathcal M$.
	Then $r\colon \mathcal{M} \to \mathcal{M}$ is a \emph{representative projection w.r.t.~$\sim$} if and only if
	\begin{align}
		\forall x \in \mathcal{M}\colon \quad
			& x \sim r(x) \label{eq:projection} \\
		\forall x,y \in \mathcal{M}\colon \quad
			& x \sim y \Rightarrow r(x) = r(y). \label{eq:eq}
	\end{align}
\end{definition}

\begin{observation}
	A representative projection is a projection, i.e., $r^2 = r$.
\end{observation}

In the following, we typically first define a suitable projection~$r$;
this projection in turn defines equivalence classes $[x]:=\{ y \in \mathcal{M} \mid r(x)=r(y) \}$.
The following lemma is useful when we need to show that equivalence classes defined this way yield a congruence relation, i.e., are suitable for \ac{MBF-like} algorithms.
\begin{lemma}\label{lem:congruence-by-r}
	Let $\mathcal M$ be a semimodule over the semiring~$\mathcal S$, let $r\colon \mathcal{M} \to \mathcal{M}$ be a projection, and for $x,y \in \mathcal{M}$, let $x \sim y :\Leftrightarrow r(x) = r(y)$.
	Then $\sim$ is a congruence relation with representative projection $r$ if:
	\begin{align}
		\forall s \in \mathcal{S}, \forall x,x' \in \mathcal{M}\colon \quad
			& r(x) = r(x') \Rightarrow r(sx) = r(sx')\text{, and} \label{eq:congruence-by-r-product} \\
		\forall x,x',y,y' \in \mathcal{M}\colon \quad
			& r(x) = r(x') \land r(y) = r(y') \Rightarrow r(x \oplus y) = r(x' \oplus y'). \label{eq:congruence-by-r-sum}
	\end{align}
\end{lemma}

\begin{proof}
	Obviously, $\sim$ is an equivalence relation, and $r$ fulfills~\eqref{eq:projection} and~\eqref{eq:eq}.
	Conditions~\eqref{eq:mult-eq} and~\eqref{eq:sum-eq} directly follow from the preconditions of the lemma.
\end{proof}

\Iac{MBF-like} algorithm has to behave in a compatible way for all vertices in that each vertex follows the same propagation, aggregation, and filtering rules.
This induces a semimodule structure on the (possible) state vectors of the algorithm in a natural way.

\begin{definition}[Power Semimodule]\label{def:powermodule}
	Given a node set $V$ and a zero-preserving semimodule $\mathcal{M} = (M, \oplus, \odot)$ over the semiring~$\mathcal S$, we define $\mathcal{M}^V = (M^V, \oplus, \odot)$ by applying the operations of $\mathcal M$ coordinatewise, i.e., $\forall v,w\in M^V, \forall s \in \mathcal{S}$:
	\begin{align}
		(x \oplus y)_v &:= x_v\oplus y_v\text{ and} \\
		(s \odot x)_v  &:= s \odot x_v.
	\end{align}
	Furthermore, by $r^V$ we denote the componentwise application of a representative projection $r$ of~$\mathcal M$,
	\begin{equation}
	(r^V x)_v := r(x_v).
	\end{equation}
	This induces the equivalence relation $\sim$ on $\mathcal M$ via $x \sim y$ if and only if $x_v \sim y_v$ for all $v \in V$.
\end{definition}

\begin{observation}
	$\mathcal{M}^V$~is a zero-preserving semimodule over~$\mathcal{S}$ and $\bot^V := (\bot,\dots,\bot)^\T \in \mathcal{M}^V$ is its neutral element w.r.t.~$\oplus$, where $\bot$ is the neutral element of~$\mathcal M$.
	The equivalence relation $\sim$ induced by $r^V$ is a congruence relation on $\mathcal{M}^V$ with representative projection~$r^V$.
\end{observation}

\subsection{The Class of \acs{MBF-like} Algorithms}
\label{sec:mbf-definition}

The following definition connects the properties introduced and motivated above.

\begin{definition}[\acs{MBF-like} Algorithm]\label{def:mbf}
	A \emph{\acf{MBF-like} algorithm~$\mathcal A$} is determined by
	\begin{enumerate}
	\item
		a zero-preserving semimodule $\mathcal M$ over a semiring~$\mathcal{S}$,

	\item
		a congruence relation on $\mathcal M$ with representative projection $r\colon \mathcal{M} \to \mathcal{M}$, and

	\item
		initial values $x^{(0)} \in \mathcal{M}^V$ for the nodes (which may depend on the input graph).
	\end{enumerate}
	On a graph $G$ with adjacency matrix~$A$, $h$~iterations of $\mathcal A$ determine
	\begin{equation}\label{eq:mbf}
		\mathcal{A}^h(G) := x^{(h)} := r^V A^h x^{(0)}.
	\end{equation}
	Since $\mathcal{A}$ reaches a fixpoint after at most $i = \spd(G) < n$ iterations, i.e., a state where $x^{(i+1)} = x^{(i)}$, we abbreviate $\mathcal{A}(G) := A^n(G)$.
\end{definition}

Note that the definition of the adjacency matrix $A \in \mathcal{S}^{V \times V}$ depends on the choice of the semiring~$\mathcal S$.
For the standard choice of $\mathcal{S} = \mathcal{S}_{\min,+}$, which suffices for all our core results, we define $A$ in Equation~\eqref{eq:minplus-adjacencymatrix};
examples using different semirings and the associated adjacency matrices are discussed in Sections~\ref{sec:examples-maxmin}--\ref{sec:examples-bool}.

The $(i+1)$-th iteration of \iac{MBF-like} algorithm $\mathcal A$ determines $x^{(i+1)} := r^V A x^{(i)}$ (propagate, aggregate, and filter).
Thus, $h$~iterations yield $(r^V A)^h x^{(0)}$, which we show to be identical to $r^V A^h x^{(0)}$ in Corollary~\ref{cor:rV-id} of Section~\ref{sec:mbf-functions}.

\subsection{Preserving State-Equivalence across Iterations}
\label{sec:mbf-functions}

As motivated above, \ac{MBF-like} algorithms filter intermediate results;
a representative projection $r^V$ determines a small representative of each node state.
This maintains efficiency:
Nodes propagate and aggregate only small representatives of the relevant information\dash---instead of the full amount of information they are exposed to.
However, as motivated in e.g.\ Section~\ref{sec:mbf-filters}, filtering is only relevant regarding efficiency, but not the correctness of \ac{MBF-like} algorithms.

In this section, we formalize this concept in the following steps.
\begin{enumerate*}
\item
	We introduce the functions needed to iterate \ac{MBF-like} algorithms without filtering, i.e., multiplications with (adjacency) matrices.
	These \emph{\acfp{SLF}} are a proper subset of the linear\footnote{%
		A linear function $f\colon\mathcal{M} \to \mathcal{M}$ on the semimodule $\mathcal M$ over the semiring $\mathcal S$ satisfies, for all $x,y \in \mathcal{M}$ and $s \in \mathcal{S}$, that $f(x \oplus y) = f(x) \oplus f(y)$ and $f(s \odot x) = s \odot f(x)$.
	} functions on~$\mathcal{M}^V$.

\item
	The next step is to observe that \acp{SLF} are well-behaved w.r.t.\ the equivalence classes $\quotient{\mathcal{M}^V}{\sim}$ of node states.

\item
	Equivalence classes of \acp{SLF} mapping equivalent inputs to equivalent outputs yield the functions required for the study of \ac{MBF-like} algorithms.
	These form a semiring of (a~subset of) the functions on~$\quotient{\mathcal{M}^V}{\sim}$.

\item
	Finally, we observe that $r^V \sim \id$, formalizing the concepts of ``operating on equivalence classes of node states'' and ``filtering being optional w.r.t.\ correctness.''
\end{enumerate*}

\Iac{SLF} $f$ is ``simple'' in the sense that it corresponds to matrix-vector multiplications, i.e., maps $x \in \mathcal{M}^V$ such that $(f(x))_v$ is a linear combination of the coordinates $x_w$, $w \in V$, of~$x$.
\begin{definition}[\acl{SLF}]\label{def:slf}
	Let $\mathcal M$ be a semimodule over the semiring~$\mathcal{S}$.
	Each matrix $A \in \mathcal{S}^{V \times V}$ defines a \emph{\acf{SLF}} $A\colon \mathcal{M}^V \to \mathcal{M}^V$ (and vice versa) by
	\begin{equation}
		A(x)_v := (Ax)_v = \bigoplus_{w \in V} a_{vw} x_w.
	\end{equation}
\end{definition}
Thus, each iteration of \iac{MBF-like} algorithm is an application of \iac{SLF} given by an adjacency matrix followed by an application of the filter~$r^V$.
In the following, fix a semiring~$\mathcal{S}$, a semimodule $\mathcal{M}$ over~$\mathcal{S}$, and a congruence relation $\sim$ on~$\mathcal M$.
Furthermore, let $F$ denote the set of \acp{SLF}, i.e., matrices $A \in \mathcal{S}^{V \times V}$, each defining a function $A\colon \mathcal{M}^V \to \mathcal{M}^V$.

\begin{example}[Non-\acl{SLF}]\label{ex:non-simple}
	We remark that not all linear functions on $\mathcal{M}^V$ are \acp{SLF}.
	Choose $V = \{1,2\}$, $\mathcal{S} = \mathcal{S}_{\min,+}$, and $\mathcal{M} = \mathcal{D}$.
	Consider $f\colon \mathcal{M}^V \to \mathcal{M}^V$ given by
	\begin{equation}
		f\binom{(x_{11}, x_{12})}{(x_{21}, x_{22})}
			:= \binom{(x_{11} \oplus x_{12}, \infty)}{\bot}.
	\end{equation}
	While $f$ is linear, $f(x)_1$~is not a linear combination of $x_1$ and~$x_2$.
	Hence, $f$~is not \iac{SLF}.
\end{example}

Let $A,B \in F$ be \acp{SLF}.
Denote by $A(x) \mapsto Ax$ the application of the \ac{SLF} $A$ to the argument $x \in \mathcal{M}^V$.
Furthermore, we write $(A \oplus B)(x) \mapsto A(x) \oplus B(x)$ and $(A \circ B)(x) \mapsto A(B(x))$ for the addition and concatenation of \acp{SLF}, respectively.
We proceed to Lemma~\ref{lem:slf}, in which we show that matrix addition and multiplication are equivalent to the addition and concatenation of \ac{SLF} functions, respectively.
It follows that the \acp{SLF} form a semiring that is isomorphic to the matrix semiring of \ac{SLF} matrices.
Hence, we may use $A(x)$ and $Ax$ interchangeably in the following.

\begin{lemma}\label{lem:slf}
	$\mathcal{F} := (F, \oplus, \circ)$, where $\oplus$ denotes the addition of functions and $\circ$ their concatenation, is a semiring.
	Furthermore, $\mathcal F$ is isomorphic to the matrix-semiring over~$\mathcal S$, i.e., for all $A,B \in F$ and $x\in \mathcal{M}^V$,
	\begin{align}
		(A \oplus B)(x) &= (A \oplus B)x\text{ and} \label{eq:slf-plus} \\
		(A \circ B)(x)  &= ABx. \label{eq:slf-mult}
	\end{align}
\end{lemma}

\begin{proof}
	Let $A, B \in F$ and $x \in \mathcal{M}^V$ be arbitrary.
	Regarding~\eqref{eq:slf-plus} and~\eqref{eq:slf-mult}, observe that we have
	\begin{gather}
		(A \oplus B)x = Ax \oplus Bx = A(x) \oplus B(x) = (A \oplus B)(x)\text{ and} \\
		ABx = A(Bx) = A(B(x)) = (A \circ B)(x),
	\end{gather}
	respectively, i.e., addition and concatenation of \acp{SLF} are equivalent to addition and multiplication of their respective matrices.
	It follows that $\mathcal F$ is isomorphic to the matrix semiring $(\mathcal{S}^{V \times V}, \oplus, \odot)$ and hence $\mathcal F$ is a semiring as claimed.
\end{proof}

Recall that \ac{MBF-like} algorithms project node states to appropriate equivalent node states.
\acp{SLF} correspond to matrices and (adjacency) matrices correspond to \ac{MBF-like} iterations.
Hence, it is important that \acp{SLF} are well-behaved w.r.t.\ the equivalence classes $\quotient{\mathcal{M}^V}{\sim}$ of node states.
Lemma~\ref{lem:slf-properties} states that this is the case, i.e., that $Ax \sim Ax'$ for all $x' \in [x]$.

\begin{lemma}\label{lem:slf-properties}
	Let $A \in F$ be \iac{SLF}.
	Then we have, for all $x, x' \in \mathcal{M}^V$,
	\begin{equation}
			x \sim x' \quad \Rightarrow \quad Ax \sim Ax'.
	\end{equation}
\end{lemma}

\begin{proof}
	First, for $k \in \N$, let $x_1, \dots, x_k, x'_1, \dots, x'_k \in \mathcal{M}$ be such that $x_i \sim x'_i$ for all $1 \leq i \leq k$.
	We show that for all $s_1, \dots, s_k \in \mathcal{S}$ it holds that
	\begin{equation}\label{eq:linear-combination}
		\bigoplus_{i=1}^k s_i x_i \sim \bigoplus_{i=1}^k s_i x'_i.
	\end{equation}
	We argue that~\eqref{eq:linear-combination} holds by induction over~$k$.
	For $k = 1$, the claim trivially follows from Equation~\eqref{eq:mult-eq}.
	Regarding $k \geq 2$, suppose the claim holds for $k - 1$.
	Since $x_k \sim x_k'$, we have that $s_k x_k \sim s_k x_k'$ by~\eqref{eq:mult-eq}.
	The induction hypothesis yields $\bigoplus_{i=1}^{k-1} s_i x_i \sim \bigoplus_{i=1}^{k-1} s_i x_i'$.
	Hence,
	\begin{equation}
		\bigoplus_{i=1}^k s_i x_k
			= \left( \bigoplus_{i=1}^{k-1} s_i x_i \right) \oplus s_k x_k
			\stackrel{\eqref{eq:sum-eq}}{\sim}
				\left( \bigoplus_{i=1}^{k-1} s_i x'_i \right) \oplus s_k x'_k
			= \bigoplus_{i=1}^k s_i x'_k.
	\end{equation}
	As for the original claim, let $v \in V$ be arbitrary and note that we have
	\begin{equation}
		(Ax)_v = \bigoplus_{w \in V} a_{vw} x_v
			\stackrel{\eqref{eq:linear-combination}}{\sim}
			\bigoplus_{w \in V} a_{vw} x'_v = (Ax')_v. \qedhere
	\end{equation}
\end{proof}

Due to Lemma~\ref{lem:slf-properties}, each \ac{SLF} $A \in F$ not only defines a function $A\colon \mathcal{M}^V \to \mathcal{M}^V$, but also a function $A\colon \quotient{\mathcal{M}^V}{\sim} \to \quotient{\mathcal{M}^V}{\sim}$ with $A[x] := [Ax]$ ($A[x]$~does not depend on the choice of the representant $x' \in [x]$).
This is important, since \ac{MBF-like} algorithms implicitly operate on~$\quotient{\mathcal{M}^V}{\sim}$ and because they do so using adjacency matrices, which are \acp{SLF}.
As a natural next step, we rule for \acp{SLF} $A,B \in F$ that
\begin{equation}\label{eq:slf-eq}
	A \sim B
		\quad :\Leftrightarrow \quad
		\forall x \in \mathcal{M}^V\colon Ax \sim Bx,
\end{equation}
i.e., that they are equivalent if and only if they yield equivalent results when presented with the same input.
This yields equivalence classes $\quotient{F}{\sim} = \{ [A] \mid A \in F \}$.
This implies, by~\eqref{eq:slf-eq}, that $[A][x] := [Ax]$ is well-defined.
In Theorem~\ref{thm:slf}, we show that the equivalence classes of \acp{SLF} w.r.t.\ summation and concatenation form a semiring~$\quotient{\mathcal F}{\sim}$.
As \ac{MBF-like} algorithms implicitly work on~$\quotient{\mathcal{M}^V}{\sim}$, we obtain with $\quotient{\mathcal F}{\sim}$ precisely the structure that may be used to manipulate the state of \ac{MBF-like} algorithms, which we leverage throughout this paper.

\begin{theorem}\label{thm:slf}
	Each $[A] \in \quotient{F}{\sim}$ defines \iac{SLF} on~$\quotient{\mathcal{M}^V}{\sim}$.
	Furthermore, $\quotient{\mathcal{F}}{\sim} := (\quotient{F}{\sim}, \oplus, \circ)$, where $\oplus$ denotes the addition and $\circ$ the concatenation of functions, is a semiring of \acp{SLF} on $\quotient{\mathcal{M}^V}{\sim}$ with
	\begin{align}
		[A] \oplus [B] &= [A \oplus B]\text{ and} \label{eq:slf-plus2} \\
		[A] \circ [B]  &= [AB]. \label{eq:slf-mult2}
	\end{align}
\end{theorem}
\begin{proof}
	As argued above, for any $A\in F$, $[A]\in \quotient{F}{\sim}$~is well-defined on $\quotient{\mathcal{M}^V}{\sim}$ by Lemma~\ref{lem:slf-properties}.
	Equations~\eqref{eq:slf-plus2} and~\eqref{eq:slf-mult2} follow from Equations~\eqref{eq:slf-plus} and~\eqref{eq:slf-mult}, respectively:
	\begin{align}
		[A \oplus B][x]
			= [(A \oplus B)x]
			&\stackrel{\eqref{eq:slf-plus}}{=} [(A \oplus B)(x)]
			= ([A] \oplus [B])([x]) \\
		[AB][x]
			= [ABx]
			&\stackrel{\eqref{eq:slf-mult}}{=} [(A \circ B)(x)]
			= [A \circ B]([x]).
	\end{align}
	To see that $[A]$ is linear, let $s \in \mathcal{S}$ and $x,y \in \mathcal{M}^V$ be arbitrary and compute
	\begin{gather}
		[A][x] \oplus [A][y]
			= [Ax] \oplus [Ay]
			= [Ax \oplus Ay]
			= [A(x \oplus y)]
			= [A][x \oplus y]
			= [A]([x] \oplus [y])\text{ and} \\
		[A](s[x])
			= [A(sx)]
			= [s(Ax)]
			= s[Ax]
			= s[A][x].
	\end{gather}
	This implies that $\quotient{\mathcal{F}}{\sim}$ is a semiring of linear functions.
	As each function $[A]$ is represented by multiplication with (any) SLF $A'\in [A]$, $[A]$ is \iac{SLF}.
\end{proof}

The following corollary is a key property used throughout the paper.
It allows us to apply filter steps whenever convenient.
We later use this to simulate \ac{MBF-like} iterations on an implicitly represented graph whose edges correspond to entire paths in the original graph.
This is efficient only because we have the luxury of applying intermediate filtering repeatedly without affecting the output.

\begin{corollary}[$r^V \sim \id$]\label{cor:rV-id}
	For any representative projection $r$ on~$\mathcal M$, we have $r^V \sim \id$, i.e., for any \ac{SLF} $A \in F$ it holds that
	\begin{equation}\label{eq:filter-product}
		r^V A \sim A r^V \sim A.
	\end{equation}
	In particular\dash---as promised in Section~\ref{sec:mbf-definition}\dash---for any \ac{MBF-like} algorithm~$\mathcal A$, we have
	\begin{equation}\label{eq:filter-only-once}
		\mathcal{A}^h(G)
			\stackrel{\eqref{eq:mbf}}{=} r^V A^h x^{(0)}
			\stackrel{\eqref{eq:filter-product}}{=} (r^V A)^h x^{(0)}.
	\end{equation}
\end{corollary}

Finally, we stress that both the restriction to \acp{SLF} and the componentwise application of $r$ in $r^V$ are crucial for Corollary~\ref{cor:rV-id}.

\begin{example}[Non-\aclp{SLF} break Corollary~\ref{cor:rV-id}]
	Consider~$V$, $\mathcal{M}$, and~$f$ from Example~\ref{ex:non-simple}.
	If $r(x)=(x_1,\infty)$ for all $x\in \mathcal M$, we have that
	\begin{equation}
		r^Vf\binom{(2,1)}{\bot}=\binom{(1,\infty)}{\bot}\not\sim \binom{(2,\infty)}{\bot}=fr^V\binom{(2,1)}{\bot},
	\end{equation}
	implying that $r^Vf \not\sim fr^V$.
\end{example}

\begin{example}[Non-component-wise filtering breaks Corollary~\ref{cor:rV-id}]
	Consider $V = \{1,2\}$, $\mathcal{S} = \mathcal{S}_{\min,+}$, and $\mathcal{M} = \mathcal{D}$.
	Suppose $f$ is the \ac{SLF} given by $fx := \binom{x_1 \oplus x_2}{\bot}$ and $r^V(x) := \binom{x_1}{\bot}$, i.e., $r^V$~is not a component-wise application of some representative projection $r$ on~$\mathcal M$, but still a representative projection on~$\mathcal{M}^V$.
	Then we have that
	\begin{equation}
		r^V f \binom{(2, \infty)}{(1, \infty)}
			= r^V \binom{(1, \infty)}{\bot}
			= \binom{(1, \infty)}{\bot}
			\not\sim \binom{(2, \infty)}{\bot}
			= f \binom{(2, \infty)}{\bot}
			= fr^V \binom{(2, \infty)}{(1, \infty)},
	\end{equation}
	again implying that $r^Vf \not\sim fr^V$.
\end{example}

\section{A Collection of \acs{MBF-like} Algorithms}
\label{sec:examples}

For the purpose of illustration and to demonstrate the generality of our framework, we show that a variety of standard algorithms are \ac{MBF-like} algorithms;
due to the machinery established above, this is a trivial task in many cases.
In order to provide an unobstructed view on the machinery\dash---and since this section is not central to our contributions\dash---we defer proofs to Appendix~\ref{app:proofs}.

We demonstrate that some more involved distributed algorithms in the Congest model have a straightforward and compact interpretation in our framework in Section~\ref{sec:distributed}.
They compute metric tree embeddings based on the \ac{FRT} distribution;
we present them alongside an improved distributed algorithm based on the other results of this work.

\ac{MBF-like} algorithms are specified by a zero-preserving semimodule $\mathcal M$ over a semiring~$\mathcal{S}$, a representative projection of a congruence relation on~$\mathcal M$, initial states~$x^{(0)}$, and the number of iterations~$h$, compare Definition~\ref{def:mbf}.
While this might look like a lot, typically, a standard semiring and semimodule can be chosen;
the general-purpose choices of $\mathcal{S} = \mathcal{S}_{\min,+}$ and $\mathcal{M} = \mathcal{D}$ (see Definition~\ref{def:distance-map} and Corollary~\ref{cor:distance-map}) or $\mathcal{M} = \mathcal{S}_{\min,+}$ (every semiring is a zero-preserving semimodule over itself) usually are up to the task.
Refer to Sections~\ref{sec:examples-maxmin} and~\ref{sec:examples-allpaths} for examples that require different semirings.
However, even in these cases, the semirings and semimodules specified in Sections~\ref{sec:examples-maxmin} and~\ref{sec:examples-allpaths} can be reused.
Hence, all that is left to do in most cases is to pick an existing semiring and semimodule, choose $h \in \N$, and specify a representative projection~$r$.

\subsection{\acs{MBF-like} Algorithms over the Min-Plus Semiring}
\label{sec:examples-minplus}

We demonstrate that the min-plus semiring $\mathcal{S}_{\min,+}$\dash---a.k.a.\ the tropical semiring\dash---is the semiring of choice to capture many standard distance problems.
Note that we also use $\mathcal{S}_{\min,+}$ in our core result, i.e., for sampling \ac{FRT} trees.
For the sake of completeness, first recall the adjacency matrix $A$ of the weighted graph $G$ in the semiring $\mathcal{S}_{\min,+}$ from Equation~\eqref{eq:minplus-adjacencymatrix} and the distance-map semimodule $\mathcal{D}$ from Definition~\ref{def:distance-map}, consider the initialization $x^{(0)} \in \mathcal{D}^V$ with
\begin{equation}\label{eq:minplus-x0}
	x^{(0)}_{vw} := \begin{cases}
		0      & \text{if $v = w$ and} \\
		\infty & \text{otherwise,}
	\end{cases}
\end{equation}
and observe that the entries of
\begin{equation}\label{eq:minplus-xh}
	x^{(h)} := A^h x^{(0)}
\end{equation}
correspond to the $h$-hop distances in~$G$:
\begin{lemma}\label{lem:minplus-xh}
	For $h \in \N$ and $x^{(h)}$ from Equation~\eqref{eq:minplus-xh}, we have
	\begin{equation}
		x^{(h)}_{vw} = \dist^h(v,w,G).
	\end{equation}
\end{lemma}
It is well-known that the min-plus semiring can be used for distance computations~\cite{agm-oteotapspp-97,m-sfasdp-02,z-apspubsarmm-02}.
Nevertheless, for the sake of completeness, we prove Lemma~\ref{lem:minplus-xh} in terms of our notation in Appendix~\ref{app:proofs}.

As a first example, we turn our attention to source detection.
It generalizes all examples covered in this section, saving us from proving each one of them correct;
well-established examples like \ac{SSSP} and \ac{APSP} follow.
Source detection was introduced by Lenzen and Peleg~\cite{lp-edsdwlb-13}. Note, however, that we include a maximum considered distance $d$ in the definition.

\begin{example}[Source Detection~\cite{lp-edsdwlb-13}]\label{ex:source-detection}
	Given a weighted graph $G = (V,E,\weight)$, sources $S \subseteq V$, hop and result limits $h,k \in \N$, and a maximum distance $d \in \Rdist$, \emph{$(S,h,d,k)$-source detection} is the following problem:
	For each $v \in V$, determine the $k$ smallest elements of $\{ (\dist^h(v,s,G), s) \mid s \in S, \dist(v,s,G) \leq d \}$ w.r.t.\ lexicographical ordering, or all of them if there are fewer than~$k$.

	Source detection is solved by $h$ iterations of \iac{MBF-like} algorithm with $\mathcal{S} = \mathcal{S}_{\min,+}$, $\mathcal{M} = \mathcal{D}$,
	\begin{equation}\label{eq:source-detection-filter}
		r(x)_v \mapsto \begin{cases}
			x_v    & \text{if $v \in S$, $x_v \leq d$, and $x_v$ is among $k$ smallest entries of $x$ (ties broken by index),} \\
			\infty & \text{otherwise,}
		\end{cases}
	\end{equation}
	and $x^{(0)}_{vv} = 0$ if $v \in S$ and $x^{(0)}_{vw} = \infty$ in all other cases.
\end{example}
Since it may not be obvious that $r$ is a representative projection, we prove it in Appendix~\ref{app:proofs}.

\begin{example}[\acs{SSSP}]
	\emph{\acf{SSSP}} requires to determine the $h$-hop distance to $s \in V$ for all $v \in V$.
	It is solved by \iac{MBF-like} algorithm with $\mathcal{S} = \mathcal{M} = \mathcal{S}_{\min,+}$, $r = \id$, and $x^{(0)}_s = 0$, $x^{(0)}_v = \infty$ for all $v \neq s$.
\end{example}
Equivalently, one may use $(\{s\}, h, \infty, 1)$-source detection, effectively resulting in $\mathcal{M} = \mathcal{S}_{\min,+}$\dash---when only storing the non-$\infty$ entries, only the $s$-entry is relevant, however, the vertex ID of $s$ is stored as well\dash---and $r = \id$, too.

\begin{example}[\acs{kSSP}]\label{ex:k-SSP}
	\emph{\acf{kSSP}} requires to determine, for each node, the $k$ closest nodes in terms of the $h$-hop distance $\dist^h(\cdot,\cdot,G)$.
	It is solved by \iac{MBF-like} algorithm, as it corresponds to $(V, h, \infty, k)$-source detection.
\end{example}

\begin{example}[\acs{APSP}]\label{ex:apsp}
	\emph{\acf{APSP}} is the task of determining the $h$-hop distance between all pairs of nodes.
	It is solved by \iac{MBF-like} algorithm because we can use $(V,h,\infty,n)$-source detection, resulting in $\mathcal{M} = \mathcal{D}$, $r = \id$, and $x^{(0)}$ from Equation~\eqref{eq:minplus-x0}.
\end{example}

\begin{example}[\acs{MSSP}]
	In the \emph{\acf{MSSP}} problem, each node is looking for the $h$-hop distances to all nodes in a designated set $S \subseteq V$ of source nodes.
	This is solved by the \ac{MBF-like} algorithm for $(S,h,\infty,|S|)$-source detection.
\end{example}

\begin{example}[Forest Fires]\label{ex:fire}
	The nodes in a graph $G$ form a distributed sensor network, the edges represent communication channels, and edge weights correspond to distances.
	Our goal is to detect, for each node~$v$, if there is a node $w$ on fire within distance $\dist(v,w,G) \leq d$ for some $d \in \Rdist$, where every node initially knows if it is on fire.
	As a suitable \ac{MBF-like} algorithm, pick $h = n$, $\mathcal{S} = \mathcal{M} = \mathcal{S}_{\min,+}$,
	\begin{equation}
		r(x) \mapsto \begin{cases}
			x      & \text{if $x \leq d$ and} \\
			\infty & \text{otherwise,}
		\end{cases}
	\end{equation}
	and $x^{(0)}_v = 0$ if $v$ is on fire and $x^{(0)}_v = \infty$ otherwise.
\end{example}
Example~\ref{ex:fire} can be handled differently by using $(S,n,d,1)$-source detection, where $S$ are the nodes on fire.
This also reveals the closest node on fire, whereas the solution from Example~\ref{ex:fire} works in anonymous networks.
One can interpret both solutions as instances of \ac{SSSP} with a virtual source $s \notin V$ that is connected to all nodes on fire by an edge of weight~$0$.
This, however, requires a simulation argument and additional reasoning if the closest node on fire is to be determined.

\subsection{\acs{MBF-like} Algorithms over the Max-Min Semiring}
\label{sec:examples-maxmin}

Some problems require using a semiring other than~$\mathcal{S}_{\min,+}$.
As an example, consider the \ac{WPP}, also referred to as the bottleneck shortest path problem:
Given two nodes $v$ and $w$ in a weighted graph, find a $v$-$w$-path maximizing the lightest edge in the path.
More formally, we are interested in the widest-path distance between $v$ and~$w$:
\begin{definition}[Widest-Path Distance]
	Given a weighted graph $G = (V, E, \weight)$, a path $p$ has \emph{width $\width(p) := \min\{ \weight(e) \mid e \in p \}$.}
	The \emph{$h$-hop widest-path distance} between $v,w \in V$ is
	\begin{equation}
		\width^h(v, w, G) := \max_{p \in \paths^h(v,w,G)} \{\width(p)\}.
	\end{equation}
	We abbreviate $\width(v, w, G) := \width^n(v, w, G)$.
\end{definition}
An application of the \ac{WPP} are trust networks:
The nodes of a graph are entities and an edge $\{v,w\}$ of weight $0 < \weight(v,w) \leq 1$ encodes that $v$ and $w$ trust each other with $\weight(v,w)$.
Assuming trust to be transitive, $v$~trusts $w$ with $\max_{p \in \paths(v,w,G)} \min_{e \in p} \weight(e) = \width(v,w,G)$.
The \ac{WPP} requires a semiring supporting the $\max$ and $\min$ operations:

\begin{definition}[Max-Min Semiring]
	We refer to $\mathcal{S}_{\max,\min} := (\Rdist, \max, \min)$ as the \emph{max-min semiring.}
\end{definition}

\begin{lemma}\label{lem:maxmin}
	$\mathcal{S}_{\max,\min}$ is a semiring with neutral elements $0$ and~$\infty$.
\end{lemma}

\proofinappendix

\begin{corollary}\label{cor:maxmin-modules}
	$\mathcal{S}_{\max,\min}$ is a zero-preserving semimodule over itself.
	Furthermore, we have that $\mathcal{W} := ((\Rdist)^V, \oplus, \odot)$ with, for all $x,y \in (\Rdist)^V$ and $s \in \Rdist$,
	\begin{align}
		(x \oplus y)_v &:= \max\{ x_v, y_v \} \\
		(s \odot x)_v  &:= \min\{ s, x_v \}
	\end{align}
	is a zero-preserving semimodule over $\mathcal{S}_{\max,\min}$ with zero $\bot = (0, \dots, 0)^\T$ by Lemma~\ref{lem:semiring-to-the-k}.
\end{corollary}

As adjacency matrix of $G = (V, E, \weight)$ w.r.t.\ $\mathcal{S}_{\max,\min}$ we propose $A \in \mathcal{S}_{\max,\min}^{V \times V}$ with
\begin{equation}\label{eq:maxmin-adjacencymatrix}
	(a_{vw}) := \begin{cases}
		\infty       & \text{if $v = w$,} \\
		\weight(v,w) & \text{if $\{v,w\} \in E$, and} \\
		0            & \text{otherwise.}
	\end{cases}
\end{equation}
This is a straightforward adaptation of the adjacency matrix w.r.t.\ $\mathcal{S}_{\min,+}$ in Equation~\eqref{eq:minplus-adjacencymatrix}.
As an initialization $x^{(0)} \in \mathcal{W}^V$ in which each node knows the trivial path of unbounded width to itself but nothing else is given by
\begin{equation}\label{eq:maxmin-x0}
	x^{(0)}_{vw} := \begin{cases}
		\infty & \text{if $v = w$ and} \\
		0      & \text{otherwise.}
	\end{cases}
\end{equation}
Then $1 \leq h \in \N$ multiplications with~$A$, i.e., $h$~iterations, yield
\begin{equation}\label{eq:maxmin-xh}
	x^{(h)} := A^h x^{(0)}
\end{equation}
which corresponds to the $h$-hop widest-path distance:

\begin{lemma}\label{lem:maxmin-xh}
	Given $x^{(h)}$ from Equation~\eqref{eq:maxmin-xh}, we have
	\begin{equation}
		x^{(h)}_{vw} = \width^h(v, w, G).
	\end{equation}
\end{lemma}

\proofinappendix

\begin{example}[\acl{SSWP}]
	\emph{\acf{SSWP}} asks for, given a weighted graph $G = (V, E, \weight)$, a designated source node $s \in V$, and $h \in \N$, the $h$-hop widest-path distance $\width^h(s, v, G)$ for every $v \in V$.
	It is solved by \iac{MBF-like} algorithm with $\mathcal{S} = \mathcal{M} = \mathcal{S}_{\max,\min}$, $r = \id$, and $x^{(0)}_s = \infty$ and $x^{(0)}_v = 0$ for all $v \neq s$.
\end{example}

\begin{example}[\acl{APWP}]\label{ex:apwp}
	\emph{\acf{APWP}} asks for, given $G = (V, E, \weight)$ and $h \in \N$, $\width^h(v, w, G)$ for all $v,w \in V$.
	\ac{APWP} is \ac{MBF-like};
	it is solved by choosing $\mathcal{S} = \mathcal{S}_{\max,\min}$, $\mathcal{M} = \mathcal{W}$, $r = \id$, and $x^{(0)}$ from Equation~\eqref{eq:maxmin-x0} by Lemma~\ref{lem:maxmin-xh}.
\end{example}

\begin{example}[\acl{MSWP}]
	In the \emph{\acf{MSWP}} problem, each node is looking for the $h$-hop widest path distance to all nodes in a designated set $S \subseteq V$ of source nodes.
	This is solved by the same \ac{MBF-like} algorithm as for \ac{APWP} in Example~\ref{ex:apwp} when changing $x^{(0)}$ to $x^{(0)}_{vw} = \infty$ if $v = w \in S$ and $x^{(0)}_{vw} = 0$ otherwise.
\end{example}

\subsection{\acs{MBF-like} Algorithms over the All-Paths Semiring}
\label{sec:examples-allpaths}

Mohri discusses \ac{kSDP}, where each $v \in V$ is required to find the $k$ shortest paths to a designated source node $s \in V$, in the light of his algebraic framework for distance computations~\cite{m-sfasdp-02}.
Our framework captures this application as well, but requires a different semiring than $\mathcal{S}_{\min,+}$:
While $\mathcal{S}_{\min,+}$ suffices for many applications, see Section~\ref{sec:examples-minplus}, it cannot distinguish between different paths of the same length.
This is a problem in the \ac{kSDP}, because there may be multiple paths of the same length among the $k$ shortest.

\begin{observation}
	No semimodule $\mathcal M$ over $\mathcal{S}_{\min,+}$ can overcome this issue:
	The left-distributive law~\eqref{eq:dist-left-module} requires, for all $x \in \mathcal{M}$ and $s, s' \in \mathcal{S}_{\min,+}$, that $sx \oplus s'x = (s \oplus s')x$.
	Consider different paths $\pi \neq \pi'$ ending in the same node with $\weight(\pi) = s = s' = \weight(\pi')$.
	W.r.t.\ $\mathcal{S}_{\min,+}$ and~$\mathcal M$, the left-distributive law yields $sx \oplus s'x = \min\{s, s'\} \odot x$, i.e., propagating $x$ over~$\pi$, over~$\pi'$, or over both and then aggregating \emph{must be indistinguishable} in the case of $s = s'$.
\end{observation}

This does not mean that the framework of \ac{MBF-like} algorithms cannot be applied, but rather indicates that the toolbox needs a more powerful semiring than~$\mathcal{S}_{\min,+}$.
The motivation of this section is to add such a semiring, the \emph{all-paths semiring~$\mathcal{P}_{\min,+}$,} to the toolbox.
Having established~$\mathcal{P}_{\min,+}$, the advantages of the previously established machinery are available:
pick a semimodule (or use~$\mathcal{P}_{\min,+}$ itself) and define a representative projection.
We demonstrate this for \ac{kSDP} and a variant.

The basic concept of $\mathcal{P}_{\min,+}$ is simple: remember \emph{paths} instead of adding up ``anonymous'' distances.
Instead of storing the sum of the traversed edges' weight, store the string of edges.
We also add the ability to remember multiple paths into the semiring.
This includes enough features in $\mathcal{P}_{\min,+}$;
we do not require dedicated semimodules for \ac{kSDP} and use the fact that $\mathcal{P}_{\min,+}$ is a zero-preserving semimodule over itself.

We begin the technical part with a convenient representation of paths:
Let $P \subset V^+$ denote the set of non-empty, loop-free, directed paths on~$V$, denoted as tuples of nodes.
Furthermore, let $\circ \subseteq P^2$ be the relation of \emph{concatenable} paths defined by
\begin{equation}
	(v_1, \dots, v_k) \circ (w_1, \dots, w_\ell)
		\quad :\Leftrightarrow \quad
		v_k = w_1.
\end{equation}
By abuse of notation, when and if its operands are concatenable, we occasionally use $\circ$ as concatenation operator.
Furthermore, we use $\{ (\pi^1, \pi^2) \mid \pi = \pi^1 \circ \pi^2 \}$ as a shorthand for the rather cumbersome $\{ (\pi^1, \pi^2) \mid \pi^1, \pi^2 \in P \land \pi^1 \circ \pi^2 \land \text{$\pi$ is the concatenation of $\pi^1$ and $\pi^2$} \}$ to iterate over all two-splits of~$\pi$.

As motivated above, the all-paths semiring can store multiple paths.
We represent this using vectors in $(\Rdist)^P$ storing a non-$\infty$ weight for every encountered path and $\infty$ for all paths not encountered so far.
This can be efficiently represented by implicitly leaving out all $\infty$ entries.
\begin{definition}[All-Paths Semiring]
	We call $\mathcal{P}_{\min,+} = ((\Rdist)^P, \oplus, \odot)$ the \emph{all-paths semiring,} where $\oplus$ and $\odot$ are defined, for all $\pi \in P$ and $x,y \in \mathcal{P}_{\min,+}$, by
	\begin{align}
		(x \oplus y)_\pi
			&:= \min\{ x_\pi, y_\pi \}\text{ and} \\
		(x \odot y)_\pi
			&:= \min\{ x_{\pi^1} + y_{\pi^2} \mid \pi = \pi^1 \circ \pi^2 \}.
	\end{align}
	We say that \emph{$x$~contains~$\pi$ (with weight~$x_\pi$)} if and only if $x_\pi < \infty$.
\end{definition}
Summation picks the smallest weight associated to each path in either operand;
multiplication $(x \odot y)_\pi$ finds the lightest estimate for $\pi$ composed of two-splits $\pi = \pi^1 \circ \pi^2$, where $\pi^1$ is picked from $x$ and $\pi^2$ from~$y$.
Observe that $\mathcal{P}_{\min,+}$ supports upper bounds on path lengths;
we do, however, not use this feature.
Intuitively, $\mathcal{P}_{\min,+}$~stores all encountered paths with their exact weights;
in this mindset, summation corresponds to the union and multiplication to the concatenability-obeying Cartesian product of the paths contained in $x$ and~$y$.

\begin{lemma}\label{lem:allpaths}
	$\mathcal{P}_{\min,+}$ is a semiring with neutral elements
	\begin{align}
		0     &:= (\infty, \dots, \infty)^\T\text{ and} \\
		1_\pi &:= \begin{cases}
				0      & \text{if $\pi = (v)$ for some $v \in V$ and} \\
				\infty & \text{otherwise}
			\end{cases}
	\end{align}
	w.r.t.\ $\oplus$ and~$\odot$, respectively.
\end{lemma}

\proofinappendix

\begin{corollary}
	$\mathcal{P}_{\min,+}$ is a zero-preserving semimodule over itself.
\end{corollary}

Computations on a graph $G = (V, E, \weight)$ w.r.t.\ $\mathcal{P}_{\min,+}$ require\dash---this is a generalization of Equation~\eqref{eq:minplus-adjacencymatrix}\dash---an adjacency matrix $A \in \mathcal{P}_{\min,+}^{V \times V}$ defined by
\begin{equation}\label{eq:allpaths-adjacencymatrix}
	(a_{vw})_\pi := \begin{cases}
		1_\pi        & \text{if $v = w$,} \\
		\weight(v,w) & \text{if $\pi = (v,w)$, and} \\
		\infty       & \text{otherwise.}
	\end{cases}
\end{equation}
On the diagonal, $a_{vv} = 1$ contains exactly the zero-hop paths of weight~$0$;
all non-trivial paths are ``unknown'' in $a_{vv}$, i.e., accounted for with an infinite weight.
An entry $a_{vw}$ with $v \neq w$ contains, if present, only the edge $\{v,w\}$, represented by the path $(v,w)$ of weight $\weight(v,w)$;
all other paths are not contained in~$a_{vw}$.
An initialization where each node $v$ knows only about the zero-hop path $(v)$ is represented by the vector $x^{(0)} \in \mathcal{P}_{\min,+}^V$ with
\begin{equation}\label{eq:allpaths-x0}
	\left( x^{(0)}_v \right)_\pi := \begin{cases}
			0      & \text{if $\pi = (v)$ and} \\
			\infty & \text{otherwise.}
		\end{cases}
\end{equation}
Then $1 \leq h \in \N$ multiplications of $x^{(0)}$ with~$A$, i.e., $h$~iterations, yield $x^{(h)}$ with
\begin{equation}\label{eq:allpaths-xh}
	x^{(h)} := A^h x^{(0)}.
\end{equation}
As expected, $x^{(h)}_v$~contains exactly the $h$-hop paths beginning in $v$ with their according weights:
\begin{lemma}\label{lem:allpaths-xh}
	Let $x^{(h)}$ be defined as in Equation~\eqref{eq:allpaths-x0}, w.r.t.\ the graph $G = (V, E, \weight)$.
	Then for all $v \in V$ and $\pi \in P$
	\begin{equation}
		\left( x^{(h)}_v \right)_\pi = \begin{cases}
				\weight(\pi) & \text{if $\pi \in \paths^h(v, \cdot, G)$ and} \\
				\infty       & \text{otherwise.}
			\end{cases}
	\end{equation}
\end{lemma}

\proofinappendix

With the all-paths semiring $\mathcal{P}_{\min,+}$ established, we turn to the \ac{kSDP}, our initial motivation for adding $\mathcal{P}_{\min,+}$ to the toolbox of \ac{MBF-like} algorithms in the first place.

\begin{definition}[\aclp{kSDP}~\cite{m-sfasdp-02}]\label{def:ksdp}
	Given a graph $G = (V, W, \weight)$ and a designated source vertex $s \in V$, the \emph{\acf{kSDP}} asks:
	For each node $v \in V$ and considering all $v$-$s$-paths, what are the weights of the $k$ lightest such paths?
	In the \emph{\acf{kDSDP},} the path weights have to be distinct.
\end{definition}

In order to solve the \ac{kSDP}, we require a representative projection that reduces the abundance of paths stored in an unfiltered $x^{(h)}$ to the relevant ones.
Relevant in this case simply means to keep the $k$ shortest $v$-$s$-paths in $x^{(h)}_v$.
In order to formalize this, let $P(v,w,x)$ denote, for $x \in \mathcal{P}_{\min,+}$ and $v,w \in V$, the set of all $v$-$w$-paths contained in~$x$:
\begin{equation}
	P(v,w,x) := \{ \pi \in P \mid \text{$\pi$ is a $v$-$w$-path with $x_{\pi}\neq \infty$} \}.
\end{equation}
Order $P(v,w,x)$ ascendingly w.r.t.\ the weights $x_\pi$, breaking ties using an arbitrary ordering on~$P$.
Then let $P_k(v,w,x)$ denote the set of the first (at most) $k$~entries of that sequence:
\begin{equation}\label{eq:ksdp-pk}
	P_k(v,w,x) := \{ \pi \in P(v,w,x) \mid \text{$x_{\pi}$ is among the $k$ smallest entries of $x$ (ties broken by order)} \}.
\end{equation}
We define the (representative, see below) projection $r\colon \mathcal{P}_{\min,+} \to \mathcal{P}_{\min,+}$ by
\begin{equation}\label{eq:ksdp-filter}
	r(x)_\pi \mapsto \begin{cases}
			x_\pi  & \text{if $\pi \in P_k(v,s,x)$ for some $v \in V$ and} \\
			\infty & \text{otherwise.}
		\end{cases}
\end{equation}
It discards everything except, for each $v \in V$, $k$ shortest $v$-$s$-paths contained in~$x$.
Following the standard approach\dash---Lemma~\ref{lem:congruence-by-r}\dash---we define vectors $x,y \in \mathcal{P}_{\min,+}$ to be equivalent if and only if their entries for $P_k(\cdot,s,x)$ do not differ:
\begin{equation}\label{eq:kdsp-congruence}
	\forall x,y \in \mathcal{P}_{\min,+}\colon
		\quad x \sim y
		\quad :\Leftrightarrow
		\quad r(x) = r(y).
\end{equation}

\begin{lemma}\label{lem:ksdp-r}
	$\sim$ is a congruence relation on $\mathcal{P}_{\min,+}$ with representative projection~$r$.
\end{lemma}
\proofinappendix

Observe that $r$ is defined to maintain the $k$ shortest $v$-$s$-paths \emph{for all $v \in V$,} potentially storing $k|V|$ paths instead of just~$k$.
Intuitively, one could argue that $r^V x^{(h)}_v$ \emph{only} needs to contain $k$ paths, since they all start in~$v$, which is what the algorithm should actually be doing.
This objection is correct in that this is what actually happens when running the algorithm with initialization~$x^{(0)}$:
By Lemma~\ref{lem:allpaths-xh}, $x^{(h)}_v$~contains the $h$-hop shortest paths \emph{starting in~$v$} and $r$ removes all that do not end in $s$ or are too long.
On the other hand, the objection is flawed.
In order for $r$ to behave correctly w.r.t.\ all $x \in \mathcal{P}_{\min,+}$\dash---especially those less nicely structured than $x^{(h)}_v$ where all paths start at~$v$\dash---we must define $r$ as it is, otherwise the proof of Lemma~\ref{lem:ksdp-r} fails for mixed starting-node inputs.

\begin{example}[\acl{kSDP}]\label{ex:ksdp}
	\ac{kSDP}, compare Definition~\ref{def:ksdp}, is solved by an \ac{MBF-like} algorithm $\mathcal A$ with $\mathcal{S} = \mathcal{M} = \mathcal{P}_{\min,+}$, the representative projection and congruence relation defined in Equations~\eqref{eq:ksdp-filter} and~\eqref{eq:kdsp-congruence}, the choices of $A$ and $x^{(0)}$ from Equations~\eqref{eq:allpaths-adjacencymatrix} and~\eqref{eq:allpaths-x0}, and $h = \spd(G)$ iterations.
\end{example}

By Lemma~\ref{lem:allpaths-xh} and due to $h = \spd(G)$, $x^{(h)}_v$~contains all paths that start in~$v$, associated with their weights.
Since $\mathcal{A}^h(G) = r^V x^{(h)}$, by definition of $r$ in Equation~\eqref{eq:ksdp-filter}, $(r^V x^{(h)})_v = r(x^{(h)}_v)$ contains the subset of those paths that have the $k$ smallest weights and start in~$v$, i.e., precisely what \ac{kSDP} asks for.

We remark that solving a generalization of \ac{kSDP} looking for the $k$ shortest $h$-hop distances is straightforward using $h$ iterations.
Furthermore, note that our approach reveals the actual paths along with their weights.

\begin{example}[\acl{kDSDP}]
	\ac{kDSDP} from Definition~\ref{def:ksdp} can be solved analogously to \ac{kSDP} in Example~\ref{ex:ksdp}.
\end{example}
The only adjustment that needs to be made is the definition of $P_k(v,w,x)$ in Equation~\eqref{eq:ksdp-pk}.
For each of the $k$ smallest weights in~$x$, the modified $P_k(v,w,x)$ contains only one representative:
the path contained in $x$ of that weight that is first w.r.t.\ lexicographically ordering by nodes.
This results in
\begin{align}
	P'_k(v,w,x) &:= \{ \pi \in P(v,w,x) \mid \text{$x_\pi$ is among the $k$ smallest weights in $x$} \}\text{ and} \label{eq:kdsdp-pk-begin} \\
	P_k(v,w,x)  &:= \{ \pi \in P'_k(v,w,x) \mid \text{$\pi$ is lexicographically smallest in $\{\pi' \mid x_{\pi'} = x_\pi\}$} \}. \label{eq:kdsdp-pk-end}
\end{align}
The proof of Lemma~\ref{lem:ksdp-r} works without modification when replacing~\eqref{eq:ksdp-pk} with~\eqref{eq:kdsdp-pk-begin}--\eqref{eq:kdsdp-pk-end}.

\subsection{\acs{MBF-like} Algorithms over the Boolean Semiring}
\label{sec:examples-bool}

A well-known semiring is the Boolean semiring $\mathcal{B} = (\{0,1\}, \lor, \land)$ and by Lemma~\ref{lem:semiring-to-the-k}, $\mathcal{B}^V$~is a zero-preserving semimodule over~$\mathcal{B}$.
It can be used to check for connectivity in a graph\footnote{%
	For this problem, we drop the assumption that graphs are connected.
} using the adjacency matrix
\begin{equation}
	(a_{vw}) := \begin{cases}
		1 & \text{if $v=w$ or $\{v,w\} \in E$ and} \\
		0 & \text{otherwise}
	\end{cases}
\end{equation}
together with initial values
\begin{equation}\label{eq:bool-x0}
	x^{(0)}_{vw} := \begin{cases}
		1 & \text{if $v=w$ and} \\
		0 & \text{otherwise}
	\end{cases}
\end{equation}
indicating that each node $v \in V$ is connected to itself.
An inductive argument reveals that
\begin{equation}
	\left( A^h x^{(0)} \right)_{vw} = 1
		\quad \Leftrightarrow \quad
		\paths^h(v,w,G) \neq \emptyset.
\end{equation}

\begin{example}[Connectivity]
	Given a graph, we want to check which pairs of nodes are connected by paths of at most $h$ hops.
	This is solved by \iac{MBF-like} algorithm using $\mathcal{S} = \mathcal{B}$, $\mathcal{M} = \mathcal{B}^V$, $r = \id$, and $x^{(0)}$ from Equation~\eqref{eq:bool-x0}.
	This example easily generalizes to single-source and multi-source connectivity variants.
\end{example}

\section{The Simulated Graph}
\label{sec:h}

In order to determine a tree embedding of the graph~$G$, we need to determine its \ac{LE} lists (compare Section~\ref{sec:frt}).
These are the result of \iac{MBF-like} algorithm using $\mathcal{S}_{\min,+}$ and~$\mathcal D$;
its filter $r$ ensures that $|r(x^{(i)})_v| \in \bigO(\log n)$ w.h.p.\ for all~$i$, i.e., that intermediate results are small.
This allows for performing an iteration with $\bigOT(m)$ work.
However, doing so requires $\spd(G)$ iterations, which in general can be as large as $n - 1$, but we aim for polylogarithmic time.

To resolve this problem, we reduce the \ac{SPD}, accepting a slight increase in stretch.
The first step is to use Cohen's $(d, 1 / \polylog n)$-hop set~\cite{c-ptnlwasusp-00}:
a small number of additional (weighted) edges for~$G$, such that for all $v,w \in V$, $\dist^d(v,w,G') \leq (1 + \hat\epsilon) \dist(v,w,G)$, where $G'$ is $G$ augmented with the additional edges and $\hat\epsilon \in 1 / \polylog n$.
Her algorithm is sufficiently efficient in terms of depth, work, and number of additional edges.
Yet, our problem is not solved:
The $d$-hop distances in $G'$ only \emph{approximate} distances (compare Observation~\ref{obs:hop-set-spd}), but constructing \ac{FRT} trees critically depends on the triangle inequality and thus on the use of \emph{exact} distances.

In this section, we resolve this issue.
After augmenting $G$ with the hop set, we embed it into a complete graph $H$ on the same node set so that $\spd(H) \in \bigO(\log^2 n)$, keeping the stretch limited.
Where hop sets preserve distances \emph{exactly} and ensure the existence of \emph{approximately} shortest paths with few hops, $H$~preserves distances \emph{approximately} but guarantees that we obtain \emph{exact} shortest paths with few hops.
Note that explicitly constructing $H$ causes $\bigOmega(n^2)$ work;
we circumnavigate this obstacle in Section~\ref{sec:oracle} with the help of the machinery developed in Section~\ref{sec:mbf}.

Since our construction requires to first add the hop set to~$G$, assume for the sake of presentation that $G$ already contains a $(d, \hat\epsilon)$-hop set for fixed $\hat\epsilon \in \R_{>0}$ and $d \in \N$ throughout this section.
We begin our construction of $H$ by sampling levels for the vertices~$V$:
Every vertex starts at level~$0$.
In step $\lambda \geq 1$, each vertex in level $\lambda - 1$ is raised to level $\lambda$ with probability~$\frac12$.
We continue until the first step $\Lambda + 1$ where no node is sampled.
$\level(v)$~refers to the level of $v \in V$ and we define the level of an edge $e \in E$ as $\level(e) := \min\{ \level(v) \mid v \in e \}$, the minimal level of its incident vertices.

\begin{lemma}\label{lem:levels}
	W.h.p., $\Lambda \in \bigO(\log n)$.
\end{lemma}

\begin{proof}
	For $c \in \R_{\geq 1}$, $v \in V$ has $\level(v) < c \log n$ with probability $1 - (\frac12)^{c \log n} = 1 - n^{-c}$, i.e., w.h.p.
	Lemma~\ref{lem:whp} yields that all nodes have a level of less than $c \log n$ w.h.p.\ and the claim follows.
\end{proof}

The idea is to use the levels in the following way.
We devise a complete graph $H$ on~$V$.
An edge of $H$ of level $\lambda$ is weighted with the $d$-hop distance between its endpoints in~$G$\dash---a $(1 + \hat\epsilon)$-approximation of their exact distance because $G$ contains a $(d, \hat\epsilon)$-hop set by assumption\dash---multiplied with a penalty of $(1 + \hat\epsilon)^{\Lambda - \lambda}$.
This way, high-level edges are ``more attractive'' for shortest paths, because they receive smaller penalties.

\begin{definition}[Simulated graph~$H$]
	Let $G = (V, E, \weight)$ be a graph that contains a $(d, \hat\epsilon)$-hop set with levels sampled as above.
	We define the complete graph $H$ as
	\begin{gather}
		H := \left( V, \binom{V}{2}, \weight_\Lambda \right) \\
		\weight_\Lambda(\{v,w\}) \mapsto (1 + \hat\epsilon)^{\Lambda - \level(v,w)} \dist^d(v, w, G).
	\end{gather}
\end{definition}

We formalize the notion of high-level edges being ``more attractive'' than low-level paths:
In~$H$, any min-hop shortest path between two nodes of level $\lambda$ is exclusively comprised of edges of level $\lambda$ or higher;
no min-hop shortest path's level locally decreases.
Therefore, all min-hop shortest paths can be split into two subpaths, the first of monotonically increasing and the second of monotonically decreasing level.

\begin{lemma}\label{lem:high-level-path}
	Consider $v,w \in V$, $\lambda = \level(v,w)$, and $p \in \mhspaths(v, w, H)$.
	Then all edges of $p$ have level at least~$\lambda$.
\end{lemma}

\begin{proof}
	The case $\lambda = 0$ is trivial.
	Consider $1 \leq \lambda \leq \Lambda$ and, for the sake of contradiction, let $q$ be a non-trivial maximal subpath of $p$ containing only edges of level strictly less than~$\lambda$.
	Observe that $q \in \mhspaths(v', w', H)$ for some $v',w'\in V$ with $\level(v'),\level(w') \geq \lambda$.
	We have
	\begin{equation}
		\weight_\Lambda(q)
			\geq (1 + \hat\epsilon)^{\Lambda - (\lambda - 1)} \dist(v', w', G).
	\end{equation}
	However, the edge $e = \{v',w'\}$ has level $\level(v',w')\geq \lambda$ and weight
	\begin{equation}
		\weight_\Lambda(e)
			\leq (1 + \hat\epsilon)^{\Lambda - \lambda} \dist^d(v', w', G)
			\leq (1 + \hat\epsilon)^{\Lambda - (\lambda - 1)} \dist(v', w', G)
			\leq \weight_\Lambda(q)
	\end{equation}
	by construction.
	Since $|q|$ is maximal and $\level(v'), \level(w') \geq \lambda$, $q$~can only be a single edge of level $\lambda$ or higher, contradicting the assumption.
\end{proof}

Since edge levels in min-hop shortest paths are first monotonically increasing and then monotonically decreasing the next step is to limit the number of hops spent on each level.

\begin{lemma}\label{lem:prefix}
	Consider vertices $v$ and $w$ of $H$ with $\level(v), \level(w) \geq \lambda$.
	Then w.h.p., one of the following statements holds:
	\begin{gather}
		\hop(v, w, H) \in \bigO(\log n)\text{ or} \label{eq:log-hops} \\
		\forall p \in \mhspaths(v, w, H)~\exists e \in p\colon \level(e) \geq \lambda+1.\label{eq:higher-level}
	\end{gather}
\end{lemma}

\begin{proof}
	Condition on the event $\mathcal{E}_{V_\lambda}$ that $V_\lambda \subseteq V$ is the set of nodes with level $\lambda$ or higher (with level $\lambda + 1$ not yet sampled).
	Let $H_\lambda := (V_\lambda, \binom{V_\lambda}{2}, \weight_\lambda)$ with $\weight_\lambda(\{v,w\}) \mapsto (1 + \hat\epsilon)^{\Lambda - \lambda} \dist^d(v, w, G)$ denote the subgraph of $H$ spanned by $V_\lambda$ and capped at level~$\lambda$.

	Consider $p \in \mhspaths(v, w, H_\lambda)$.
	Observe that $\Prob[\level(u) \geq \lambda + 1 \mid \mathcal{E}_{V_\lambda}] = \frac12$ independently for all $u \in V_\lambda$, and hence $\Prob[\level(e) \geq \lambda + 1 \mid \mathcal{E}_{V_\lambda}] = \frac14$ for all $e \in p$.
	This probability holds independently for every other edge of~$p$.
	If $|p| \geq 2c \log_{4/3} n$ for some choice of $c \in \R_{\geq 1}$, the probability that $p$ contains no edge of level $\lambda + 1$ or higher is bounded from above by $(\frac34)^{|p|/2} \leq (\frac34)^{c \log_{4/3} n} = n^{-c}$, so $p$ contains such an edge w.h.p.

	Fix a single arbitrary $p \in \mhspaths(v, w, H_\lambda)$.
	Let $\mathcal{E}_p$ denote the event that $p$ fulfills $|p| \in \bigO(\log n)$ or contains an edge of level $\lambda + 1$ or higher;
	as argued above, $\mathcal{E}_p$~occurs w.h.p.
	Note that we cannot directly apply the union bound to deduce a similar statement for all $q \in \mhspaths(v, w, H_\lambda)$:
	There are more than polynomially many $v$-$w$-paths.
	Instead, we we argue that if $\mathcal{E}_p$ holds, it follows that all $q \in \mhspaths(v, w, H)$ must behave as claimed.

	To show that all $q \in \mhspaths(v, w, H)$ fulfill~\eqref{eq:log-hops} or~\eqref{eq:higher-level} under the assumption that $\mathcal{E}_p$ holds, first recall that $q$ only uses edges of level $\lambda$ or higher by Lemma~\ref{lem:high-level-path}.
	Furthermore, observe that $\weight_\Lambda(q) \leq \weight_\Lambda(p)$.
	If $q$ contains an edge of level $\lambda + 1$ or higher, \eqref{eq:higher-level} holds for~$q$.
	Otherwise, we have $\weight_\lambda(q) = \weight_\Lambda(q)$, and distinguish two cases:
	\begin{description}
	\item [Case~1 ($|p| \in \bigO(\log n)$):]
		We have
		\begin{equation}
			\weight_\Lambda(p)
				\leq \weight_\lambda(p)
				\leq \weight_\lambda(q)
				= \weight_\Lambda(q),
		\end{equation}
		so $\weight_\Lambda(q) = \weight_\Lambda(p)$ and $|q| \leq |p| \in \bigO(\log n)$ follows from $q \in \mhspaths(v, w, H)$.

	\item [Case~2 ($p$~contains an edge of level $\lambda + 1$ or higher):]
		This yields $\weight_\Lambda(p) < \weight_\lambda(p)$, implying
		\begin{equation}
			\weight_\Lambda(p)
				< \weight_\lambda(p)
				\leq \weight_\lambda(q)
				= \weight_\Lambda(q),
		\end{equation}
		which contradicts $q \in \mhspaths(v, w, H)$.
	\end{description}
	So far, we condition on~$\mathcal{E}_{V_\lambda}$.
	In order to remove this restriction, let $\mathcal{E}_{vw}$ denote the event that~\eqref{eq:log-hops} or~\eqref{eq:higher-level} holds for $v,w \in V$.
	The above case distinction shows that $\Prob[\mathcal{E}_{vw} \mid \mathcal{E}_{V_\lambda}] \geq 1 - n^{-c}$ for an arbitrary $c \in \R_{\geq 1}$.
	We conclude that
	\begin{align}
		\Prob[\mathcal{E}_{vw} \mid \level(v,w) \geq \lambda]
			& = \sum_{V_\lambda \subseteq V}
				\Prob[\mathcal{E}_{V_\lambda} \mid \level(v,w) \geq \lambda]
				\Prob[\mathcal{E}_{vw} \mid \mathcal{E}_{V_\lambda}] \\
			& = \sum_{\{v,w\} \subseteq V_\lambda \subseteq V}
				\Prob[\mathcal{E}_{V_\lambda} \mid \level(v,w) \geq \lambda]
				\Prob[\mathcal{E}_{vw} \mid \mathcal{E}_{V_\lambda}] \\
			& \geq \sum_{\{v,w\} \subseteq V_\lambda \subseteq V}
				\Prob[\mathcal{E}_{V_\lambda} \mid \level(v,w) \geq \lambda]
				(1 - n^{-c}) \\
			& = (1 - n^{-c})
				\sum_{\{v,w\} \subseteq V_\lambda \subseteq V}
				\Prob[\mathcal{E}_{V_\lambda} \mid \level(v,w) \geq \lambda] \\
			& = 1 - n^{-c},
	\end{align}
	which is the statement of the lemma.
\end{proof}

We argue above that any min-hop shortest path in $H$ traverses every level at most twice, Lemma~\ref{lem:prefix} states that each such traversal, w.h.p., only has a logarithmic number of hops, and Lemma~\ref{lem:levels} asserts that, w.h.p., there are only logarithmically many levels.
Together, this means that min-hop shortest paths in $H$ have $\bigO(\log^2 n)$ hops w.h.p.
Additionally, our construction limits the stretch of shortest paths in $H$ as compared to $G$ by $(1 + \hat\epsilon)^{\Lambda + 1}$, i.e., by $(1 + \hat\epsilon)^{\bigO(\log n)}$ w.h.p.

\begin{theorem}\label{thm:h}
	W.h.p., $\spd(H) \in \bigO(\log^2 n)$ and, for all $v,w \in V$,
	\begin{equation}\label{eq:h-stretch}
		\dist(v, w, G) \leq \dist(v, w, H) \leq (1 + \hat\epsilon)^{\bigO(\log n)} \dist(v, w, G).
	\end{equation}
\end{theorem}

\begin{proof}
	Fix a level~$\lambda$.
	Any fixed pair of vertices of level $\lambda$ or higher fulfills, w.h.p.,~\eqref{eq:log-hops} or~\eqref{eq:higher-level} by Lemma~\ref{lem:prefix}.
	Since there are at most $\binom{n}{2}$ such pairs, w.h.p., all of them fulfill~\eqref{eq:log-hops} or~\eqref{eq:higher-level} by Lemma~\ref{lem:whp}.

	Let $\mathcal{E}_{\log}$ denote the event that there is no higher level than $\Lambda \in \bigO(\log n)$, which holds w.h.p.\ by Lemma~\ref{lem:levels}.
	Furthermore, let $\mathcal{E}_\lambda$ denote the event hat all pairs of vertices of level $\lambda$ or higher fulfill~\eqref{eq:log-hops} or~\eqref{eq:higher-level}, which holds w.h.p.\ as argued above.
	Then $\mathcal{E} := \mathcal{E}_{\log} \cap \mathcal{E}_0 \cap \dots \cap \mathcal{E}_\Lambda$ holds w.h.p.\ by Lemma~\ref{lem:whp}.

	Condition on~$\mathcal{E}$;
	in particular, no min-hop shortest path whose edges all have the same level has more than $\bigO(\log n)$ hops.
	Consider some min-hop shortest path $p$ in~$H$.
	By Lemma~\ref{lem:high-level-path}, $p$~has two parts:
	The edge level monotonically increases in the first and monotonically decreases in the second part.
	Hence, $p$~can be split up into at most $2 \Lambda - 1$ segments, in each of which all edges have the same level.
	As this holds for all min-hop shortest paths, we conclude that $\spd(H) \in \bigO(\Lambda \log n) \subseteq \bigO(\log^2 n)$ w.h.p., as claimed.

	As for Inequality~\eqref{eq:h-stretch}, recall that $H$ is constructed from $G = (V, E, \weight)$, and that $G$ contains a $(d, \hat\epsilon)$-hop set.
	For all $v,w \in V$, we have
	\begin{equation}
		\dist(v,w,H)
			\leq \weight_\Lambda(v, w)
			\leq (1 + \hat\epsilon)^\Lambda \dist^d(v, w, G)
			\leq (1 + \hat\epsilon)^{\Lambda + 1} \dist(v, w, G)
	\end{equation}
	by construction of~$H$.
	Recalling that $\Lambda \in \bigO(\log n)$ due to $\mathcal E$ completes the proof.
\end{proof}

We use Cohen's construction to obtain a $(d, \hat\epsilon)$-hop set with $\hat\epsilon \in 1 / \polylog n$, where the exponent of $\polylog n$ is under our control~\cite{c-ptnlwasusp-00}.
A sufficiently large exponent yields $(1 + \hat\epsilon)^{\bigO(\log n)} \subseteq e^{\hat\epsilon \bigO(\log n)} \subseteq e^{1 / \polylog n} = 1 + 1 / \polylog n$, upper-bounding~\eqref{eq:h-stretch} by
\begin{equation}\label{eq:h-stretch-o1}
	\dist(v, w, G)
		\leq \dist(v, w, H)
		\in (1 + 1 / \polylog n) \dist(v, w, G)
		\subseteq (1 + \bigo(1)) \dist(v, w, G).
\end{equation}

To wrap things up:
Given a weighted graph~$G$, we augment $G$ with a $(d, 1 / \polylog n)$-hop set.
After that, the $d$-hop distances in $G$ approximate the actual distances in~$G$, but these approximations may violate the triangle inequality.
We fix this by embedding into~$H$, using geometrically sampled node levels and an exponential penalty on the edge weights with decreasing levels.
Since $H$ is a complete graph, explicitly constructing it is prohibitively costly in terms of work.
The next section shows how to avoid this issue by efficiently simulating \ac{MBF-like} algorithms on~$H$.

\section{An Oracle for \acs{MBF-like} Queries}
\label{sec:oracle}

Given a weighted graph $G$ and $\hat\epsilon \in 1 / \polylog n$, Section~\ref{sec:h} introduces a complete graph $H$ that $(1 + \bigo(1))$-approximates the distances of $G$ and w.h.p.\ has a polylogarithmic \ac{SPD}, using a $(d, \hat\epsilon)$-hop set.
$H$~would solve our problem, but we cannot explicitly write $H$ into memory, as this requires an unacceptable $\bigOmega(n^2)$ work.

Instead, we dedicate this section to an \emph{oracle that answers \ac{MBF-like} queries,} i.e., an oracle that, given a weighted graph~$G$, \iac{MBF-like} algorithm $\mathcal A$ and a number of iterations~$h$, returns $\mathcal{A}^h(H)$.
Note that while the oracle can answer distance queries in polylogarithmic depth (when, e.g., queried by \ac{SSSP}, \ac{kSSP}, or \ac{APSP}), \ac{MBF-like} queries are more general (compare Section~\ref{sec:examples}) and allow for more work-efficient algorithms (like in Section~\ref{sec:frt}).
The properties of \ac{MBF-like} algorithms discussed in Section~\ref{sec:mbf} allow the oracle to internally work on $G$ and simulate iterations of $\mathcal A$ on $H$ using~$d$, i.e., polylogarithmically many, iterations on~$G$.

Throughout this section, we denote by $A_G$ and $A_H$ the adjacency matrices of $G$ and~$H$, respectively.
Furthermore, we fix the semiring to be~$\mathcal{S}_{\min,+}$, since we explicitly calculate distances;
generalizations to other semirings are possible but require appropriate generalizations of adjacency matrices and hence obstruct presentation.

We establish this section's results in two steps:
Section~\ref{sec:oracle-decompose} derives a representation of $A_H$ in terms of~$A_G$, which is then used to efficiently implement the oracle in Section~\ref{sec:oracle-implementation}.
The oracle is used to approximate the metric of $G$ in Section~\ref{sec:metric} and to construct \iac{FRT} tree using in Section~\ref{sec:frt}, both with polylogarithmic depth.

\subsection{Decomposing \texorpdfstring{$H$}{H}}
\label{sec:oracle-decompose}

The idea is to simulate each iteration of \iac{MBF-like} algorithm $\mathcal A$ on $H$ using $d$ iterations on~$G$.
This is done for each level $\lambda \in \{0, \dots, \Lambda\}$ in parallel.
For level~$\lambda$, we run $\mathcal A$ for $d$ iterations on $G$ with edge weights scaled up by $(1 + \hat\epsilon)^{\Lambda - \lambda}$, where the initial vector is obtained by discarding all information at nodes of level smaller than~$\lambda$.
Afterwards, we again discard everything stored at vertices with a level smaller than~$\lambda$.
Since $(A_G^d)_{vw} = \dist^d(v,w,G)$, this ensures that we propagate information between nodes $v,w \in V$ with $\level(v,w) = \lambda$ with the corresponding edge weight, while discarding any exchange between nodes with $\level(v,w) < \lambda$ (which is handled by the respective parallel run).
While we also propagate information between $v$ and $w$ if $\level(v,w) > \lambda$\dash---over too long a distance because edge weights are scaled by $(1 + \hat\epsilon)^{\Lambda - \lambda} > (1 + \hat\epsilon)^{\Lambda - \level(v,w)}$\dash---the parallel run for $\level(v,w)$ correctly propagates values.
Therefore, aggregating the results of all levels (i.e., applying~$\oplus$, the source-wise minimum) and applying $r^V$ completes the simulation of an iteration of $\mathcal A$ on~$H$.

This approach resolves two complexity issues.
First, we multiply (polylogarithmically often) with~$A_G$, which\dash---as opposed to the dense~$A_H$\dash---has $\bigO(m)$ non-$\infty$ entries only.
Second, Corollary~\ref{cor:rV-id} shows that we are free to filter using $r^V$ at any time, keeping the entries of intermediate state vectors small.

We formalize the above intuition.
Recall that
\begin{equation}
	(A_H)_{vw}
		= \weight_\Lambda(v,w)
		= (1 + \hat\epsilon)^{\Lambda - \level(v,w)} \dist^d (v, w, G)
		= (1 + \hat\epsilon)^{\Lambda - \level(v,w)}(A_G^d)_{vw}.
\end{equation}
For $\lambda \in \{0, \dots, \Lambda\}$, denote by $P_\lambda$ the $\mathcal{M}^V$-projection to coordinates $V_\lambda := \{ v \in V \mid \level(v) \geq \lambda \}$:
\begin{equation}
	(P_\lambda x)_v := \begin{cases}
		x_v  & \text{if $\level(v) \geq \lambda$ and} \\
		\bot & \text{otherwise.}
	\end{cases}
\end{equation}
Observe that $P_\lambda$ is \iac{SLF} on~$\mathcal{M}^V$, where $(P_\lambda)_{vw} = 0$ if $v = w \in V_\lambda$ and $(P_\lambda)_{vw} = \infty$ otherwise.
This gives us the tools to decompose $A_H$ as motivated above.

\begin{lemma}\label{lem:decompose}
	With $(A_\lambda)_{vw} := (1 + \hat\epsilon)^{\Lambda - \lambda} (A_G)_{vw}$ (w.r.t.\ multiplication in~$\R$, not~$\odot$), we have
	\begin{equation}\label{eq:decompose}
		A_H = \bigoplus_{\lambda = 0}^\Lambda P_\lambda A_\lambda^d P_\lambda.
	\end{equation}
\end{lemma}

\begin{proof}
	Since $(A^d_G)_{vw} = \dist^d(v, w, G)$, it holds that $(A_\lambda^d)_{vw} = (1 + \hat\epsilon)^{\Lambda - \lambda} \dist^d(v, w, G)$.
	Therefore,
	\begin{equation}
		(A_\lambda^d P_\lambda)_{vw}
			= \min_{u \in V} \left\{ (A_\lambda^d)_{vu} + (P_\lambda)_{uw}\right\}
			= \begin{cases}
				(1 + \hat\epsilon)^{\Lambda - \lambda} \dist^d(v, w, G) & \text{if $w \in V_\lambda$ and} \\
				\infty                                                  & \text{otherwise,}
			\end{cases}
	\end{equation}
	and hence
	\begin{equation}
		(P_\lambda A_\lambda^d P_\lambda)_{vw}
			= \min_{u \in V} \left\{ (P_\lambda)_{vu} + (A_\lambda^d P_\lambda)_{uw} \right\}
			= \begin{cases}
				(1 + \hat\epsilon)^{\Lambda - \lambda} \dist^d(v, w, G) & \text{if $v, w \in V_\lambda$ and} \\
				\infty                                                  & \text{otherwise.}
			\end{cases}
	\end{equation}
	We conclude that
	\begin{align}
		\left( \bigoplus_{\lambda = 0}^\Lambda P_\lambda A_\lambda^d P_\lambda \right)_{vw}
			& = \min_{\lambda = 0}^{\level(v,w)}
				\left\{ (1 + \hat\epsilon)^{\Lambda - \lambda} \dist^d(v, w, G) \right\} \\
			& = (1 + \hat\epsilon)^{\Lambda - \level(v,w)} \dist^d(v, w, G) \\
			& = (A_H)_{vw}. \qedhere
	\end{align}
\end{proof}

Having decomposed~$A_H$, we analyze $\mathcal{A}^h(H)$ in that regard, taking the freedom to apply filters intermediately.
For all $h \in \N$, we have
\begin{equation}\label{eq:decompose-intermediate}
	A_H^h
		\stackrel{\eqref{eq:decompose}}{=}
			\left( \bigoplus_{\lambda = 0}^\Lambda P_\lambda A_\lambda^d P_\lambda \right)^h
		\stackrel{\eqref{eq:filter-product}}{\sim}
			\left( r^V \left( \bigoplus_{\lambda = 0}^\Lambda P_\lambda(r^V A_\lambda)^d P_\lambda \right)\right)^h r^V,
\end{equation}
and hence
\begin{equation}\label{eq:computation}
	\mathcal{A}^h(H)
		= r^V A_H^h x^{(0)}
		\stackrel{\eqref{eq:eq}, \eqref{eq:decompose-intermediate}}{=} \left( r^V \left(
				\bigoplus_{\lambda = 0}^\Lambda P_\lambda(r^V A_\lambda)^d P_\lambda
			\right)\right)^h r^V x^{(0)}.
\end{equation}
Observe that we can choose $h = \spd(H) \in \bigO(\log^2 n)$ w.h.p.\ by Theorem~\ref{thm:h} and recall that $d \in \polylog n$.
Overall, this allows us to determine $\mathcal{A}(H)$ with polylogarithmic depth and $\bigOT(m)$ work, provided we can implement the individual steps, see below, at this complexity.

\subsection{Implementing the Oracle}
\label{sec:oracle-implementation}

The oracle determines iterations of $\mathcal A$ on $H$ using iterations on $G$ while only introducing a polylogarithmic overhead w.r.t.\ iterations in~$G$.
With the decomposition from Lemma~\ref{lem:decompose} at hand, it can be implemented as follows.

Given a state vector $x^{(i)} \in \mathcal{M}^V$, simulate one iteration of $\mathcal A$ on $H$ for edges of level~$\lambda$, i.e., determine $y_\lambda := P_\lambda (r^V A_\lambda)^d P_\lambda x^{(i)}$ by
\begin{enumerate*}
\item
	discarding entries at nodes of a level smaller than~$\lambda$,

\item
	running $d$ iterations of $\mathcal A$ with distances stretched by $(1 + \hat\epsilon)^{\Lambda - \lambda}$ on~$G$, applying the filter after each iteration, and

\item
	again discarding entries at nodes with levels smaller than~$\lambda$.
\end{enumerate*}
After running this procedure in parallel for all $0 \leq \lambda \leq \Lambda$, perform the $\oplus$-operation and apply the filter, i.e., for each node $v \in V$ determine $x^{(i+1)}_v = r(\bigoplus_{\lambda = 0}^\Lambda y_\lambda)_v$.

The efficiency of the above procedure depends on the semimodule $\mathcal{M}$ and the filter used by the \ac{MBF-like} algorithm.
Since our core results as well as many examples work with $\mathcal{M} = \mathcal{D}$, as specified Definition~\ref{def:distance-map}, we fix $\mathcal{M} = \mathcal{D}$ for Theorem~\ref{thm:oracle};
see Remark~\ref{rem:oracle} for how to generalize Theorem~\ref{thm:oracle} to arbitrary semimodules.
Nevertheless, we do not give such a general statement as it obstructs presentation and is not required for our results in the following sections.

\begin{theorem}[Oracle]\label{thm:oracle}
	Consider the zero-preserving semimodule~$\mathcal{D}$ (see Definition~\ref{def:distance-map}) over $\mathcal{S}_{\min,+}$;
	suppose $x \in \mathcal{D}$ is represented as list of index--distance pairs, where all $\infty$-distances are dropped (compare Lemma~\ref{lem:aggregation}).
	If for each intermediate state vector $y = (r^V A_\lambda)^f x^{(i)}$ (corresponding to $f$ iterations w.r.t.\ $A_\lambda$ starting at state~$x^{(i)}$), for non-negative integers $f \leq d$, $i \leq h$, and $\lambda \leq \Lambda$, we can compute $r^V A_\lambda y$ and $r^V y$ with depth $D$ and work~$W$, we can w.h.p.
	\begin{enumerate}
	\item
		determine $\mathcal{A}^h(H)$ using $\bigO((d + \log n) W h \log n) \subseteq \bigOT(dWh)$ work and a depth bounded by $\bigO((dD + \log n)h) \subseteq \bigOT(dDh)$, i.e., we can
	\item
		calculate $\mathcal{A}(H)$ using $\bigO((d + \log n) W \log^3 n) \subseteq \bigOT(dW)$ work and $\bigO((dD + \log n) \log^2 n) \subseteq \bigOT(dD)$ depth.
	\end{enumerate}
\end{theorem}

\begin{proof}
	By Equation~\eqref{eq:computation}, we have to compute
	\begin{equation}
		\mathcal{A}^h(H) = \left( r^V \left( \bigoplus_{\lambda = 0}^\Lambda P_\lambda (r^V A_\lambda)^d P_\lambda \right) \right)^h r^V x^{(0)}.
	\end{equation}
	Computing $r^V x^{(0)}$ requires work $W$ and depth $D$ by assumption.
	Concerning~$P_\lambda$, note that we can evaluate $(P_\lambda y)_{v \in V}$ lazily, i.e., determine whether $(P_\lambda y)_v$ evaluates to $\bot$ or to $y_v$ only if it is accessed.
	Thus, work and depth can increase by at most a constant factor due to all applications of~$P_\lambda$, $0 \leq \lambda \leq \Lambda$.
	Together with the assumption, this means that $(r^V A_\lambda P_\lambda)y$ can be determined in $\bigO(W)$ work and $\bigO(D)$ depth;
	hence, $(r^V A_\lambda)^d P_\lambda y$ requires $\bigO(dW)$ work and $\bigO(dD)$ depth.

	The set of summands of $\bigoplus_{\lambda = 0}^\Lambda P_\lambda (r^V A_\lambda)^d P_\lambda y$ can be determined using $\bigO(\Lambda dW)$ work and the same depth, since this is independent for each~$\lambda$.
	Performing the aggregation is possible in $\bigO(\log n)$ depth and an overhead of factor $\bigO(\log n)$ in work as compared to writing the lists by Lemma~\ref{lem:aggregation}.
	As each list can be determined with work $W$ by assumption, their total length is at most~$\Lambda W$, so we arrive at $\bigO((d + \log n)\Lambda W)$ work and $\bigO(dD + \log n)$ depth.
	Determining $r^V (\bigoplus_{\lambda = 0}^\Lambda P_\lambda (r^V A_\lambda)^d P_\lambda y)$ requires an extra $W$ work and $D$ depth by assumption, which is dominated by the depth and work accumulated so far.

	Repeating this $h$ times to determine $\mathcal{A}^h(H)$ yields $\bigO((d + \log n)\Lambda Wh)$ work and $\bigO((dD + \log n)h)$ depth.
	By Lemma~\ref{lem:levels}, w.h.p.\ $\Lambda \in \bigO(\log n)$ and we arrive at $\bigO((d + \log n)Wh\log n)$ work and $\bigO((dD + \log n)h)$ depth, which is the first claim.
	Recalling that by Theorem~\ref{thm:h} w.h.p.\ $\spd(H) \in \bigO(\log^2 n)$ yields the second claim.
\end{proof}

\begin{remark}[Generalization to other Semimodules]\label{rem:oracle}
	It is possible to generalize Theorem~\ref{thm:oracle} to other semimodules.
	This can be done directly for a specific semimodule or, more generally, by parameterizing Theorem~\ref{thm:oracle} with the work $W_\oplus(W,\Lambda)$ and depth $D_\oplus(W,\Lambda)$ required for the aggregation step, i.e., to determine $r^V \bigoplus_{\lambda=0}^\Lambda y_\lambda$ from $y_\lambda = P_\lambda (r^V A_\lambda)^d P_\lambda x^{(i)}$.
	For this approach, $W_\oplus(W,\Lambda)$ and $D_\oplus(W,\Lambda)$ may not only depend on~$\Lambda$, the number of aggregated elements, but also on~$W$, since the work to determine each $y$ bounds the size of its representation from above (we do this in the proof of Theorem~\ref{thm:oracle}).
	As an example, observe that in the case of $\mathcal{M} = \mathcal{D}$ we have $W_\oplus(W,\Lambda) \in \bigO(\Lambda W \log n)$ and $D_\oplus(W,\Lambda) \in \bigO(\log n)$ by Lemma~\ref{lem:aggregation}.
\end{remark}

\section{Approximate Metric Construction}
\label{sec:metric}

As a consequence of the machinery in Section~\ref{sec:oracle}, observe that we can determine a $(1 + \bigo(1))$-approximate metric on an arbitrary graph by querying the oracle with \ac{APSP} on $H$ using polylogarithmic depth and $\bigOT(nm^{1+\epsilon})$ work.
This is much more work-efficient on sparse graphs than the naive approach using $\bigO(n^3 \log n)$ work (squaring the adjacency matrix $\lceil \log_2 n \rceil$ times) for obtaining $\dist(\cdot,\cdot,G)$ exactly.
Furthermore, this section serves as an example on how to apply Theorem~\ref{thm:oracle}.

\begin{theorem}[$(1 + \bigo(1))$-Approximate Metric]\label{thm:apsp}
	Given a weighted graph $G = (V, E, \weight)$ and a constant $\epsilon > 0$, we can w.h.p.\ compute, using $\bigOT(n(m+n^{1+\epsilon}))$ work and $\polylog n$ depth, a metric on $V$ offering constant-time query access\dash---e.g.\ represented as $V \times V$ matrix over $\Rdist$\dash---that $(1 + 1 / \polylog n)$-approximates $\dist(\cdot, \cdot, G)$.
\end{theorem}

\begin{proof}
	First augment $G$ with a $(d, 1 / \polylog n)$-hop set using $\bigOT(m^{1 + \epsilon})$ work and $\polylog n$ depth with $d \in \polylog n$ using Cohen's hop-set construction~\cite{c-ptnlwasusp-00}.
	The resulting graph has $\bigOT(m+n^{1+\epsilon})$ edges.
	An iteration of \ac{APSP}, compare Example~\ref{ex:apsp}, incurs $\bigO(\log n)$ depth and $\bigO(\delta_v n \log n)$ work at a node $v$ of degree $\delta_v$ by Lemma~\ref{lem:aggregation}.
	Hence, $D \in \bigO(\log n)$ depth and $W \in \bigO(\sum_{v \in V} \delta_v n \log n) \subseteq \bigOT(n(m+n^{1+\epsilon}))$ work suffice for an entire iteration;
	the trivial filter $r^V = \id$ does not induce any overhead.
	By Theorem~\ref{thm:oracle}, we can w.h.p.\ simulate $\spd(H)$ iterations of \ac{APSP} on $H$ using $\bigOT(n(m+n^{1+\epsilon}))$ work and $\bigOT(1)$ depth.
	Due to Theorem~\ref{thm:h} and Equation~\eqref{eq:h-stretch-o1}, this yields a metric which $(1 + 1 / \polylog n)$-approximates $\dist(\cdot, \cdot, G)$.
\end{proof}

Using the sparsifier of Baswana and Sen~\cite{bs-sltracsswg-07}, we can obtain a metric with a different work--approximation trade-off.
Note that this is near-optimal in terms of work due to the trivial lower bound of $\bigOmega(n^2)$ for writing down the solution.

\begin{theorem}[$\bigO(1)$-Approximate Metric]\label{thm:apsp2}
	For a weighted graph $G = (V, E, \weight)$ and a constant $\epsilon > 0$, we can w.h.p.\ compute a metric that $\bigO(1)$-approximates $\dist(\cdot,\cdot,G)$ using $\bigOT(n^{2+\epsilon})$ work and $\polylog n$ depth.
\end{theorem}

\begin{proof}
	Baswana and Sen show how to compute a $(2k - 1)$-spanner of $G = (V, E, \weight)$, i.e., $E' \subseteq E$ such that $G' := (V, E', \weight)$ fulfills, for all $v,w \in V$,
	\begin{equation}
		\dist(v, w, G) \leq \dist(v, w, G') \leq (2k - 1) \dist(v, w, G),
	\end{equation}
	using $\bigOT(1)$ depth and $\bigOT(m)$ work with $|E'| \in \bigO(kn^{1 + 1/k})$ in expectation~\cite{bs-sltracsswg-07}.
	W.l.o.g., $k \in \bigO(\log n)$ because $kn^{1/k} = k2^{\log n / k}$ starts growing beyond that point.
	This results in $\bigOT(n^{1 + 1/k})$ edges in expectation.
	Furthermore, the algorithm of Baswana and Sen uses $\bigOT(n^{1 + 1/k})$ edges w.h.p.

	We compute an $\bigO(1)$-approximate metric of as follows.
	\begin{enumerate*}
	\item
		Compute a $(2k - 1)$-spanner for $k = \lceil 1 / (\sqrt{1+\epsilon} - 1) \rceil$.
		This is possible within the given bounds of work and depth, and w.h.p.\ yields $|E'| \in \bigOT(n^{1+1/k}) = \bigOT(n^{\sqrt{1+\epsilon}})$ edges and a stretch that is constant w.r.t.\ $n$ and~$m$.

	\item
		Apply Theorem~\ref{thm:apsp} to $G' := (V, E', \weight)$ and $\epsilon' := \sqrt{1 + \epsilon} - 1$.
		This induces $\bigOT(1)$ depth and $\bigOT(n^{2 + \epsilon})$ work.
	\end{enumerate*}
	By construction, the resulting metric has stretch $(2k - 1)(1 + \bigo(1)) \subseteq \bigO(1)$.
\end{proof}

Blelloch et~al.~\cite{bgt-pptekmbabnd-12} show how to construct \iac{FRT} tree from a metric using $\bigO(n^2)$ work and $\bigO(\log^2 n)$ depth.
Combining this with Theorem~\ref{thm:apsp2} enables us to w.h.p.\ construct \iac{FRT} tree from a graph $G$ using polylogarithmic depth and $\bigOT(n^{2+\epsilon})$ work.
While this does not yield the same \ac{FRT} tree as when directly embedding $G$ since we ``embed an approximation of~$\dist(\cdot,\cdot,G)$,'' it has the same expected asymptotic stretch of $\bigO(\log n)$ due to the constant-factor approximation provided by Theorem~\ref{thm:apsp2}.
This can, however, be done more efficiently on sparse graphs:
Constructing \ac{FRT} trees is \iac{MBF-like} algorithm and solving the problem directly\dash---using the oracle\dash---reduces the work to $\bigOT(m^{1+\epsilon})$;
this is the goal of Section~\ref{sec:frt}.

\section{\acs{FRT} Construction}
\label{sec:frt}

Given a weighted graph~$G$, determining a metric that $\bigO(1)$-approximates $\dist(\cdot,\cdot,G)$---using polylogarithmic depth and $\bigOmegaT(n^{2+\epsilon})$ work\dash---is straightforward, see Theorem~\ref{thm:apsp2};
the oracle is queried with the \ac{MBF-like} \ac{APSP} algorithm, implicitly enjoying the benefits of the \acs{SPD}-reducing sampling technique of Section~\ref{sec:h}.
In this section, we show that collecting the information required to construct \ac{FRT} trees\dash---\ac{LE} lists\dash---is \iac{MBF-like} algorithm, i.e., a query that can be directly answered by the oracle.
Since collecting \ac{LE} lists is more work-efficient than \ac{APSP}, this leads to our main result:
w.h.p.\ sampling from the \ac{FRT} distribution using polylogarithmic depth and $\bigOT(m^{1+\epsilon})$ work.

We begin with a formal definition of metric (tree) embeddings in general and the \ac{FRT} embedding in particular in Section~\ref{sec:frt-embedding}, proceed to show that the underlying algorithm is \ac{MBF-like} (Section~\ref{sec:frt-mbf}) and that all intermediate steps are sufficiently efficient in terms of depth and work (Section~\ref{sec:frt-efficient}), and present our main results in Section~\ref{sec:frt-result}.
Section~\ref{sec:frt-datastructure} describes how to retrieve the original paths in $G$ that correspond to the edges of the sampled \ac{FRT} tree.

\subsection{Metric Tree Embeddings}
\label{sec:frt-embedding}

We use this section to introduce the (distribution over) metric tree embeddings of Fakcharoenphol, Rao, and Talwar, referred to as \ac{FRT} embedding, which has expected stretch $\bigO(\log n)$~\cite{frt-tbaamtm-04}.

\begin{definition}[Metric Embedding]\label{def:embedding}
	Let $G = (V, E, \weight)$ be a graph.
	A \emph{metric embedding of stretch $\alpha$ of~$G$} is a graph $G' = (V', E', \weight')$, such that $V \subseteq V'$ and
	\begin{equation}\label{eq:embedding-stretch}
		\forall v,w \in V\colon \quad
			\dist(v, w, G) \leq \dist(v, w, G') \leq \alpha \dist(v, w, G),
	\end{equation}
	for some $\alpha \in \R_{\geq 1}$.
	If $G'$ is a tree, we refer to it as \emph{metric tree embedding.}
	For a random distribution of metric embeddings~$G'$, we require $\dist(v,w,G) \leq \dist(v,w,G')$ and define the \emph{expected stretch} as
	\begin{equation}\label{eq:embedding-stretch-expected}
		\alpha := \max_{v \neq w \in V} \E\left[ \frac{\dist(v,w,G')}{\dist(v,w,G)} \right].
	\end{equation}
\end{definition}

We show how to efficiently sample from the \ac{FRT} distribution for the graph $H$ introduced in Section~\ref{sec:h}.
As $H$ is an embedding of $G$ with a stretch in $1 + \bigo(1)$, this results in a tree embedding of $G$ of stretch $\bigO(\log n)$.
Khan et~al.~\cite{kkmpt-edaapte-12} show that a suitable representation of (a~tree sampled from the distribution of) the \acs{FRT} embedding~\cite{frt-tbaamtm-04} can be constructed as follows.
\begin{enumerate}
\item\label{item:frt-random-begin}
	Choose $\beta \in [1,2)$ uniformly at random.

\item\label{item:frt-random-end}
	Choose uniformly at random a total order of the nodes (i.e., a uniformly random permutation).
	In the following, $v < w$ means that $v$ is smaller than $w$ w.r.t.\ to this order.

\item\label{item:frt-tree-begin}
	Determine for each node $v \in V$ its \emph{\acf{LE} list:}
	This is the list obtained by deleting from $\{ (\dist(v,w,H), w) \mid w \in V \}$ all pairs $(\dist(v,w,H), w)$ for which there is some $u \in V$ with $\dist(v,u,H) \leq \dist(v,w,H)$ and $u < w$.
	Essentially, $v$~learns, for every distance~$d$, the smallest node within distance at most~$d$, i.e., $\min\{ w \in V \mid \dist(v,w,G) \leq d \}$.

\item\label{item:frt-tree-end}
	Denote by $\weight_{\min} := \min_{e \in E}\{\weight(e)\}$ and $\weight_{\max} := \max_{e \in E}\{\weight(e)\}$ the minimum and maximum edge weight, respectively;
	recall that $\weight_{\max}/\weight_{\min} \in \poly n$ by assumption.
	From the \ac{LE} lists, determine for each $v \in V$ and distance $\beta 2^i \in [\weight_{\min}/2, 2\weight_{\max}]$, $i \in \Z$, the node $v_i := \min \{w \in V \mid \dist(v,w,H) \leq \beta 2^i\}$.
	W.l.o.g., we assume that $i \in \{0, \dots, k\}$ for $k \in \bigO(\log n)$ (otherwise, we shift the indices of the nodes $v_i$ accordingly).
	Hence, for each $v \in V$, we obtain a sequence of nodes $(v_0, v_1,\dots, v_k)$.
	$(v_0, v_1, \dots, v_k)$ is the leaf corresponding to $v = v_0$ of the tree embedding, $(v_1, \dots, v_k)$ is its parent, and so on; the root is~$(v_k)$.
	The edge from $(v_i, \dots, v_k)$ to $(v_{i+1}, \dots, v_k)$ has weight $\beta 2^i$.
\end{enumerate}
We refer to~\cite{gl-nodte-14} for a more detailed summary.

The above procedure implicitly specifies a random distribution over tree embeddings with expected stretch $\bigO(\log n)$~\cite{frt-tbaamtm-04}, which we call the \emph{\ac{FRT} distribution.}
We refer to following the procedure \ref{item:frt-random-begin}--\ref{item:frt-tree-end} as \emph{sampling} from the \ac{FRT} distribution.
Once the randomness is fixed, i.e., steps \ref{item:frt-random-begin}--\ref{item:frt-random-end} are completed, the tree resulting from steps~\ref{item:frt-tree-begin}--\ref{item:frt-tree-end} is unique;
we refer to them as \emph{constructing \iac{FRT} tree}.

The next lemma shows that step~\ref{item:frt-tree-end}, i.e., constructing the \ac{FRT} tree from the \ac{LE} lists, is easy.

\begin{lemma}\label{lem:tree-explicit}
	Given \ac{LE} lists of length $\bigO(\log n)$ for all vertices, the corresponding \ac{FRT} tree can be determined using $\bigO(n \log^3 n)$ work and $\bigO(\log^2 n)$ depth.
\end{lemma}

\begin{proof}
	Determining $\weight_{\max}$, $\weight_{\min}$, and the range of indices $i$ is straightforward at this complexity, as is sorting of each node's list in ascending order w.r.t.\ distance.
	Note that in each resulting list of distance/node pairs, the nodes are strictly decreasing in terms of the random order on the nodes, and each list ends with an entry for the minimal node.
	For each node $v$ and entry $(d,u)$ in its list in parallel, we determine the values of $i\in \{0,\ldots,k\}$ such that $u$ is the smallest node within distance $\beta 2^i$ of $v$.
	This is done by reading the distance value $d'$ of the next entry of the list (using $d'=\beta 2^k+1$ if $(d,u)$ is the last entry) and writing to memory $v_i=u$ for each $i$ satisfying that $d\leq \beta 2^i <d'$.
	Since $\weight_{\max}/\weight_{\min}\in \poly n$, this has depth $\bigO(\log n)$ and a total work of $\bigO(n\log^2 n)$.

	Observe that we computed the list $(v_0,\ldots,v_k)$ for each $v\in V$.
	Recall that the parents of the leaf $(v_0, \dots, v_k)$ are determined by its $k$ suffixes.
	It remains to remove duplicates wherever nodes share a parent.
	To this end, we sort the list (possibly with duplicates) of $(k+1)n\in \bigO(n\log n)$ suffixes\dash---each with $\bigO(\log n)$ entries\dash---lexicographically, requiring $\bigO(n\log^3 n)$ work and depth $\bigO(\log^2 n)$, as comparing two suffixes requires depth and work $\bigO(\log n)$.
	Then duplicates can be removed by comparing each key to its successor in the sorted sequence, taking another $\bigO(n\log^2 n)$ work and $\bigO(\log n)$ depth.

	Note that tree edges and their weights are encoded implicitly, as the parent of each node is given by removing the first node from the list, and the level of a node (and thus the edge to its parent) is given by the length of the list representing it.
	If required, it is thus trivial to determine, e.g., an adjacency list with $\bigO(n\log^2 n)$ work and depth $\bigO(\log^2 n)$.
	Overall, we spent $\bigO(n\log^3 n)$ work at $\bigO(\log^2 n)$ depth.
\end{proof}

\subsection{Computing \acs{LE} Lists is \acs{MBF-like}}
\label{sec:frt-mbf}

Picking $\beta$ is trivial and choosing a random order of the nodes can be done w.h.p.\ by assigning to each node a string of $\bigO(\log n)$ uniformly and independently chosen random bits.
Hence, in the following, we assume this step to be completed, w.l.o.g.\ resulting in a random assignment of the vertex IDs $\{1, \dots, n\}$.
It remains to establish how to efficiently compute \ac{LE} lists.

We establish that \ac{LE} lists can be computed by \iac{MBF-like} algorithm, compare Definition~\ref{def:mbf}, using the parameters in Definition~\ref{def:le};
the claim that Equations~\eqref{eq:le-r-definition} and~\eqref{eq:le-sim-definition} define a representative projection and a congruence relation is shown in Lemma~\ref{lem:frt-mbflike}.
\begin{definition}\label{def:le}
	For constructing \ac{LE} lists, use the semiring $\mathcal{S} = \mathcal{S}_{\min,+}$ and the distance map $\mathcal{M} = \mathcal{D}$ from Definition~\ref{def:distance-map} as zero-preserving semimodule.
	For all $x \in \mathcal{D}$, define
	\begin{gather}
		r(x)_v := \begin{cases}
			\infty & \text{$\exists w < v\colon\ x_w \leq x_v$ and} \\
			x_v    & \text{otherwise, and}
		\end{cases} \label{eq:le-r-definition} \\
		x \sim y \quad :\Leftrightarrow \quad r(x) = r(y) \label{eq:le-sim-definition}
	\end{gather}
	as representative projection and congruence relation, respectively.
	As initialization $x^{(0)} \in \mathcal{D}^V$ use
	\begin{equation}\label{eq:le-x0}
		x^{(0)}_{vw} := \begin{cases}
				0      & \text{if $v=w$ and} \\
				\infty & \text{otherwise.}
			\end{cases}
	\end{equation}
\end{definition}

Hence, $r(x)$~is the \ac{LE} list of $v \in V$ if $x_w = \dist(v,w,H)$ for all $w \in V$ and we consider two lists equivalent if and only if they result in the same \ac{LE} list.
This allows us to prepare the proof that retrieving \ac{LE} lists can be done by \iac{MBF-like} algorithm in the following lemma.
It states that filtering keeps the relevant information:
If a node--distance pair is dominated by an entry in a distance map, the filtered distance map also contains a\dash---possibly different\dash---dominating entry.

\begin{lemma}
	Consider arbitrary $x,y \in \mathcal{D}$, $v \in V$, and $s \in \Rdist$.
	Then
	\begin{equation}\label{eq:le-filtering}
		\exists w < v\colon x_w \leq s \quad \Leftrightarrow \quad
			\exists w < v\colon r(x)_w \leq s
	\end{equation}
\end{lemma}

\begin{proof}
	Observe that the necessity ``$\Leftarrow$'' is trivial.
	As for sufficiency~``$\Rightarrow$,'' suppose that there is $w < v$ such that $x_w \leq s$.
	If $r(x)_w = x_w$, we are done.
	Otherwise, there must be some $u < w < v$ satisfying $x_u \leq x_w \leq x_v$.
	Since $|V|$ is finite, an inductive repetition of the argument yields that there is some $w' < v$ with $r(x)_{w'} = x_{w'} \leq s$.
\end{proof}

Equipped with this lemma, we can prove that $\sim$ is a congruence relation on $\mathcal{D}$ with representative projection~$r$.
We say that a node--distance pair \emph{$(v, d)$ dominates $(v', d')$} if and only if $v < v'$ and $d \leq d'$;
in the context of $x \in \mathcal{D}$, we say that $x_w$ dominates $x_v$ if and only if $(w, x_w)$ dominates~$(v, x_v)$.

\begin{lemma}\label{lem:frt-mbflike}
	The equivalence relation $\sim$ from Equation~\eqref{eq:le-sim-definition} of Definition~\ref{def:le} is a congruence relation.
	The function $r$ from Equation~\eqref{eq:le-r-definition} Definition~\ref{def:le} is a representative projection w.r.t.~$\sim$.
\end{lemma}

\begin{proof}
	Trivially, $r$~is a projection, i.e., $r^2(x) = r(x)$ for all $x \in \mathcal{D}$.
	By Lemma~\ref{lem:congruence-by-r}, it hence suffices to show that~\eqref{eq:congruence-by-r-product} and~\eqref{eq:congruence-by-r-sum} hold.
	In order to do that, let $s \in \mathcal{S}_{\min,+}$ be arbitrary, and $x,x',y,y' \in \mathcal{D}$ such that $r(x) = r(x')$ and $r(y) = r(y')$.
	As we have $x_v \leq x_w \Leftrightarrow s + x_v \leq s + x_w$ for all $v,w \in V$, \eqref{eq:congruence-by-r-product}~immediately follows from \eqref{eq:le-filtering}.

	Regarding~\eqref{eq:congruence-by-r-sum}, we show that
	\begin{equation}\label{eq:frt-mbflike-1}
		r(x \oplus y) = r(r(x) \oplus r(y))
	\end{equation}
	which implies~\eqref{eq:congruence-by-r-sum} due to $r(x \oplus y) = r(r(x) \oplus r(y)) = r(r(x') \oplus r(y')) = r(x' \oplus y')$.
	Let $v \in V$ be an arbitrary vertex and observe that $(x \oplus y)_v$ is dominated if and only if
	\begin{align}
			&\exists w < v\colon\quad
				(x \oplus y)_w \leq (x \oplus y)_v \\
		\Leftrightarrow\quad
			&\exists w < v\colon\quad
				\min\{x_w, y_w\} \leq (x \oplus y)_v \\
		\Leftrightarrow\quad
			&\exists w < v\colon\quad
				x_w \leq (x \oplus y)_v \lor y_w \leq (x \oplus y)_v \\
		\stackrel{\eqref{eq:le-filtering}}{\Leftrightarrow}\quad
			&\exists w < v\colon\quad
				r(x)_w \leq (x \oplus y)_v \lor r(y)_w \leq (x \oplus y)_v.\label{eq:frt-mbflike-2}
	\end{align}
	In order to show~\eqref{eq:frt-mbflike-1}, we distinguish two cases.
	\begin{description}
	\item [Case~1 ($(x \oplus y)_v$ is dominated):]
		By Definition~\ref{def:le}, we have $r(x \oplus y)_v = \infty$.
		Additionally, we know that $(r(x) \oplus r(y))_v=\min\{r(x)_v,r(y)_v\}\geq \min\{x_v,y_v\}=(x\oplus y)_v$ must be dominated due to~\eqref{eq:frt-mbflike-2}, and hence $r(r(x) \oplus r(y))_v = \infty = r(x \oplus y)_v$.

	\item [Case~2 ($(x \oplus y)_v$ is not dominated):]
		This means that by Definition~\ref{def:le}, $r(x \oplus y)_v = (x \oplus y)_v = \min\{ x_v, y_v \}$.
		Furthermore, the negation of~\eqref{eq:frt-mbflike-2} holds, i.e., $\forall w<v\colon \min\{ r(x)_w, r(y)_w \} > (x \oplus y)_v = \min\{x_v, y_v\}$.
		Assuming w.l.o.g.\ that $x_v\leq y_v$ (the other case is symmetric), we have that $x_v = (x \oplus y)_v = r(x \oplus y)_v$ and that $x_v = r(x)_v = (r(x) \oplus r(y))_v$, where $x_v = r(x)_v$ is implied by~\eqref{eq:le-filtering} because $r(x)_w \geq \min\{r(x)_w, r(y)_w\} > \min\{x_v, y_v\} = x_v$ for any $w < v$.
		It follows that
		\begin{equation}
			r(r(x) \oplus r(y))_v
				= r(r(x))_v
				= r(x)_v
				= x_v
				= r(x \oplus y)_v.
		\end{equation}
	\end{description}
	Altogether, this shows~\eqref{eq:frt-mbflike-1} and, as demonstrated above, implies~\eqref{eq:congruence-by-r-sum}.
\end{proof}

Having established that determining \ac{LE} lists can be done by \iac{MBF-like} algorithm allows us to apply the machinery developed in Sections~\ref{sec:mbf}--\ref{sec:oracle}.
Next, we establish that \ac{LE} list computations can be performed efficiently, which we show by bounding the length of LE lists.

\subsection{Computing \acs{LE} Lists is Efficient}
\label{sec:frt-efficient}

Our course of action is to show that \ac{LE} list computations are efficient using Theorem~\ref{thm:oracle}, i.e., the oracle theorem.
The purpose of this section is to prepare the lemmas required to apply Theorem~\ref{thm:oracle}.
We stress that the key challenge is to perform each iteration in polylogarithmic depth;
this allows us to determine $\mathcal{A}(H)$ in polylogarithmic depth due to $\spd(H) \in \bigO(\log^2 n)$.
To this end, we first establish the length of intermediate \ac{LE} lists to be logarithmic w.h.p.\ (Lemma~\ref{lem:lists-short}).
This permits to apply $r^V$ and determine the matrix-vector multiplication with~$A_{\lambda}$\dash---the scaled version of~$A_G$, the adjacency matrix of $G$ from Section~\ref{sec:oracle}\dash---in a sufficiently efficient manner (Lemmas~\ref{lem:compute-r(x)} and~\ref{lem:computation-cheap}).
Section~\ref{sec:frt-result} plugs these results into Theorem~\ref{thm:oracle} to establish our main result.

We remark that \ac{LE} lists are known to have length $\bigO(\log n)$ w.h.p.\ throughout intermediate computations~\cite{gl-nodte-14,kkmpt-edaapte-12}, assuming that \ac{LE} lists are assembled using $h$-hop distances.
Lemma~\ref{lem:lists-short}, while using the same key argument, is more general since it makes no assumption about $x$ except for its independence of the random node order;
we need the more general statement due to our decomposition of~$A_H$.

Recall that by $|x|$ we denote the number of non-$\infty$ entries of $x \in \mathcal{D}$ and that we only need to keep the non-$\infty$ entries in memory.
Lemma~\ref{lem:lists-short} shows that any \ac{LE} list $r(x) \in \mathcal{D}$ has length $|r(x)| \in \bigO(\log n)$ w.h.p., provided that $x$ does not depend on the random node ordering.
Observe that, in fact, the lemma is quite powerful, as it suffices that there is \emph{any} $y \in [x]$ that does not depend on the random node ordering:
as $r(x)=r(y)$, then $|r(x)|=|r(y)|\in \bigO(\log n)$ w.h.p.

\begin{lemma}\label{lem:lists-short}
	Let $x \in \mathcal{D}$ be arbitrary but independent of the random order of the nodes.
	Then $|r(x)| \in \bigO(\log n)$ w.h.p.
\end{lemma}

\begin{proof}
	Order the non-$\infty$ values of $x$ by ascending distance, breaking ties independently of the random node order.
	Denote for $i \in \{ 1, \dots, |x| \}$ by $v_i \in V$ the $i$-th node w.r.t.\ this order, i.e., $x_{v_i}$~is the $i$-th smallest entry in~$x$.
	Furthermore, denote by $X_i$ the indicator variable which is~$1$ if $v_i < v_j$ for all $j \in \{ 1, \dots, i-1 \}$ and $0$ otherwise.
	Clearly, $\E[X_i] = 1/i$, implying for $X := \sum_{i=1}^{|x|} X_i$ that
	\begin{equation}
		\E[X]
			= \sum_{i=1}^{|x|} \frac{1}{i}
			\leq \sum_{i=1}^{n} \frac{1}{i}
			\in \bigTheta(\log n).
	\end{equation}
	Observe that $X_i$ is independent of $\{ X_1, \dots, X_{i-1} \}$, as whether $v_i < v_j$ for all $j < i$ is independent of the internal order of the set $\{ v_1, \dots, v_{i-1} \}$.
	This is sufficient to apply Chernoff's bound\dash---we detail on this in Lemma~\ref{lem:chernoff} for the sake of self-containment\dash---yielding that $X \in \Theta(\log n)$ w.h.p.
	As $\Prob[X = k] = \Prob[|r(x)| = k]$, this concludes the proof.
\end{proof}

Hence, filtered, possibly intermediate \ac{LE} lists $r(x)$ w.h.p.\ comprise $\bigO(\log n)$ entries.
We proceed to show that under these circumstances, $r(x)$~can be computed efficiently.

\begin{lemma}\label{lem:compute-r(x)}
	Let $x \in \mathcal{D}$ be arbitrary.
	Then $r(x)$ can be computed using $\bigO(|r(x)| \log n)$ depth and $\bigO(|r(x)| |x|)$ work.
\end{lemma}

\begin{proof}
	We use one iteration per non-$\infty$ entry of~$r(x)$.
	In each iteration, the smallest non-dominated entry of $x_v$ is copied to $r(x)_v$ and all entries of $x$ dominated by $x_v$ are marked as dominated.
	This yields $|r(x)|$ iterations as follows:
	\begin{enumerate}
	\item
		Initialize $r(x) \gets \bot$.
		Construct a tournament tree on the non-$\infty$ elements of $x$ and identify its leaves with their indices $v \in V$ ($\bigO(\log n)$ depth and $\bigO(|x|)$ work).

	\item\label{enum:rx-iteration}
		Find the element with the smallest node index $v$ w.r.t.\ the random node order whose corresponding leaf is not marked as \emph{discarded} ($\bigO(\log n)$ depth and $\bigO(|x|)$ work).
		Set $r(x)_v \gets x_v$.

	\item
		Mark each leaf $w$ for which $x_v \leq x_w$, including~$v$, as \emph{discarded} ($\bigO(1)$ depth and $\bigO(|x|)$ work).

	\item
		If there are non-\emph{discarded} leaves ($\bigO(\log n)$ depth and $\bigO(|x|)$ work), continue at step~\ref{enum:rx-iteration}.
	\end{enumerate}
	Note that for each $w \neq v$ for which the corresponding node is \emph{discarded,} we have $r(x)_w = \infty$.
	On the other hand, by construction we have for all $v$ for which we stored $r(x)_v=x_v$ that there is no $w \in V$ satisfying both $x_w \leq x_v$ and $w < v$.
	Thus, the computed list is indeed $r(x)$.

	The depth and work bounds follow from the above bounds on the complexities of the individual steps and by observing that in each iteration, we add a distinct index--value pair (with non-$\infty$ value) to the list that after termination equals~$r(x)$.
\end{proof}

Based on Lemmas~\ref{lem:lists-short} and~\ref{lem:compute-r(x)}, Lemma~\ref{lem:computation-cheap} establishes that w.h.p.\ each of the intermediate results can be computed efficiently.
Any such intermediate result is of the form $r^V A_\mu y$ with
\begin{equation}\label{eq:computation-cheap}
	y = (r^V A_\mu)^f P_\mu
		\underbrace{\left( r^V \left(
			\bigoplus_{\lambda=0}^\Lambda P_\lambda (r^V A_\lambda)^d P_\lambda
		\right) \right)^h r^V x^{(0)}}_{x^{(h)}},
\end{equation}
where $x^{(h)} = r^V A_H^h x^{(0)}$ is the intermediate result of $h$ iterations on~$H$, $\mu \in \{ 0, \dots, \Lambda \}$ is a level, and $(r^V A_\mu)^f P_\mu$ represents another $f$ iterations in $G$ with edge weights stretched according to level~$\mu$.
The oracle uses this to simulate the $(h+1)$-th iteration on~$H$.

\begin{lemma}\label{lem:computation-cheap}
	Suppose $x^{(0)} \in \mathcal{D}^V$ is given by $(x_v)_w = 0$ for $v = w$ and $(x_v)_w = \infty$ everywhere else ($x^{(0)}_v$~is the $v$-th unit vector).
	For arbitrary $d, f, h, \mu \in \N$ with $\mu \leq \Lambda$, suppose that $y$ is defined as in~\eqref{eq:computation-cheap}.
	Then w.h.p., $r^V y$~and $r^V A_\mu y$ can be computed using $W \in \bigO(m \log^2 n)$ work and $D \in \bigO(\log^2 n)$ depth.
\end{lemma}

\begin{proof}
	By~\eqref{eq:filter-product} and~\eqref{eq:decompose-intermediate}, we may remove the intermediate filtering steps from~\eqref{eq:computation-cheap}, obtaining
	\begin{equation}
		y = r^V A_\mu^f P_\mu \left( \bigoplus_{\lambda=0}^\Lambda P_\lambda A_\lambda^d P_\lambda \right)^h x^{(0)}
			= r^V \underbrace{A_\mu^f P_\mu ~ A_H^h x^{(0)}}_{:= y'}.
	\end{equation}
	The key observation is that\dash---since the random order of $V$ only plays a role for $r$ and we removed all intermediate applications of~$r^V$\dash---$y'$~does not depend on that order.
	Hence, we may apply Lemma~\ref{lem:lists-short} which yields that for each $v \in V$, $|y_v| = |r(y'_v)| \in \bigO(\log n)$ w.h.p.
	Condition on $|y_v| \in \bigO(\log n)$ for all $v \in V$ in the following, which happens w.h.p.\ by Lemma~\ref{lem:whp}.

	As all $(r^V y)_v = r(y_v)$ can be computed in parallel, $r^V y$ can be computed using a depth of $\bigO(\max_{v \in V} |r(y_v)| \log n) \subseteq \bigO(\log^2 n)$ and $\bigO(\sum_{v \in V} |r(y_v)| |y_v|) \subseteq \bigO(n \log^2 n)$ work by Lemma~\ref{lem:compute-r(x)}.

	Regarding the second claim, i.e., the computation of $r^V A_\mu y$, we first compute each $(A_\mu y)_v$ in parallel for all $v \in V$.
	By Lemma~\ref{lem:aggregation} and because $|y_v| \in \bigO(\log n)$, this can be done using $\bigO(\log n)$ depth and work
	\begin{equation}
		\bigO\left( \sum_{v \in V} \sum_{\substack{w \in V \\ \{v,w\} \in E}} |y_w| \log n \right)
			\subseteq \bigO\left( \sum_{\{v,w\} \in E} \log^2 n \right)
			= \bigO(m \log^2 n).
	\end{equation}
	Here we use that propagation w.r.t.~$\mathcal{D}$\dash---uniformly increasing weights\dash---requires, due to $|y_v| \in \bigO(\log n)$, no more than $\bigO(1)$ depth and $\bigO(m \log n)$ work and is thus dominated by aggregation.
	To bound the cost of computing $r^V A_\mu y$ from $A_\mu y$ observe that we have
	\begin{equation}\label{eq:computation-cheap-list-bound}
		|(A_\mu y)_v|
			\in \bigO\left( \sum_{\substack{w \in V \\ \{v,w\} \in E}} |y_w| \right).
	\end{equation}
	Hence, by Lemma~\ref{lem:compute-r(x)} and due to conditioning on $|y_v| \in \bigO(\log n)$, we can compute $r^V A_\mu y$ in parallel for all $v \in V$ using $\bigO(\log^2 n)$ depth and
	\begin{equation}
		\bigO\left( \sum_{v \in V} |(A_\mu y)_v| \log n \right)
			\stackrel{\eqref{eq:computation-cheap-list-bound}}{\subseteq}
				\bigO\left( \sum_{v \in V} \sum_{\substack{w \in V \\ \{v,w\} \in E}} |y_w| \log n \right)
			\subseteq \bigO(m \log^2 n)
	\end{equation}
	work.
	Since all operations are possible using depth $D\in \bigO(\log^2 n)$ and work $W\in \bigO(m\log^2 n)$, and we condition only on an event that occurs w.h.p., this concludes the proof.
\end{proof}

\subsection{Metric Tree Embedding in Polylogarithmic Time and Near-Linear Work}
\label{sec:frt-result}

Determining \ac{LE} lists on $H$ yields a probabilistic tree embedding of $G$ with expected stretch $\bigO(\log n)$ (Section~\ref{sec:frt-embedding}), is the result of \iac{MBF-like} algorithm (Section~\ref{sec:frt-mbf}), and each iteration of this algorithm is efficient (Theorem~\ref{thm:oracle} and Section~\ref{sec:frt-efficient}).
We assemble these pieces in Theorem~\ref{thm:embedding}, which relies on $G$ containing a suitable hop set.
Corollaries~\ref{cor:embedding} and~\ref{cor:embedding-spanner} remove this assumption by invoking known algorithms to establish this property first.
Note that Theorem~\ref{thm:embedding} serves as a blueprint yielding improved tree embedding algorithms when provided with improved hop-set constructions.

\begin{theorem}\label{thm:embedding}
	Suppose we are given the weighted incidence list of a graph $G = (V, E, \weight)$ satisfying for some $\alpha \in \R_{\geq 1}$ and $d \in \N$ that $\dist(v,w,G) \leq \alpha \dist^d(v,w,G)$ for all $v,w \in V$.
	Then, w.h.p., we can sample a tree embedding of $G$ of expected stretch $\bigO(\alpha^{\bigO(\log n)} \log n)$ with depth $\bigO(d \log^4 n) \subset \bigOT(d)$ and work $\bigO(m(d + \log n) \log^5 n) \subset \bigOT(md)$.
\end{theorem}

\begin{proof}
	By Lemma~\ref{lem:computation-cheap}, we can apply Theorem~\ref{thm:oracle} with $D \in \bigO(\log^2 n)$ and $W \in \bigO(m \log^2 n)$, showing that we can compute the \ac{LE} lists of $H$ using depth $\bigO(d \log^4 n)$ and work $\bigO(m(d + \log n) \log^5 n)$.
	As shown in~\cite{frt-tbaamtm-04}, the \ac{FRT} tree $T$ represented by these lists has expected stretch $\bigO(\log n)$ w.r.t.\ the distance metric of~$H$.
	By Theorem~\ref{thm:h}, w.h.p.\ $\dist(v,w,G) \leq \dist(v,w,H) \leq \alpha^{\bigO(\log n)} \dist(v,w,G)$ and hence
	\begin{equation}
		\dist(v,w,G)
			\leq \dist(v,w,T)
			\in \bigO\left( \alpha^{\bigO(\log n)} \log n ~ \dist(v,w,G) \right)
	\end{equation}
	in expectation (compare Definition~\ref{def:embedding}).
	Observe that by Lemma~\ref{lem:tree-explicit}, explicitly constructing the \ac{FRT} tree is possible within the stated bounds.
\end{proof}

As stated above, we require $G$ to contain a $(d, 1 / \polylog n)$-hop set with $d \in \polylog n$ in order to achieve polylogarithmic depth.
We also need to determine such a hop set using $\polylog n$ depth and near-linear work in~$m$, and that it does not significantly increase the problem size by adding too many edges.
Cohen's hop sets~\cite{c-ptnlwasusp-00} meet all these requirements, yielding the following corollary.

\begin{corollary}\label{cor:embedding}
	Given the weighted incidence list of a graph $G$ and an arbitrary constant $\epsilon > 0$, we can w.h.p.\ sample a tree embedding of expected stretch $\bigO(\log n)$ using depth $\polylog n$ and work $\bigOT(m^{1 + \epsilon})$.
\end{corollary}

\begin{proof}
	We apply the hop-set construction by Cohen~\cite{c-ptnlwasusp-00} to $G = (V, E, \weight)$ to w.h.p.\ determine an intermediate graph $G'$ with vertices $V$ and an additional $\bigOT(m^{1 + \epsilon})$ edges.
	The algorithm guarantees $\dist(v,w,G) \leq \alpha \dist^d(v,w,G')$ for $d \in \polylog n$ and $\alpha \in 1 + 1 / \polylog n$ (where the $\polylog n$ term in $\alpha$ is under our control), and has depth $\polylog n$ and work $\bigOT(m^{1 + \epsilon})$.
	Choosing $\alpha \in 1 + \bigO(1 / \log n)$ and applying Theorem~\ref{thm:embedding}, the claim follows due to Equation~\eqref{eq:h-stretch-o1}.
\end{proof}

Adding a hop set to~$G$, embedding the resulting graph in~$H$, and sampling \iac{FRT} tree on $H$ is a 3-step sequence of embeddings of~$G$.
Still, in terms of stretch, the embedding of Corollary~\ref{cor:embedding} is\dash---up to a factor in $1 + \bigo(1)$\dash---as good as directly constructing \iac{FRT} tree of~$G$:
\begin{enumerate*}
\item
	Hop sets do not stretch distances.

\item
	By Theorem~\ref{thm:h} and Equation~\eqref{eq:h-stretch-o1}, $H$~introduces a stretch of $1 + 1 / \polylog n$.

\item
	Together, this ensures that the expected stretch of the \ac{FRT} embedding w.r.t.\ $G$ is $\bigO(\log n)$.
\end{enumerate*}

It is possible to reduce the work at the expense of an increased stretch by first applying the spanner construction by Baswana and Sen~\cite{bs-sltracsswg-07}:

\begin{corollary}\label{cor:embedding-spanner}
	Suppose we are given the weighted incidence list of a graph~$G$.
	Then, for any constant $\epsilon > 0$ and any $k \in \N$, we can w.h.p.\ compute a tree embedding of $G$ of expected stretch $\bigO(k \log n)$ using depth $\polylog n$ and work $\bigOT(m + n^{1 + 1/k + \epsilon})$.
\end{corollary}

\begin{proof}
	The algorithm of Baswana and Sen~\cite{bs-sltracsswg-07} computes a $(2k - 1)$-spanner of $G = (V, E, \weight)$, i.e., a subgraph $G' = (V, E', \weight)$ satisfying for all $v,w \in V$ that $\dist(v, w, G) \leq \dist(v, w, G') \leq (2k - 1) \dist(v, w, G)$ using $\polylog n$ depth and $\bigOT(m)$ work.
	We argue in the proof of Theorem~\ref{thm:apsp2} that $|E'| \in \bigOT(n^{1 + 1/k})$ w.h.p.
	The claim follows from applying Corollary~\ref{cor:embedding} to~$G'$.
\end{proof}

\subsection{Reconstructing Paths from Virtual Edges}
\label{sec:frt-datastructure}

Given that we only deal with distances and not with paths in the \ac{FRT} construction, there is one concern:
Consider an arbitrary graph $G = (V, E, \weight)$, its augmentation with a hop set resulting in~$G'$, which is then embedded into the complete graph~$H$, and finally into \iac{FRT} tree $T = (V_T, E_T, \weight_T)$.
How can an edge $e \in E_T$ of weight $\weight_T(e)$ be mapped to a path $p$ in $G$ with $\weight(p) \leq \weight_T(H)$?
Note that this question has to be answered in polylogarithmic depth and without incurring too much memory overhead.
Our purpose is not to provide specifically tailored data structures, but we propose a three-step approach that maps edges in $T$ to paths in~$H$, edges in $H$ to paths in~$G$, and finally edges from $G'$ to paths in~$G$.

Concerning a tree edge $e \in E_T$, observe that $e$ maps back to a path $p$ of at most $\spd(H)$ hops in $H$ with $\weight_H(p) \leq 3 \weight_T(e)$ as follows.
First, to keep the notation simple, identify each tree node\dash---given as tuple $(v_i, \dots, v_j)$\dash---with its ``leading'' node $v_i \in V$;
in particular, each leaf has $i = 0$ and is identified with the node in $V$ that is mapped to it.
A leaf $v_0$ has \iac{LE} entry $(\dist(v_0, v_1, H), v_1)$ and we can trace the shortest $v_0$-$v_1$-path in $H$ based on the \ac{LE} lists (nodes locally store the predecessor of shortest paths just like in \ac{APSP}).
Moreover, $\dist(v_i, v_{i+1}, H) \leq \weight_T(v_i, v_{i+1})$, i.e., we may map the tree edge back to the path without incurring larger cost than in~$T$.
If $i > 0$, $v_i$~and $v_{i+1}$ are inner nodes.
Choose an arbitrary leaf $v_0$ that is a common descendant (this choice can, e.g., be fixed when constructing the tree from the \ac{LE} list without increasing the asymptotic bounds on depth or work).
We then can trace shortest paths from $v_0$ to $v_i$ and from $v_0$ to $v_{i+1}$ in~$H$, respectively.
The cost of their concatenation is $\dist(v_0, v_i, H) + \dist(v_0, v_{i+1}, H) \leq \beta 2^i + \beta 2^{i+1} = 3(\beta 2^i) = 3\weight_T(v,w)$ by the properties of \ac{LE} lists and the \ac{FRT} embedding.
Note that, due to the identification of each tree node with its ``leading'' graph node, paths in $T$ map to concatenable paths in~$H$.

Regarding the mapping from edges in $H$ to paths in~$G$, recall that we compute the \ac{LE} lists of $H$ by repeated application of the operations~$r^V$, $\oplus$, $P_\lambda$, and $A_\lambda$ with $0 \leq \lambda \leq \Lambda$.
Observe that~$r^V$, $\oplus$, and $P_\lambda$ discard information, i.e., distances to nodes that do not make it into the final \ac{LE} lists and therefore are irrelevant to routing.
$A_\lambda$, on the other hand, is \iac{MBF} step.
Thus, we may store the necessary information for backtracing the induced paths at each node;
specifically, we can store, for each iteration $h \in \bigO(\log^2 n)$ w.r.t.~$H$, each of the intermediate $d$ iterations in~$G$, and each $\lambda \in \bigO(\log n)$, the state vector $y$ of the form in Equation~\eqref{eq:computation-cheap} in a lookup table.
This requires $\bigOT(d)$ memory and efficiently maps edges of $H$ to $d$-hop paths in~$G$\dash---or rather to $d$-hop paths in $G'$, if we construct $H$ after augmenting $G$ to $G'$ using a hop set.

Mapping edges of $G'$ to edges in $G$ depends on the hop set.
Cohen~\cite{c-ptnlwasusp-00} does not discuss this in her article, but her hop-set edges can be efficiently mapped to paths in the original graph by a lookup table:
Hop-set edges either correspond to a shortest path in a small cluster, or to a cluster that has been explored using polylogarithmic depth.
Regarding other hop-set algorithms, we note that many techniques constructing hop set edges using depth $D$ allow for reconstruction of corresponding paths at depth $\bigO(D)$, i.e., that polylogarithmic-depth algorithms are compatible analogously to Cohen's hop sets.
For instance, this is the case for the hop-set construction by Henziger et~al.~\cite{hkn-atdacsssp-15}, which we leverage in Section~\ref{sec:distributed-henziger}.

\section{Distributed \acs{FRT} Construction}
\label{sec:distributed}

Distributed algorithms for constructing \acs{FRT}-type tree embeddings in the Congest model are covered by our framework as well.
In the following, we recap two existing algorithms~\cite{gl-nodte-14,kkmpt-edaapte-12}\dash---our framework allows to do this in a very compact way\dash---and improve upon the state of the art reducing a factor of $n^\epsilon$ in the currently best known round complexity for expected stretch $\bigO(\log n)$~\cite{gl-nodte-14} to $n^{o(1)}$.
We use the hop set of Henzinger et~al.~\cite{hkn-atdacsssp-15} instead of Cohen's~\cite{c-ptnlwasusp-00}, because it is compatible with the Congest model.
Note that replacing the hop set is straightforward since our theorems in the previous sections are formulated w.r.t.\ generic $(d, \hat\epsilon)$-hop sets.

\paragraph{The Congest Model}

We refer to Peleg~\cite{p-dclsa-00} for a formal definition of the Congest model, but briefly outline its core aspects.
The Congest model is a model of computation that captures distributed computations performed by the nodes of a graph, where communication is restricted to its edges.
Each node is initialized with a unique ID of $\bigO(\log n)$ bits, knows the IDs of its adjacent nodes along with the weights of the corresponding incident edges, and ``its'' part of the input (in our case the input is empty);
each node has to compute ``its'' part of the output (in our case, as detailed in Section~\ref{sec:frt-embedding}, its \ac{LE} list).
Computations happen in rounds, and we are interested in how many rounds it takes for an algorithm to complete.
In each round, each node does the following:
\begin{enumerate}
\item
	Perform finite, but otherwise arbitrary local computations.

\item
	Send a message of $\bigO(\log n)$ bits to each neighboring node.

\item
	Receive the messages sent by neighbors.
\end{enumerate}
Recall that, by assumption, edge weights can be encoded using $\bigO(\log n)$ bits, i.e., an index--distance pair can be encoded in a single message.

\paragraph{Overview}

Throughout this section, let $G = (V, E, \weight)$ be a weighted graph and denote, for any graph~$G$, by $A_G \in (\Rdist)^{V \times V}$ its adjacency matrix according to Equation~\eqref{eq:minplus-adjacencymatrix}.
Fix the semiring $\mathcal{S} = \mathcal{S}_{\min,+}$, the zero-preserving semimodule $\mathcal{M} = \mathcal{D}$ from Definition~\ref{def:distance-map}, as well as $r$,~$\sim$, and $x^{(0)}$ as given in Definition~\ref{def:le}.

Sections~\ref{sec:distributed-kahn} and~\ref{sec:distributed-ghaffari} briefly summarize the distributed \ac{FRT} algorithms by Kahn et~al.~\cite{kkmpt-edaapte-12} and Ghaffari and Lenzen~\cite{gl-nodte-14}, respectively.
We use these preliminaries, our machinery, and a distributed hop-set construction due to Henziger et~al.~\cite{hkn-atdacsssp-15} in Section~\ref{sec:distributed-henziger} to propose an algorithm that reduces a multiplicative overhead of $n^\epsilon$ in the round complexity of~\cite{gl-nodte-14} to $n^{o(1)}$.

\subsection{The Algorithm by Khan et~al.}
\label{sec:distributed-kahn}

In our terminology, the algorithm of Khan et~al.~\cite{kkmpt-edaapte-12} performs $\spd(G)$ iterations of the \ac{MBF-like} algorithm for collecting \ac{LE} lists implied by Definition~\ref{def:le}, i.e.,
\begin{equation}
	r^V A_G^{\spd(G)} x^{(0)}
		\stackrel{\eqref{eq:filter-product}}{=} \left( r^V A_G \right)^{\spd(G)} x^{(0)}.
\end{equation}
It does so in $\spd(G) + 1$ iterations by initializing $x^{(0)}$ as in Equation~\eqref{eq:le-x0} and iteratively computing $x^{(i+1)} := r^V A_G x^{(i)}$ until a fixpoint is reached, i.e., until $x^{(i+1)} = x^{(i)}$.
As $(r^V A_G)^i x^{(0)} = r^V A_G^i x^{(0)}$, Lemma~\ref{lem:lists-short} shows that w.h.p.\ $|x^{(i)}_v| \in \bigO(\log n)$ for all $0 \leq i \leq \spd(G)$ and all $v \in V$.
Therefore, $v \in V$ can w.h.p.\ transmit $x^{(i)}_v$ to all of its neighbors using $\bigO(\log n)$ messages, and upon reception of its neighbors' lists locally compute $x^{(i+1)}_v$.
Thus, each iteration takes $\bigO(\log n)$ rounds w.h.p., implying the round complexity of $\bigO(\spd(G) \log n)$ w.h.p.\ shown in~\cite{kkmpt-edaapte-12}.

\subsection{The Algorithm by Ghaffari and Lenzen}
\label{sec:distributed-ghaffari}

The strongest lower bound regarding the round complexity for constructing a (low-stretch) metric tree embedding of $G$ in the Congest model is $\bigOmegaT(\sqrt{n} + \diam(G))$~\cite{dhkknppw-dvhda-12,gl-nodte-14}.
If $\spd(G) \gg \max\{\diam(G), \sqrt{n}\}$, one may thus hope for a solution that runs in $\bigoT(\spd(G))$ rounds.
For any $\epsilon \in \R_{>0}$, in~\cite{gl-nodte-14} it is shown that expected stretch $\bigO(\varepsilon^{-1}\log n)$ can be achieved in $\bigOT(n^{1/2 + \epsilon} + \diam(G))$ rounds;
below we summarize this algorithm.

The strategy is to first determine the \ac{LE} lists of a constant-stretch metric embedding of (the induced submetric of) an appropriately sampled subset of~$V$. The resulting graph is called the skeleton spanner, and its LE lists are then used to jump-start the computation on the remaining graph.
When sampling the skeleton nodes in the right way, stretching non-skeleton edges analogously to Section~\ref{sec:h}, and fixing a shortest path for each pair of vertices, w.h.p.\ all of these paths contain a skeleton node within a few hops.
Ordering skeleton nodes before non-skeleton nodes w.r.t.\ the random ordering implies that each \ac{LE} list has a short prefix accounting for the local neighborhood, followed by a short suffix containing skeleton nodes only.
This is due to the fact that skeleton nodes dominate all non-skeleton nodes for which the respective shortest path passes through them.
Hence, no node has to learn information that is further away than~$d_S$, an upper bound on the number of hops when a skeleton node is encountered on a shortest path that holds w.h.p.

\paragraph*{The Graph \texorpdfstring{$H$}{H}}

In~\cite{gl-nodte-14}, $G$ is embedded into $H$ and \iac{FRT} tree is sampled on~$H$, where $H$ is derived as follows.
Abbreviate $\ell := \lceil \sqrt{n} \rceil$.
For a sufficiently large constant~$c$, sample $\lceil c \ell \log n \rceil$ nodes uniformly at random;
call this set~$S$.
Define the \emph{skeleton graph}
\begin{align}
	G_S &:=
		(S, E_S, \weight_S)\text{, where}\label{eq:gs-begin} \\
	E_S &:=
		\left\{ \{s,t\} \in \binom{S}{2} \mid \dist^\ell(s, t, G) < \infty \right\}\text{ and} \\
	\weight_S(s, t) &\mapsto
		\dist^{\ell}(s, t, G).\label{eq:gs-end}
\end{align}
Then w.h.p.\ $\dist(s, t, G_S) = \dist(s, t, G)$ for all $s,t \in S$ (Lemma~4.6 of~\cite{lp-frtcusm-13}).
For $k \in \Theta(\epsilon^{-1})$, construct a $(2k-1)$-spanner
\begin{equation}
	G'_S := (S, E'_S, \weight_S)
\end{equation}
of the skeleton graph $G_S$ that has $\bigOT(\ell^{1+1/k}) \subseteq \bigOT(n^{1/2+\epsilon})$ edges w.h.p.\ (Lemma~4.9 of~\cite{lp-frtcusm-13}).
Define
\begin{align}
	H &:=
		(V, E_H, \weight_H)\text{, where}\label{eq:ghaffari-h-begin} \\
	E_H &:=
		E'_S \cup E\text{, and} \\
	\weight_H(e) &\mapsto \begin{cases}
			\weight_S(e)        & \text{if $e \in E'_S$ and} \\
			(2k - 1) \weight(e) & \text{otherwise.}
		\end{cases}\label{eq:ghaffari-h-end}
\end{align}
By construction, $G$~embeds into $H$ with a stretch of $2k - 1$ w.h.p., i.e., $\dist(v, w, G) \leq \dist(v, w, H) \leq (2k - 1) \dist(v, w, G)$.
Computing an \iac{FRT} tree $T$ of $H$ of expected stretch $\bigO(\log n)$ thus implies that $G$~embeds into $T$ with expected stretch $\bigO(k \log n) = \bigO(\epsilon^{-1} \log n)$.

\paragraph*{\acs{FRT} Trees of \texorpdfstring{$H$}{H}}

Observe that min-hop shortest paths in $H$ contain only a single maximal subpath consisting of spanner edges, where the maximal subpaths of non-spanner edges have at most $\ell$ hops w.h.p.
This follows analogously to Lemma~\ref{lem:prefix} with $2$ levels and a sampling probability of $\bigThetaT(1 / \ell)$.
Assuming $s < v$ for all $s \in S$ and $v \in V \setminus S$\dash---we discuss this below\dash---for each $v \in V$ and each entry $(w, \dist(v, w, H))$ of its \ac{LE} list, w.h.p.\ there is a min-hop shortest $v$-$w$-path with a prefix of $\ell$ non-spanner edges followed by a shortest path in~$G'_S$.
This entails that w.h.p.
\begin{equation}\label{eq:gl-algebra}
	r^V A_H^{\spd(H)} x^{(0)}
		= r^V A_{G,2k-1}^{\ell} A_{G'_S}^{|S|} x^{(0)}
		= r^V A_{G,2k-1}^{\ell} \underbrace{\left( r^V A_{G'_S}^{|S|} x^{(0)} \right)}_{=: \bar{x}^{(0)}},
\end{equation}
where $A_{G,s}$ is $A_G$ with entries stretched by factor of $s \in \Rdist$ and we extend $A_{G'_S}$ to be a $V \times V$ matrix by setting $(A_{G'_S})_{vw} = \infty$ if $v\neq w\in V\setminus S$ and $(A_{G'_S})_{vv}=0$ for $v\in V\setminus S$.

In order to construct \iac{FRT} tree, suppose we have sampled uniform permutations of $S$ and $V \setminus S$, and a random choice of~$\beta$.
We extend the permutations to a permutation of $V$ by ruling that for all $s \in S$ and $v \in V \setminus S$, we have $s < v$, fulfilling the above assumption.
Lemma~4.9 of~\cite{gl-nodte-14} shows that the introduced dependence between the topology of $H$ and the resulting permutation on $V$ does not increase the expected stretch of the embedding beyond $\bigO(\log n)$.
The crucial advantage of this approach lies in the fact that now the \ac{LE} lists of nodes in $S$ may be used to jump-start the construction of \ac{LE} lists for $H$, in accordance with \eqref{eq:gl-algebra}.

\paragraph*{The Algorithm}

In~\cite{gl-nodte-14}, it is shown that \ac{LE} lists of $H$ can be determined fast in the Congest model as follows.
\begin{enumerate}
\item
	Some node $v_0$ starts by broadcasting $k$ and a random choice of~$\beta$, constructing \iac{BFS} tree on the fly.
	Upon receipt, each node generates a random ID of $\bigO(\log n)$ bits which is unique w.h.p.
	Querying the amount of nodes with an ID of less than some threshold via the \ac{BFS} tree, $v_0$~determines the bottom $\ell$ node IDs via binary search;
	these nodes form the set $S$ and satisfy the assumption that went into Equation~\eqref{eq:gl-algebra}.
	All of these operations can be performed in $\bigOT(\diam(G))$ rounds.

\item
	The nodes in $S$ determine $G'_S$, which is possible in $\bigOT(\diam(G) + \ell^{1 + 1/k}) \subseteq \bigOT(\diam(G) + n^{1/2 + \epsilon})$ rounds, such that all $v \in V$ learn $E'_S$ and $\weight_S$~\cite{gl-nodte-14,lp-frtcusm-13}.
	After that, $G'_S$~is global knowledge and each $v \in V$ can locally compute~$\bar{x}^{(0)}_v$.

\item\label{enum:ghaffari-final}
	Subsequently, nodes w.h.p.\ determine their component of $r^V A_{G,2k-1}^\ell \bar{x}^{(0)} = (r^V A_{G,2k-1})^\ell \bar{x}^{(0)}$ via $\ell$ \ac{MBF-like} iterations of
	\begin{equation}
		\bar{x}^{(i+1)} := r^V A_{G,2k-1} \bar{x}^{(i)}.
	\end{equation}
	Here, one exploits that for all~$i$, $|\bar{x}^{(i)}_v| \in \bigO(\log n)$ w.h.p.\ by Lemma~\ref{lem:lists-short},\footnote{%
		We apply Lemma~\ref{lem:lists-short} twice, as it requires $x \in \mathcal{D}$ to be independent of the permutation.
		First consider a computation initialized with $y^{(0)}_{vw} := 0$ if $v = w \in S$ and $y^{(0)}_{vw} := \infty$ else.
		By Lemma~\ref{lem:lists-short}, we have $|y_v^{(i)}| \in \bigO(\log n)$ w.h.p.\ for all $y^{(i)} := r^V A_{H_S}^i y^{(0)}$ and iterations $i \in \{1, \dots, |S|\}$.
		Analogously, apply Lemma~\ref{lem:lists-short} to $z^{(i)} := r^V A_{G,2k-1}^i z^{(0)}$, $i \in \{1, \dots, \ell\}$ with $z^{(0)}_{vw} := 0$ if $v = w \in V \setminus S$ and $z^{(0)}_{vw} := \infty$ else;
		this yields that $|z^{(i)}_v| \in \bigO(\log n)$ for all $v \in V$ w.h.p., too.
		As we have $x_v^{(i)} = r^V (y_v^{(j)} \oplus z_v^{(k)})$ for all $v \in V$ and appropriate $i,j,k \in \N$, we obtain $|x_v^{(i)}| \in \bigO(\log n)$ w.h.p.%
	} and thus each iteration can be performed by sending $\bigO(\log n)$ messages over each edge, i.e., in $\bigO(\log n)$ rounds;
	the entire step hence requires $\bigOT(\ell) \subseteq \bigOT(n^{1/2})$ rounds.
\end{enumerate}
Together, this w.h.p.\ implies the round complexity of $\bigOT(n^{1/2 + \epsilon} + \diam(G))$ for an embedding of expected stretch $\bigO(\epsilon^{-1}\log n)$.

\subsection{Achieving Stretch \texorpdfstring{$\bigO(\log n)$}{O(log n)} in Near-Optimal Time}
\label{sec:distributed-henziger}

The multiplicative overhead of $n^{\epsilon}$ in the round complexity is due to constructing and broadcasting the skeleton spanner~$G'_S$.
We can improve upon this by relying on hop sets, just as we do in our parallel construction.
Henziger et~al.~\cite{hkn-atdacsssp-15} show how to compute an $(n^{\bigo(1)}, \bigo(1))$-hop set of the skeleton graph in the Congest model using $n^{1/2 + \bigo(1)} + \diam(G)^{1 + \bigo(1)}$ rounds.

Our approach is similar to the one outlined in Section~\ref{sec:distributed-ghaffari}. The key difference is that we replace the use of a spanner by combining a hop set of the skeleton graph with the construction from Section~\ref{sec:h}; using the results from Section~\ref{sec:oracle}, we can then efficiently construct the LE lists on $S$ to jump-start the construction of LE lists for all nodes.

\paragraph*{The Graph \texorpdfstring{$H$}{H}}

Let~$\ell$, $c$, and the skeleton graph $G_S = (S, E_S, \weight_S)$ be defined as in Section~\ref{sec:distributed-ghaffari} and Equations~\eqref{eq:gs-begin}--\eqref{eq:gs-end}, w.h.p.\ yielding $\dist(s, t, G_S) = \dist(s, t, G)$ for all $s,t \in S$.
Suppose for all $s,t\in S$, we know approximate weights $\weight'_S(s,t)$ with
\begin{equation*}
\dist(s,t,G)\leq \weight'_S(s,t)\in (1+o(1))\weight_S(s,t)
\end{equation*}
\dash---our algorithm has to rely on an approximation to meet the stated round complexity\dash---and add an $(n^{\bigo(1)}, \bigo(1/\log n))$-hop set to $G_S$ using the construction of Henzinger et~al.~\cite{hkn-atdacsssp-15}.
Together, this results in a graph
\begin{equation}
	G'_S := (S, E'_S, \weight'_S),
\end{equation}
where $E'_S$ contains the skeleton edges $E_S$ and some additional edges, and w.h.p.\ it holds for all $s,t \in S$ that
\begin{equation}
	\dist(s, t, G_S)
		\leq \dist^d(s, t, G'_S)
		\in (1 + o(1/\log n)) \dist(v, w, G_S)
\end{equation}
for some $d \in n^{\bigo(1)}$ and $\dist(v,w,G)\leq \dist(v,w,G_S)\in (1+o(1))\dist(v,w,G)$.
Next, embed $G'_S$ into $H_S$ as in Section~\ref{sec:h}, yielding node and edge levels $\level(e) \in \{0, \dots, \Lambda\}$:
\begin{align}
	H_S &:=
		\left( S, \binom{S}{2}, \weight_{H_S} \right)\text{ with} \\
	\weight_{H_S}(\{s,t\}) &\mapsto
		(1 + \hat\epsilon)^{\Lambda - \level(s,t)} \dist^d(s, t, G'_S)
\end{align}
with $d$ as above, $\hat\epsilon \in \bigo(1 / \log n)$.
By Theorem~\ref{thm:h}, w.h.p.\ we have that $\spd(G) \in \bigO(\log^2 n)$ and for all $s,t\in S$ that
\begin{equation}
	\dist(s,t,G)\leq \dist(s,t,G_S)\leq \dist(s,t,H_S) \in (1+o(1)) \dist(s, t, G_S),
\end{equation}
which is bounded from above by $\alpha \dist(s,t,G)$ for some $\alpha \in 1 + \bigo(1)$.
Analogously to Equations~\eqref{eq:ghaffari-h-begin}--\eqref{eq:ghaffari-h-end}, define
\begin{align}
	H &:=
		(V, E_H, \weight_H),\text{ where} \\
	E_H &:=
		E \cup \binom{S}{2}\text{, and} \\
	\weight_H(e) &\mapsto \begin{cases}
			\weight_{H_S}(e)    & \text{if $e \in \binom{S}{2}$ and} \\
			\alpha \weight_G(e) & \text{otherwise.}
		\end{cases}
\end{align}
By construction, we thus have
\begin{equation}
	\forall v,w \in V\colon\quad
		\dist(v,w,G)
			\leq \dist(v,w,H)
			\leq \alpha \dist(v,w,G)
			\in (1 + \bigo(1)) \dist(v,w,G)
\end{equation}
w.h.p.

\paragraph*{\acs{FRT} Trees of \texorpdfstring{$H$}{H}}

Analogously to Section~\ref{sec:distributed-ghaffari}, assume that the node IDs of $S$ are ordered before those of $V \setminus S$;
then min-hop shortest paths in $H$ contain a single maximal subpath of edges in~$E_{H_S}$.
To determine the \ac{LE} lists for~$H$, we must hence compute
\begin{equation}
	r^V A_H^{\spd(H)} x^{(0)}
		= \left( r^V A_{G,\alpha} \right)^{\ell}
			\underbrace{\left( r^V A_{H_S} \right)^{\spd(H_S)} x^{(0)}}_{=: \bar{x}^{(0)}},
\end{equation}
where $A_{G,\alpha}$ is given by multiplying each entry of $A_G$ by the abovementioned factor of~$\alpha$, and $A_{H_S}$ is extended to an adjacency matrix on the node set $V$ as in Section~\ref{sec:distributed-ghaffari}.

\paragraph*{The Algorithm}

We determine the \ac{LE} lists of $H$ as follows, adapting the approach from~\cite{gl-nodte-14} outlined in Section~\ref{sec:distributed-ghaffari}.
\begin{enumerate}
\item
	A node $v_0$ starts the computation by broadcasting a random choice of~$\beta$.
	The broadcast is used to construct \iac{BFS} tree, nodes generate distinct random IDs of $\bigO(\log n)$ bits w.h.p., and $v_0$ figures out the ID threshold of the bottom $c\ell$ nodes $S$ w.r.t.\ the induced random ordering.
	This can be done in $\bigOT(\diam(G))$ rounds.

\item
	Each skeleton nodes $s\in S$ computes $\weight'_S(s,t)$ as above for all $t\in S$, using the $(1 + 1 / \log^2 n)$-approximate $(S,\ell,|S|)$-detection
	algorithm given in~\cite{lp-fpdea-15}. This takes $\bigOT(\ell+\ell)=\bigOT(n^{1/2})$ rounds.
	
\item
	Run the algorithm of Henzinger et~al.~\cite{hkn-atdacsssp-15} to compute an $(n^{\bigo(1)}, \bigo(1))$-hop set of~$G_S'$\dash---in the sense that nodes in $S$ learn their incident weighted edges.
	This takes $n^{1/2 + \bigo(1)} + \diam(G)^{1 + \bigo(1)}$ rounds.

\item
	Next, we (implicitly) construct~$H_S$.
	To this end, nodes in $S$ locally determine their level and broadcast it over the \ac{BFS} tree, which takes $\bigO(|S| + \diam(G)) \subset \bigOT(\sqrt{n} + \diam(G))$ rounds;
	thus, $s \in S$ knows the level of $\{s,t\} \in E_{H_S}$ for each $t \in S$.

\item
	To determine~$\bar{x}^{(0)}$, we follow the same strategy as in Theorem~\ref{thm:oracle}, i.e., we simulate matrix-vector multiplication with $A_{H_S}$ via matrix-vector multiplications with $A_{G'_S}$.
	Hence, it suffices to show that we can efficiently perform a matrix-vector multiplication $A_{G'_S} x$ for any $x$ that may occur during the computation\dash---applying $r^V$ is a local operation and thus free\dash---assuming each node $v \in V$ knows $x_v$ and its row of the matrix.

	Since multiplications with $A_{G'_S}$ only affects lists at skeleton nodes, this can be done by local computations once all nodes know $x_s$ for each $s \in S$.
	As before, $|x_s| \in \bigO(\log n)$ w.h.p., so $\sum_{s \in S} |x_s| \in \bigO(|S| \log n)\subset \bigOT(\sqrt{n})$ w.h.p.
	We broadcast these lists over the \ac{BFS} tree of~$G$, taking $\bigOT(\sqrt{n} + \diam(G))$ rounds per matrix-vector multiplication.
	Due to $\spd(H_S) \in \bigOT(\log^2 n)$ by Theorem~\ref{thm:h}, this results in a round complexity of $\bigOT(n^{1/2 + \bigo(1)} + \diam(G)^{1 + \bigo(1)})$.

\item
	Applying $r^V A_{G,\alpha}^{\ell}$ is analogous to step~\ref{enum:ghaffari-final} in Section~\ref{sec:distributed-ghaffari} and takes $\bigOT(\ell) \subseteq \bigOT(n^{1/2})$ rounds.
\end{enumerate}
Altogether, this yields a round complexity of $n^{1/2 + \bigo(1)} + \diam(G)^{1 + \bigo(1)}$.
Combining this result with the algorithm by Khan et~al.~\cite{kkmpt-edaapte-12}, which terminates quickly if $\spd(G)$ is small, yields the following result.

\begin{theorem}\label{thm:distributed}
	There is a randomized distributed algorithm w.h.p.\ computing a metric tree embedding of expected stretch $\bigO(\log n)$ in $\min\{ (\sqrt{n} + \diam(G)) n^{o(1)} , \bigOT(\spd(G)) \}$ rounds of the Congest model.
\end{theorem}

\section{\texorpdfstring{$k$}{k}-Median}
\label{sec:kmedian}

In this section, we turn to the $k$-median problem, an application considered by Blelloch et~al.~\cite{bgt-pptekmbabnd-12} and show how their results are improved by applying our techniques.
The contribution is that we work on a weighted graph $G$ that only implicitly provides the distance metric $\dist(\cdot,\cdot,G)$;
Blelloch et~al.\ require a metric providing constant-time query access.
Our solution is more general, as any finite metric defines a complete graph of \ac{SPD}~$1$, whereas determining exact distances in graphs requires $\bigOmega(\spd(G))$ depth.
The use of hop sets, however, restricts us to polynomially bounded edge-weight ratios.

\begin{definition}[$k$-Median]
	In the \emph{$k$-median problem} we are given a weighted graph $G = (V, E, \weight)$ and an integer $k \in \N$.
	The task is to determine $F \subseteq V$ with $|F| \leq k$ that minimizes
	\begin{equation}
		\sum_{v \in V} \dist(v, F, G),
	\end{equation}
	where $\dist(v, F, G) := \min\{ \dist(v, f, G) \mid f \in F \}$ is the distance of $v$ to the closest member of~$F$.
\end{definition}

Blelloch et~al.~\cite{bgt-pptekmbabnd-12} solve the following problem:
Given a metric with constant-time query access, determine an expected $\bigO(\log k)$-approximation of $k$-median using $\bigO(\log^2 n)$ depth and $\bigOT(nk + k^3)$ work for $k \geq \log n$;
the special case of $k < \log n$ admits an $\bigOT(n)$-work solution of the same depth~\cite{bt-paaflp-10}.
Below, we show how to determine an expected $\bigO(\log k)$-approximation of $k$-median on a weighted graph, using $\polylog n$ depth and $\bigOT(m^{1+\epsilon} + k^3)$ work.

The algorithm of Blelloch et~al.~\cite{bgt-pptekmbabnd-12} essentially comprises three steps:
\begin{enumerate}
\item\label{enum:kmedian1}
	Use a parallel version of a sampling technique due to Mettu and Plaxton~\cite{mp-otbac-04}.
	It samples candidates~$Q$, such that $|Q| \in \bigO(k)$ and there is $F \subseteq Q$ that $\bigO(1)$-approximates $k$-median.

\item\label{enum:kmedian2}
	Sample \iac{FRT} tree regarding the submetric spanned by~$Q$.
	Normalize the tree to a binary tree (required by the next step);
	this is possible without incurring too much overhead w.r.t.\ the depth of the tree~\cite{bgt-pptekmbabnd-12}.

\item
	Run an $\bigO(k^3)$-work dynamic programming algorithm to solve the tree instance optimally without using any Steiner nodes.
	This yields an $\bigO(\log k)$-approximate solution on the original metric due to the expected stretch from the \ac{FRT} embedding.
\end{enumerate}
We keep the overall structure but modify steps~\ref{enum:kmedian1}--\ref{enum:kmedian2}, resulting in the following algorithm:
\begin{enumerate}
\item
	The sampling step generates $\bigO(k)$ candidate points~$Q$.

	It requires $\bigO(\log \frac{n}{k})$ iterations and maintains a candidate set $U$ that initially contains all points.
	In each iteration, $\bigO(\log n)$ candidates $S$ are sampled and a constant fraction of vertices in~$U$, those closest to~$S$, is removed~\cite{bgt-pptekmbabnd-12}.

	They key to adapting this procedure to graphs lies in efficiently determining $\dist(u, S, G)$ for all $u \in U$ (this would be trivial with constant-time query access to the metric).
	We achieve this by sampling after embedding in $H$ from Section~\ref{sec:h} which only costs a factor of $(1 + \bigo(1))$ in approximation, regardless of~$k$.
	By Theorem~\ref{thm:h}, we only require $\bigO(\log^2 n)$ iterations of the \ac{MBF-like} algorithm from Example~\ref{ex:fire} (for $d = \infty$) to determine each node's distance to the closest vertex in~$S$ w.h.p.
	Hence, we require polylogarithmic depth and $\bigOT(m^{1 + \epsilon})$ work for this step.

	Since $|U|$ decreases by a constant factor in each iteration and we have $\bigO(\log n)$ iterations, we require a total of $\bigOT(m^{1 + \epsilon})$ work and polylogarithmic depth, including the costs for determining Cohen's hop set~\cite{c-ptnlwasusp-00}.

\item
	Sample \iac{FRT} tree on the submetric spanned by~$Q$.

	To compute the embedding only on $Q$ set $x^{(0)}_{vv} = 0$ if $v \in Q$ and $x^{(0)}_{vw} = \infty$ everywhere else.
	Consider only the \ac{LE} lists of nodes in $Q$ when constructing the tree.

	As we are limited to polynomially bounded edge-weight ratios, our \ac{FRT} trees have logarithmic depth.
	We normalize to a binary tree using the same technique as Blelloch et~al.~\cite{bgt-pptekmbabnd-12}.

\item
	The $\bigOT(k^3)$-work polylogarithmic-depth dynamic-programming algorithm of Blelloch et~al.\ can be applied without modification.
\end{enumerate}
W.h.p., we arrive at an expected $\bigO(\log k)$-approximation of $k$-median:
\begin{theorem}
	For any fixed constant $\epsilon > 0$, w.h.p., an expected $\bigO(\log k)$-approximation to $k$-median on a weighted graph can be computed using $\polylog n$ depth and $\bigOT(m^{1+\epsilon} + k^3)$ work.
\end{theorem}

\section{Buy-at-Bulk Network Design}
\label{sec:buyatbulk}

In this section, we reduce the work of the approximation algorithm for the buy-at-bulk network design problem given by Blelloch et~al.~\cite{bgt-pptekmbabnd-12} that requires $\bigO(n^3 \log n)$ work and $\bigO(\log^2 n)$ depth w.h.p., while providing the same asymptotic approximation guarantees.
Blelloch et~al.\ transform the input graph $G$ into a metric which allows constant-time query access on which they sample \iac{FRT} embedding, hence their work is dominated by solving \ac{APSP}.

Replacing the \ac{APSP} routine in the algorithm Blelloch et~al.\ with our $(1 + \epsilon)$-approximate metric from Theorem~\ref{thm:apsp}\dash---and keeping the rest of the algorithm in place\dash---directly reduces the work to $\bigOT(n^2)$ while incurring $\polylog n$ depth.
However, using our result from Section~\ref{sec:frt} to sample \iac{FRT} without the detour over the metric, we can guarantee a stronger work bound of $\bigOT(\min\{ m^{1+\epsilon} + kn, n^2 \}) \subseteq \bigOT(n^2)$, which achieves the same depth.
The use of hop sets, however, restricts us to polynomially bounded edge ratios (or our solution loses efficiency).

\begin{definition}[Buy-at-Bulk Network Design]
	In the \emph{buy-at-bulk network design problem,} one is given a weighted graph $G = (V, E, \weight)$, demands $(s_i, t_i, d_i) \in V \times V \times \R_{>0}$ for $1 \leq i \leq k$, and a finite set of cable types $(u_i, c_i) \in \R_{>0} \times \R_{>0}$, $1 \leq i \leq \ell$, where the cable of type $i$ incurs costs $c_i \weight(e)$ when purchased for edge~$e$ (multiple cables of the same type can be bought for an edge).
	The goal is to find an assignment of cable types and multiplicities to edges minimizing the total cost, such that the resulting edge capacities allow to simultaneously route $d_i$ units of (distinct) flow from $s_i$ to $t_i$ for all $1 \leq i \leq k$.
\end{definition}

Andrews showed that the buy-at-bulk network design problem is hard to approximate better than with factor $\log^{1/2 - \bigo(1)} n$~\cite{a-hbbnd-04}.
Blelloch et~al.~\cite{bgt-pptekmbabnd-12} give an expected $\bigO(\log n)$-approximation w.h.p.\ using $\polylog n$ depth and $\bigO(n^3\log n)$ work for the buy-at-bulk network design problem.
It is a straightforward parallelization of the algorithm by Awerbuch and Azar~\cite{aa-bbnd-97}. 
Our tools allow for a more work-efficient parallelization of this algorithm, as the work of the implementation by Blelloch et~al.\ is dominated by solving \ac{APSP} to determine the distance metric of the graph;
we achieve the same approximation guarantee as Blelloch et~al.\ using $\polylog n$ depth and $\bigOT(n^2)$ work.
We propose the following modification of the approach of Blelloch et~al.
\begin{enumerate}
\item
	Metrically embed $G$ into a tree $T = (V_T, E_T, \weight_T)$ with expected stretch $\bigO(\log n)$.
	As the objective is linear in the edge weights, an optimal solution in $G$ induces a solution in $T$ whose expected cost is by at most a factor $\bigO(\log n)$ larger.

\item
	$\bigO(1)$-approximate on~$T$:
	For $e \in E_T$, pick the cable of type $i$ that minimizes $c_i \lceil d_e / u_i \rceil$, where $d_e$ is the accumulated flow on~$e$, see~\cite{bgt-pptekmbabnd-12}).

\item
	Map the tree solution back to~$G$, increasing the cost by a factor of $\bigO(1)$.
\end{enumerate}
Combining these steps yields an $\bigO(\log n)$-approximation.
Using Corollary~\ref{cor:embedding}, the first step has $\polylog n$ depth and $\bigOT(m^{1+\epsilon})$ work;
for the second step, Blelloch et~al.\ discuss an algorithm of $\polylog n$ depth and $\bigOT(n + k)$ work.

Concerning the third step, recall that each tree edge $\{v,w\}$ maps back to a path $p$ of at most $\spd(H)$ hops in $H$ with $\weight(p) \leq 3 \weight_T(v,w)$ as argued in Section~\ref{sec:frt-datastructure}.
Using this observation, we can map the solution on $T$ back to one in $H$ whose cost is at most by factor $3$ larger.
Assuming suitable data structures are used, this operation has depth $\polylog n$ and requires $\bigOT(\min\{k,n\})$ work w.h.p., where we exploit that $\spd(H) \in \bigO(\log^2 n)$ w.h.p.\ by Theorem~\ref{thm:h} and the fact that $T$ has depth $\bigO(\log n)$, implying that the number of edges in $T$ with non-zero flow is bounded by $\bigO(\min\{k,n\} \log n)$.

Finally, we map back from $H$ to $G'$ ($G$~augmented with hop set edges) and then to~$G$.
This can be handled with depth $\polylog n$ and $\bigOT(n)$ work for a single edge in $H$ because edges in $H$ and hop-set edges in $G'$ correspond to polylogarithmically many edges in $G'$ and at most $n$ edges in~$G$, respectively.
The specifics depend on the hop set and, again, we assume that suitable data structures are in place, see Section~\ref{sec:frt-datastructure}.
Since we deal with $\bigOT(\min\{k,n\})$ edges in~$H$, mapping back the edges yields $\bigOT(\min\{kn,n^2\})$ work in total.
Together with the computation of the hop set, we have $\bigOT(\min\{ m^{1+\epsilon}, n^2 \} + \min\{kn,n^2\}) = \bigOT(\min\{ m^{1+\epsilon} + kn, n^2 \}) \subseteq \bigOT(n^2)$ work (the work to determine Cohen's hop set~\cite{c-ptnlwasusp-00} is bounded by $\bigOT(n^2)$ due to the same reasoning as in the proof of Theorem~\ref{thm:apsp}).

\begin{theorem}
	For any constant $\epsilon > 0$, w.h.p., an expected $\bigO(\log n)$-approximation to the buy-at-bulk network design problem can be computed using $\polylog n$ depth and $\bigOT(\min\{ m^{1+\epsilon} + kn, n^2 \}) \subseteq \bigOT(n^2)$ work.
\end{theorem}

\section{Conclusion}
\label{sec:conclusion}

In this work, we show how to sample from an \acs{FRT}-style distribution of metric tree embeddings at low depth and near-optimal work, provided that the maximum ratio between edge weights is polynomially bounded.
While we consider the polylogarithmic factors too large for our algorithm to be of practical interest, this result motivates the search for solutions that achieve low depth despite having work comparable to the currently best known sequential bound of $\bigO(m \log^3 n)$~\cite{ms-fckrpsg-09}.
Concretely, better hop-set constructions could readily be plugged into our machinery to yield improved bounds, and one may seek to reduce the number of logarithmic factors incurred by the remaining construction.

Our second main contribution is an algebraic interpretation of \ac{MBF-like} algorithms, reducing the task of devising and analyzing such algorithms to the following recipe:
\begin{enumerate}
\item
	Pick a suitable semiring $\mathcal{S}$ and semimodule $\mathcal{M}$ over~$\mathcal{S}$.

\item
	Choose a filter $r$ and initial values $x^{(0)} \in \mathcal{M}^V$ so that $r^V A^h x^{(0)}$ is the desired output.

\item
	Verify that $r$ induces a congruence relation on~$\mathcal{M}$.

\item
	Leverage (repeated use of) $r^V$ to ensure that iterations can be implemented efficiently.
\end{enumerate}
As can be seen by example of our metric tree embedding algorithm, further steps may be required to control the number of iterations~$h$;
concretely, we provide an embedding into a complete graph of small \ac{SPD} and an oracle allowing for efficient \ac{MBF-like} queries.
Nevertheless, we believe that our framework unifies and simplifies the interpretation and analysis of \ac{MBF-like} algorithms, as illustrated by the examples listed in Sections~\ref{sec:examples} and the discussion of distributed tree embeddings in Section~\ref{sec:distributed}.
Therefore, we hope that our framework will be of use in the design of further efficient \ac{MBF-like} algorithms in the future.


\bibliographystyle{abbrv}
\bibliography{bibliography}


\begin{appendix}
	\section{Algebraic Foundations}
\label{app:algebra}

For the sake of self-containment and unambiguousness, we give the algebraic definitions required in this paper as well as a standard result.
Definitions~\ref{def:semigroup}, \ref{def:semiring}, and~\ref{def:semimodule} are slightly adapted from Chapters 1 and 5 of~\cite{hw-satacs-98}.
In this section, we refer to the neutral elements of addition and multiplication as $0$ and~$1$.
Note, however, that in the min-plus semiring $\mathcal{S}_{\min,+}$ the neutral element of ``addition''~($\min$) is~$\infty$ and that of ``multiplication''~($+$) is~$0$.

\begin{definition}[Semigroup]\label{def:semigroup}
	Let $M \neq \emptyset$ be a set and $\circ\colon M \times M \to M$ a binary operation.
	$(M, \circ)$ is a \emph{semigroup} if and only if $\circ$ is associative, i.e.,
	\begin{equation}
		\forall x,y,z \in M\colon \quad x \circ (y \circ z) = (x \circ y) \circ z.
	\end{equation}
	A semigroup $(M, \circ)$ is \emph{commutative} if and only if
	\begin{equation}
		\forall x,y \in M\colon \quad x \circ y = y \circ x.
	\end{equation}
	$e \in M$ is a \emph{neutral element} of $(M, \circ)$ if and only if
	\begin{equation}
		\forall x \in M\colon \quad e \circ x = x \circ e = x.
	\end{equation}
\end{definition}

Some authors do not require semirings to have neutral elements or an annihilating~$0$.
We, however, need them and work on semirings\dash---mostly on $\mathcal{S}_{\min,+}$, $\mathcal{S}_{\max,\min}$, and $\mathcal{P}_{\min,+}$\dash---which provide them, anyway.

\begin{definition}[Semiring]\label{def:semiring}
	Let $M \neq \emptyset$ be a set, and $\oplus,\odot\colon M \times M \to M$ binary operations.
	Then $(M, \oplus, \odot)$ is a \emph{semiring} if and only if
	\begin{enumerate}
	\item
		$(M, \oplus)$ is a commutative semigroup with neutral element~$0$,

	\item
		$(M, \odot)$ is a semigroup with neutral element~$1$,

	\item
		the left- and right-distributive laws hold:
		\begin{align}
			\forall x,y,z \in M\colon \quad x \odot (y \oplus z) &= (x \odot y) \oplus (x \odot z), \label{eq:dist-left} \\
			\forall x,y,z \in M\colon \quad (y \oplus z) \odot x &= (y \odot x) \oplus (z \odot x)\text{, and} \label{eq:dist-right}
		\end{align}

	\item
		$0$~annihilates w.r.t.\ $\odot$:
		\begin{equation}
			\forall x \in M\colon \quad 0 \odot x = x \odot 0 = 0.
		\end{equation}
\end{enumerate}
\end{definition}

\begin{definition}[Semimodule]\label{def:semimodule}
	Let $\mathcal{S} = (S, \oplus, \odot)$ be a semiring.
	$\mathcal{M} = (M, \oplus, \odot)$ with binary operations $\oplus\colon M \times M \to M$ and $\odot\colon S \times M \to M$ is a \emph{semimodule over~$\mathcal{S}$} if and only if
	\begin{enumerate}
	\item
		$(M, \oplus)$ is a semigroup and

	\item
		for all $s,t \in S$ and all $x,y \in M$:
		\begin{align}
			1 \odot x            &= x, \label{eq:semimodule-begin} \\
			s \odot (x \oplus y) &= (s \odot x) \oplus (s \odot y), \\
			(s \oplus t) \odot x &= (s \odot x) \oplus (t \odot x)\text{, and} \label{eq:dist-left-module} \\
			(s \odot t) \odot x  &= s \odot (t \odot x). \label{eq:semimodule-end}
		\end{align}
	\end{enumerate}

	$\mathcal{M}$ is \emph{zero-preserving} if and only if
	\begin{enumerate}
	\item
		$(M, \oplus)$ has the neutral element $0$ and

	\item
		$0 \in S$ is an annihilator for~$\odot$:
		\begin{equation}
			\forall x \in M\colon \quad 0 \odot x = 0.
		\end{equation}
\end{enumerate}
\end{definition}

A frequently used semimodule over the semiring $\mathcal{S}$ is $\mathcal{S}^k$ with coordinate-wise addition, i.e., $k$-dimensional vectors over~$\mathcal{S}$.
Note that $\mathcal{S} = \mathcal{S}^1$ always is a semimodule over itself.

\begin{lemma}\label{lem:semiring-to-the-k}
	Let $\mathcal{S} = (S, \oplus, \odot)$ be a semiring and $k \in \N$ an integer.
	Then $\mathcal{S}^k := (S^k, \oplus, \odot)$ with, for all $s \in \mathcal{S}$, $x,y \in \mathcal{S}^k$, and $1 \leq i \leq k$,
	\begin{align}
		(x \oplus y)_i &:= x_i \oplus y_i\text{ and} \\
		(s \odot x)_i  &:= s \odot x_i
	\end{align}
	is a zero-preserving semimodule over $\mathcal{S}$ with zero $(0, \dots, 0)$.
\end{lemma}

\begin{proof}
	We check the conditions of Definition~\ref{def:semimodule} one by one.
	Throughout the proof, let $s,t \in \mathcal{S}$ and $x,y \in \mathcal{S}^k$ be arbitrary.
	\begin{enumerate}
	\item
		$(S^k, \oplus)$ is a semigroup because $(S, \oplus)$ is.

	\item
		Equations~\eqref{eq:semimodule-begin}--\eqref{eq:semimodule-end} hold due to
		\begin{gather}
			(1 \odot x)_i
				= 1 \odot x_i
				= x_i, \\
			(s \odot (x \oplus y))_i
				= s \odot (x_i \oplus y_i)
				= (s \odot x_i) \oplus (s \odot y_i)
				= ((s \odot x) \oplus (s \odot y))_i, \\
			((s \oplus t) \odot x)_i
				= (s \oplus t) \odot x_i
				= (s \odot x_i) \oplus (t \odot x_i)
				= ((s \odot x) \oplus (t \odot x))_i\text{, and} \\
			((s \odot t) \odot x)_i
				= (s \odot t) \odot x_i
				= s \odot (t \odot x_i)
				= (s \odot (t \odot x))_i.
		\end{gather}

	\item
		$(0, \dots, 0)$ is the neutral element of $(S^k, \oplus)$ because $0$ is the neutral element of $(S, \oplus)$.

	\item
		$0$~is an annihilator for~$\odot$:
		\begin{equation}
			(0 \odot x)_i = 0 \odot x_i = 0. \qedhere
		\end{equation}
	\end{enumerate}
\end{proof}

	\section{Deferred Proofs}
\label{app:proofs}

This appendix contains the proofs deferred from Section~\ref{sec:examples} for the sake of presentation.

\paragraph*{Proof of Lemma~\ref{lem:minplus-xh}}

\begin{proof}
	The claim trivially holds for $h = 0$.
	As induction hypothesis, suppose the claim holds for $h \in \N$.
	We obtain
	\begin{align}
		x^{(h+1)}_{vw}
			&= (A x^{(h)})_{vw} \\
			&= \left( \bigoplus_{u \in V} a_{vu} \odot x^{(h)}_u \right)_w \\
			&= \bigoplus_{u \in V} a_{vu} \odot x^{(h)}_{uw} \\
			&= \min_{u \in V} \left\{ a_{vu} + x^{(h)}_{uw} \right\} \\
			&= \min \left\{ \weight(v,u) + \dist^h(u,w,G) \mid \{v,u\} \in E \right\} \cup \left\{ 0 + \dist^h(v,w,G) \right\},
	\end{align}
	i.e., exactly the definition of $\dist^{h+1}(v,w,G)$, as claimed.
\end{proof}

\paragraph*{Proof for Example~\ref{ex:source-detection}}

\begin{proof}
	Let $s \in \mathcal{S}_{\min,+}$ be arbitrary and let $x,x',y,y' \in \mathcal{D}$ be such that $x \sim x'$ and $y \sim y'$, where $x \sim y :\Leftrightarrow r(x) = r(y)$.
	By Lemma~\ref{lem:congruence-by-r}, it suffices to show
	\begin{enumerate*}
	\item
		that $r^2 = r$,

	\item
		that $r(sx) = r(sx')$, and

	\item
		that $r(x \oplus y) = r(x' \oplus y')$.
	\end{enumerate*}

	We show the claims one by one.
	First observe that $r(x)_v = \infty$ for all $v \in V \setminus S$, hence w.l.o.g.\ assume $v \in S$ in the following.
	\begin{enumerate*}
	\item
		$r(x)$~has at most $k$ entries, each at most~$d$, so $r(r(x)) = r(x)$ by~\eqref{eq:source-detection-filter}.

	\item
		Since multiplication with $s$ uniformly increases the non-$\infty$ entries of $x$ and~$x'$, it does not affect their ordering w.r.t.~\eqref{eq:source-detection-filter}.
		As the $k$ smallest $S$-entries of $x$ and $x'$ w.r.t.~\eqref{eq:source-detection-filter} are identical, so are those of $sx$ and~$sx'$.
		Some entry $(sx)_v$ may become larger than~$d$, but that happens for $(sx)'_v$ as well, hence $r(sx) = r(sx')$.

	\item
		We have $r(x \oplus y)_v \leq d$ only if $(x \oplus y)_v = \min\{ x_v, y_v \} \leq d$ is among the $k$ smallest entries of $(x \oplus y)$ w.r.t.~\eqref{eq:source-detection-filter}.
		If that is the case, there are no $k$ entries smaller than $r(x \oplus y)_v$ in $x$ or in~$y$.
		Hence, these entries exist in $x'$ and $y'$ as well, form the $k$ smallest entries of $(x' \oplus y')$, and $r(x \oplus y)_v = r(x' \oplus y')_v$ follows.
	\end{enumerate*}
\end{proof}

\paragraph*{Proof of Lemma~\ref{lem:maxmin}}

\begin{proof}
	We check each of the requirements of Definition~\ref{def:semiring} in Appendix~\ref{app:algebra}.
	Throughout the proof, let $x,y,z \in \Rdist$ be arbitrary.
	\begin{enumerate}
	\item
		$(\Rdist, \max)$ is a commutative semigroup because $\max$ is associative and commutative.
		Since $0$ is the minimum of $\Rdist$, it is the neutral element of $(\Rdist, \max)$.

	\item
		$(\Rdist, \min)$ is a semigroup because $\min$ is associative.
		Like above, $\infty$~is its neutral element because it is the maximum of $\Rdist$.

	\item
		Regarding the left- and right-distributive laws in Equations~\eqref{eq:dist-left}--\eqref{eq:dist-right}, a case distinction between the cases
		\begin{enumerate*}
		\item
			$x \leq y \leq z$,

		\item
			$y \leq x \leq z$, and

		\item
			$y \leq z \leq x$
		\end{enumerate*}
		is exhaustive due to the commutativity of $\min$ and $\max$ and reveals that
		\begin{equation}
			\min\{ x, \max\{y, z\} \} = \max\{ \min\{x, y\}, \min\{x, z\} \},
		\end{equation}
		i.e., that the left-distributive law holds.
		Since $\min$ is commutative,
		\begin{equation}
			\min\{ \max\{y, z\}, x \} = \max\{ \min\{y, x\}, \min\{z, x\} \}
		\end{equation}
		immediately follows;
		hence $\mathcal{S}_{\max,\min}$ fulfills both distributive laws.

	\item
		$0$~is an annihilator for $\min$ because
		\begin{equation}
			\min\{0,x\} = \min\{x,0\} = 0.
		\end{equation}
	\end{enumerate}
	Together, it follows that $\mathcal{S}_{\max,\min}$ is a semiring as claimed.
\end{proof}

\paragraph*{Proof of Lemma~\ref{lem:maxmin-xh}}

\begin{proof}
	The claim holds for $h = 0$ by Equation~\eqref{eq:maxmin-x0}.
	As induction hypothesis, suppose the claim holds for some $h \in \N$.
	We obtain
	\begin{equation}
		x^{(h+1)}_v
			\stackrel{\eqref{eq:maxmin-xh}}{=} \left( A x^{(h)} \right)_v
			= \bigoplus_{w \in V} a_{vw} \odot x^{(h)}_w
			\stackrel{\eqref{eq:maxmin-adjacencymatrix}}{=}
				\underbrace{\infty \odot x^{(h)}_v}_{x^{(h)}_v}
				\oplus \bigoplus_{\{v,w\} \in E} \weight(v,w) \odot x^{(h)}_w.
	\end{equation}
	Recall that $\oplus$ in $\mathcal W$ is the element-wise maximum by Corollary~\ref{cor:maxmin-modules}.
	Hence, we have
	\begin{equation}
		x^{(h+1)}_{vu}
			= \max\left\{ x^{(h)}_{vu} \right\}
				\cup \left\{ \min\{ \weight(v,w), x^{(h)}_{wu} \} \mid \{v,w\} \in E \right\}
	\end{equation}
	and the induction hypothesis yields
	\begin{equation}
		x^{(h+1)}_{vu}
			= \max\left\{ \width^h(v, u, G) \right\}
				\cup \left\{ \min\{ \weight(v,w), \width^h(w, u, G) \} \mid \{v,w\} \in E \right\},
	\end{equation}
	which is exactly $\width^{h+1}(v, u, G)$.
\end{proof}

\paragraph*{Proof of Lemma~\ref{lem:allpaths}}

\begin{proof}
	We check the requirements of Definition~\ref{def:semiring} in Appendix~\ref{app:algebra} step by step.
	Throughout the proof, let $\pi \in P$ and $x,y,z \in \mathcal{P}_{\min,+}$ be arbitrary.
	\begin{enumerate}
	\item
		We first show that $((\Rdist)^P, \oplus)$ is a commutative semigroup with neutral element~$0$.
		The associativity of~$\oplus$\dash---and with it the property of $((\Rdist)^P, \oplus)$ being a semigroup\dash---follows from the associativity of $\min$:
		\begin{equation}
			((x \oplus y) \oplus z)_\pi
				= \min\{ \min\{ x_\pi, y_\pi \}, z_\pi \}
				= \min\{ x_\pi, \min\{ y_\pi, z_\pi \} \}
				= (x \oplus (y \oplus z))_\pi.
		\end{equation}
		Since $\min$ is commutative, $\oplus$~is too and it is easy to check that $(x \oplus 0)_\pi = (0 \oplus x)_\pi = x_\pi$.

	\item
		To see that $((\Rdist)^P, \odot)$ is a semigroup with neutral element~$1$, we first check that $\odot$ is associative, i.e., that it is a semigroup:
		\begin{align}
			((x \odot y) \odot z)_\pi
				&= \min\{ \min\{ x_{\pi^1} + y_{\pi^2} \mid \pi^{12} = \pi^1 \circ \pi^2 \} + z_{\pi^3} \mid \pi = \pi^{12} \circ \pi^3 \} \\
				&= \min\{ (x_{\pi^1} + y_{\pi^2}) + z_{\pi^3} \mid \pi = (\pi^1 \circ \pi^2) \circ \pi^3 \} \\
				&= \min\{ x_{\pi^1} + (y_{\pi^2} + z_{\pi^3}) \mid \pi = \pi^1 \circ (\pi^2 \circ \pi^3) \} \\
				&= (x \odot (y \odot z))_\pi.
		\end{align}
		Furthermore, $(1 \odot x)_\pi = \min\{ 0 + x_\pi \} = x_\pi = (x \odot 1)_\pi$, hence $1$ is the neutral element w.r.t.~$\odot$.

	\item
		Regarding the distributive laws, we begin with the left-distributive law~\eqref{eq:dist-left}:
		\begin{align}
			(x \odot (y \oplus z))_\pi
				&= \min\{ x_{\pi^1} + \min\{ y_{\pi^2}, z_{\pi^2} \} \mid \pi = \pi^1 \circ \pi^2 \} \\
				&= \min\{ \min\{ x_{\pi^1} + y_{\pi^2}, x_{\pi^1} + z_{\pi^2} \} \mid \pi = \pi^1 \circ \pi^2 \} \\
				&= \min\{ \min\{ x_{\pi^1} + y_{\pi^2} \mid \pi = \pi^1 \circ \pi^2 \},
					\min\{ x_{\pi^1} + z_{\pi^2} \mid \pi = \pi^1 \circ \pi^2 \} \} \\
				&= ((x \odot y) \oplus (x \odot z))_\pi.
		\end{align}
		Regarding the right-distributive law~\eqref{eq:dist-right}, we obtain:
		\begin{align}
			((y \oplus z) \odot x)_\pi
				&= \min\{ \min\{ y_{\pi^1}, z_{\pi^1} \} + x_{\pi^2} \mid \pi = \pi^1 \circ \pi^2 \} \\
				&= \min\{ \min\{ y_{\pi^1} + x_{\pi^2}, z_{\pi^1} + x_{\pi^2} \} \mid \pi = \pi^1 \circ \pi^2 \} \\
				&= \min\{ \min\{ y_{\pi^1} + x_{\pi^2} \mid \pi = \pi^1 \circ \pi^2 \},
					\min\{ z_{\pi^1} + x_{\pi^2} \mid \pi = \pi^1 \circ \pi^2 \} \} \\
				&= ((y \odot x) \oplus (z \odot x))_\pi.
		\end{align}

	\item
		It remains to check that $0$ is an annihilator for~$\odot$.
		We have
		\begin{equation}
			(0 \odot x)_\pi
				= \min\{ 0_{\pi^1} + x_{\pi^2} \mid \pi = \pi^1 \circ \pi^2 \}
				= \min \emptyset
				= \infty
				= 0_\pi
		\end{equation}
		and, equivalently, $(x \odot 0)_\pi = 0_\pi$.
	\end{enumerate}
	Hence, $\mathcal{P}_{\min,+}$~is a semiring as claimed.
\end{proof}

\paragraph*{Proof of Lemma~\ref{lem:allpaths-xh}}

\begin{proof}
	We prove the claim by induction.
	By Equation~\eqref{eq:allpaths-x0}, the claim holds for $h = 0$.
	As induction hypothesis, suppose the claim holds for all $0 \leq h' \leq h$.
	The induction step yields
	\begin{equation}\label{eq:allpaths-proof}
		x^{(h+1)}_v
			\stackrel{\eqref{eq:allpaths-xh}}{=} \left( A x^{(h)} \right)_v
			= \bigoplus_{w \in V} a_{vw} x^{(h)}_w
			\stackrel{\eqref{eq:allpaths-adjacencymatrix}}{=} \underbrace{a_{vv}}_{1} x^{(h)}_v \oplus \bigoplus_{\{v,w\} \in E} a_{vw} x^{(h)}_w.
	\end{equation}
	We have $a_{vv} x^{(h)}_v = 1 x^{(h)}_v = x^{(h)}_v$ by construction, i.e., $a_{vv} x^{(h)}_v$ contains exactly the properly weighted $h$-hop paths beginning at $v$ by the induction hypothesis.
	Next, consider $\{v,w\} \in E$.
	By induction, $x^{(h)}_w$ contains exactly the $h$-hop paths beginning in $w$ and $a_{vw}$ contains only the edge $\{v,w\}$ of weight $\weight(v,w)$ by Equation~\eqref{eq:allpaths-adjacencymatrix}.
	Hence, $a_{vw} x^{(h)}$ contains all $(h+1)$-hop paths beginning with $\{v,w\}$.
	Due to Equation~\eqref{eq:allpaths-proof} and
	\begin{equation}
		\paths^{h+1}(v, \cdot, G)
			= \paths^h(v, \cdot, G) ~\cup
				\bigcup_{\{v,w\} \in E} \left\{ (v,w) \circ \pi \mid \pi \in \paths^h(w, \cdot, G) \right\},
	\end{equation}
	$x^{(h+1)}_v$~contains exactly the properly weighted $(h+1)$-hop paths, as claimed.
\end{proof}

\paragraph*{Proof of Lemma~\ref{lem:ksdp-r}}

\begin{proof}
	%
	%
	%
	Clearly, $r$~is a projection.
	We show in one step each that it fulfills Conditions~\eqref{eq:congruence-by-r-product} and~\eqref{eq:congruence-by-r-sum} of Lemma~\ref{lem:congruence-by-r}.
	Throughout the proof, let $x,x',y,y' \in \mathcal{P}_{\min,+}$ be such that $x \sim x'$ and $y \sim y'$.
	\begin{enumerate}
	\item
		To see that $r$ fulfills~\eqref{eq:congruence-by-r-product}, suppose for contradiction and w.l.o.g.\ that $r(yx)_\pi < r(yx')_\pi$ for some $v$-$s$-path~$\pi$.
		By definition, we have $r(yx)_\pi = y_{\pi^1} + x_{\pi^2}$ for some partition $\pi = \pi^1 \circ \pi^2$.
		Suppose that $\pi^1$ is a $v$-$w$-path and $\pi^2$ a $w$-$s$-path.
		Furthermore, $r(yx)_\pi < \infty$, i.e., $\pi \in P_k(v,s,yx)$, by assumption.

		Observe that $\pi^2 \in P_k(w,s,x)$, otherwise $\pi \notin P_k(v,s,yx)$.
		Because $x \sim x'$, it holds that $P_k(w,s,x') = P_k(w,s,x)$ with $x_{\pi'} = x'_{\pi'}$ for any $\pi' \in P_k(w,s,x')$.
		In particular, $\pi^2 \in P_k(w,s,x')$ and hence $\pi \in P_k(v,s,yx')$, where $(yx')_{\pi} = (yx)_{\pi}$.
		In other words, $r(yx)_\pi = r(yx')_\pi$, contradicting the assumption that $r(yx)_{\pi} < r(yx')_{\pi}$.

	\item
		We show that $r$ fulfills~\eqref{eq:congruence-by-r-sum} by contradiction;
		assume w.l.o.g.\ that $r(x \oplus y)_\pi < r(x' \oplus y')_\pi$ for a $v$-$s$-path~$\pi$.
		This implies $r(x \oplus y)_\pi < \infty$, i.e., $\pi \in P_k(v, s, r(x \oplus y))$.
		By definition, $r(x \oplus y)_\pi = \min\{ x_\pi, y_\pi \} < \infty$.
		Assume w.l.o.g.\ that $\min\{x_\pi, y_\pi\} = x_\pi$, so in particular $\pi \in P_k(v, s, x)$.
		As $x \sim x'$, $\pi \in P_k(v, s, x') = P_k(v,s,x)$ and $x'_\pi = x_\pi$.
		Hence,
		\begin{equation}
			(x' \oplus y')_\pi
				= \min\{ x'_\pi, y'_\pi \}
				\leq x'_\pi
				= x_\pi
				= r(x \oplus y)_\pi
				< r(x' \oplus y')_\pi,
		\end{equation}
		implying that $(x' \oplus y')_\pi \neq r(x' \oplus y')_\pi = \infty$.
		Using that $P_k(v, s, x' \oplus y') \subseteq P_k(v, s, x') \cup P_k(v, s, y')$, we see that this means that, together, $r(x')$ and $r(y')$ must contain at least $k$ distinct paths $\pi'$ such that $(r(x')_{\pi'}, \pi') < (r(x' \oplus y')_{\pi}, \pi)$ or $(r(y')_{\pi'}, \pi') < (r(x'\oplus y')_{\pi}, \pi)$.
		Since $x \sim x'$ and $y \sim y'$, for all such $\pi'$ we have that
		\begin{equation}
			((r(x) \oplus r(y))_{\pi'}, \pi')
				= ((r(x') \oplus r(y'))_{\pi'}, \pi')
				< (r(x' \oplus y')_{\pi}, \pi).
		\end{equation}
		This contradicts $\pi \in P_k(v, s, r(x \oplus y))$.
	\end{enumerate}
	Since $x \sim x'$, $y \sim y'$, and $\pi$ are arbitrary, the claim follows.
\end{proof}

\paragraph*{Chernoff's Bound}

We use a variant of Chernoff's bound regarding the sum of $0$--$1$ random variables $X_1, \dots, X_n$ that imposes weaker assumptions regarding the independence of the individual variables:
Instead of the standard assumption that all $\{X_1, \dots, X_n\}$ are independent, it suffices to require each $X_i$ to be independent of $\{X_1, \dots, X_{i-1}\}$.
This bound can be derived using well-known techniques\dash---we adapt the derivation of Mitzenmacher and Upfal~\cite{mu-pc-05}\dash---which we present for the sake of self-containment.

\begin{lemma}[Chernoff's Bound]\label{lem:chernoff}
	Let $X_1, \dots, X_n$ be $0$--$1$ random variables such that for all $2 \leq i \leq n$, $X_i$~is independent of $\{X_1, \dots, X_{i-1}\}$.
	Then for $X := \sum_{i=1}^n X_i$ and all $\delta \in \R_{>0}$ it holds that
	\begin{equation}\label{eq:chernoff}
		\Prob[X \geq (1 + \delta) \E[X]]
			\quad\leq\quad
			\left( \frac{e^\delta}{(1+\delta)^{(1+\delta)}} \right)^{\E[X]}.
	\end{equation}
\end{lemma}

\begin{proof}
	First consider random variables $Y_1, \dots, Y_k$ such that for all $2 \leq i \leq k$, $Y_i$~is independent of $\{Y_1, \dots, Y_{i-1}\}$.
	We claim that under these circumstances, we have
	\begin{equation}\label{eq:chernoff-product}
		\E\left[ \prod_{i=1}^k Y_i \right] = \prod_{i=1}^k \E[Y_i].
	\end{equation}
	For $k = 1$, \eqref{eq:chernoff-product}~trivially holds.
	As induction hypothesis, suppose that~\eqref{eq:chernoff-product} holds for some $k \in \N$ and define $Y := \prod_{i=1}^k Y_i$.
	Then $Y_{k+1}$ is independent from $Y$ by assumption and, using the induction hypothesis, we obtain~\eqref{eq:chernoff-product}:
	\begin{equation}
		\E\left[ \prod_{i=1}^{k+1} Y_i \right]
			= \E[Y \cdot Y_{k+1}]
			= \E[Y] \cdot \E[Y_{k+1}]
			= \prod_{i=1}^{k+1} \E[Y_i].
	\end{equation}

	Since $X$ is non-negative, we may apply Markov's bound and obtain, for arbitrary $t, \delta \in \R_{>0}$,
	\begin{equation}\label{eq:chernoff-markov}
		\Prob[X \geq (1 + \delta) \E[X]]
			= \Prob\left[ e^{tX} \geq e^{t (1 + \delta) \E[X]} \right]
			\leq \frac{\E[e^{tX}]}{e^{t (1 + \delta) \E[X]}}.
	\end{equation}
	Defining $Y_i := e^{t X_i}$ and scrutinizing $\E[e^{tX}]$ yields
	\begin{align}
		\E\left[ e^{tX} \right]
			&= \E\left[ \prod_{i=1}^n e^{t X_i} \right]
			\stackrel{\eqref{eq:chernoff-product}}{=} \prod_{i=1}^n \E\left[ e^{t X_i} \right]
			\stackrel{\Prob[X_i = 1] = \E[X_i]}{=}
				\prod_{i=1}^n \left( (e^t - 1) \E[X_i] + 1 \right) \\
			&\stackrel{1+x \leq e^x}{\leq}
				\prod_{i=1}^n e^{(e^t - 1) \E[X_i]}
			= e^{\sum_{i=1}^n (e^t - 1) \E[X_i]}
			= e^{(e^t - 1) \E[X]}. \label{eq:chernoff-exponent}
	\end{align}
	Combining~\eqref{eq:chernoff-markov} and~\eqref{eq:chernoff-exponent}, it follows that
	\begin{equation}\label{eq:chernoff-t}
		\Prob[X \geq (1 + \delta) \E[X]]
			\stackrel{\eqref{eq:chernoff-markov}}{\leq}
				\frac{\E[e^{tX}]}{e^{t(1 + \delta) \E[X]}}
			\stackrel{\eqref{eq:chernoff-exponent}}{\leq} \frac{e^{(e^t - 1) \E[X]}}{e^{t (1 + \delta) \E[X]}}
			= \left( \frac{e^{(e^t - 1)}}{e^{t (1 + \delta)}} \right)^{\E[X]}.
	\end{equation}
	Choosing $t := \ln(1 + \delta)$ in~\eqref{eq:chernoff-t} yields the claim.
\end{proof}

Mitzenmacher and Upfal~\cite{mu-pc-05} show that for $R \geq 6\E[X]$, it follows from~\eqref{eq:chernoff} that $\Prob[X \geq R] \leq 2^{-R}$.
In Lemma~\ref{lem:lists-short}, we have $\E[X] \in \bigO(\log n)$, i.e., that $\E[X] \leq c' \log_2 n$ for some $c' \in \R_{\geq 1}$.
Hence, for an arbitrary $c \in \R_{\geq 1}$, we can choose $R := 6cc' \log_2 n$ and obtain
\begin{equation}
	\Prob[X \geq R]
		\leq 2^{-R}
		\leq 2^{-6cc' \log_2 n}
		= n^{-6cc'}
		\leq n^{-c},
\end{equation}
i.e., $X \in \bigO(\log n)$ w.h.p.

\end{appendix}

\end{document}